\documentclass[ba, preprint]{imsart}
\pubyear{2024}
\volume{TBA}
\issue{TBA}
\firstpage{1}
\lastpage{1}

\usepackage{latexsym, amsmath, amssymb, amsthm, enumerate, tikz, graphicx, wrapfig, tabu, multirow, comment, bbm, bm, listings,subcaption, import, appendix}

\usepackage{amsthm}
\usepackage{amsmath}
\usepackage{natbib}
\usepackage[colorlinks,citecolor=blue,urlcolor=blue,filecolor=blue,backref=page]{hyperref}
\usepackage{graphicx}
\usepackage{siunitx,booktabs,caption,cellspace}
\usepackage{algorithm}
\usepackage{algpseudocode}

\usepackage{booktabs,caption,cellspace}
\usepackage{algorithm}
\usepackage{algpseudocode}
\usepackage{mathtools}
\numberwithin{equation}{section}

\newtheorem{theorem}{Theorem}

\theoremstyle{remark}
\newtheorem{remark}{Remark}
\theoremstyle{definition}

\newcommand{\PP}{\mathbb{P}}

\newcommand{\indep}{\perp \!\!\! \perp}

\usepackage{tikz}
\usetikzlibrary{bayesnet}
\usetikzlibrary{shapes,decorations,arrows,calc,arrows.meta,fit,positioning}

\tikzset{every node/.append style={font=\small}} 

\tikzstyle{latent} = [circle, draw, minimum size=.5em] 
\tikzstyle{obs} = [latent, fill=gray!50]
\tikzstyle{doublecircle} = [latent, double, double distance=0.8pt, outer sep=0.4pt]

\tikzset{
    -latex,auto,
    state/.style ={ellipse, draw, minimum width = 0.7cm},
    el/.style = {inner sep=1pt, align=left, sloped},
    point/.style = {circle, 
                    draw = black, 
                    inner sep = 0.01cm,
                    fill = black,
                    node contents={}},
    square/.style={rectangle, 
                   draw = black, 
                   inner sep = 0.02cm,
                   node distance = 0.5cm and 0.5cm,
                   fill=black,  
                   node contents={}},
    bidirected/.style={Latex-Latex,dashed}
}

\startlocaldefs
\endlocaldefs

\begin{document}

\begin{frontmatter}
\title{Causally Sound Priors for Binary Experiments}
\runtitle{Causally Sound Priors for Binary Experiments}

\begin{aug}
\author{\fnms{Nicholas J.} \snm{Irons}\thanksref{addr1}\ead[label=e1]{nicholas.irons@stats.ox.ac.uk}\ead[label=u1]{https://njirons.github.io}}
\and
\author{\fnms{Carlos} \snm{Cinelli}\thanksref{addr2}\ead[label=e2]{cinelli@uw.edu}\ead[label=u2]{http://carloscinelli.com}}

\runauthor{Nicholas J. Irons \& Carlos Cinelli}

\address[addr1]{Department of Statistics and Leverhulme Centre for Demographic Science, University of Oxford, Oxford, UK.\\ Email: \printead{e1}; \printead{u1}.
}
\address[addr2]{Assistant Professor, Department of Statistics, University of Washington, Seattle, USA.\\ Email: \printead{e2}; \printead{u2}.
}

\end{aug}

\begin{abstract}
We introduce the BREASE framework for the Bayesian analysis of randomized controlled trials with a binary treatment and a binary outcome. Approaching the problem from a causal inference perspective, we propose parameterizing the likelihood in terms of the \underline{\textbf{b}}aseline \underline{\textbf{r}}isk,  \underline{\textbf{e}}fficacy, and \underline{\textbf{a}}dverse \underline{\textbf{s}}ide \underline{\textbf{e}}ffects of the treatment, along with a flexible, yet intuitive and tractable jointly independent beta prior distribution on these parameters, which we show to be a generalization of the Dirichlet prior for the joint distribution of potential outcomes. Our approach has a number of desirable characteristics when compared to current mainstream alternatives:  
(i)~it naturally induces prior dependence between expected outcomes in the treatment and control groups;  (ii)~as the baseline risk, efficacy and risk of adverse side effects are quantities commonly present in the clinicians' vocabulary, the hyperparameters of the prior are directly interpretable, thus facilitating the elicitation of prior knowledge and sensitivity analysis; and (iii)~we provide analytical formulae for the marginal likelihood, Bayes factor, 
and other posterior quantities, 
as well as 
an exact posterior sampling algorithm and an accurate and fast data-augmented Gibbs sampler
in cases where traditional MCMC fails. Empirical examples demonstrate the utility of our methods for estimation, hypothesis testing, and sensitivity analysis of treatment effects. 
\end{abstract}


\end{frontmatter}

\section{Introduction}
\label{sec:intro}

Randomized controlled trials (RCTs) form the cornerstone of scientific research across numerous disciplines. In their most basic form, these trials compare the occurrence of an adverse (or favorable) outcome between treatment and control groups. This is particularly evident in a drug or vaccine trial, in which the efficacy of an intervention is established by comparing the number of individuals who die or develop a disease in each arm of the study. We refer to this type of study design as a ``binary experiment,'' wherein each participant is subjected to either a treatment or a control condition  (a binary exposure), and we observe either the presence or absence of the adverse effect of interest (a binary outcome).  

If participants of the trial are independent draws from a common (super-)population, statistical inference in binary experiments amounts to what is perhaps the simplest of tasks in statistics---the comparison of two binomial proportions. Indeed, from a Bayesian perspective, inference on the parameter of a binomial distribution dates back to at least as early as the origins of Bayesian inference itself, as evidenced by the seminal works of \citet{bayes1763essay} and \citet{laplace1774memoire}. The task comprises specifying a joint prior distribution for both binomial parameters, and computing the posterior distribution (or Bayes factors) of (relevant contrasts of) such parameters (e.g., the risk difference, or the risk ratio). 
Yet,  despite this long tradition, their widespread occurrence in the sciences, and the apparent simplicity of the inferential task, mainstream approaches for prior specification in the analysis of binary experiments 
have several shortcomings.

As reviewed in \citet{agresti2005-intervals} and \citet{dablander2022}, and also evident from perusing popular textbooks (e.g., \citealp{gelman1995bayesian,kruschke2014,mcelreath2020}),
the two predominant approaches for the Bayesian analysis of binary experiments consist of: (i) assigning independent beta priors to each of the binomial proportions, which are conjugate priors to the (also independent) binomials comprising the likelihood; and, (ii) what is essentially a logistic regression, i.e., applying a logit transformation to the binomial proportions, and assigning Gaussian priors to the average log odds and the log odds ratio. For all their popularity, these two approaches are unsatisfactory in several ways. For example, in the first case, the assumption of prior independence of the two proportions is often not credible---e.g., in most settings, one expects that learning about the mortality rate in the control group should inform our beliefs about the mortality rate in the treatment group. Moreover, while the logit approach addresses the problem of prior dependence, it does so at the sacrifice of clarity and interpretation---odds ratios are notoriously difficult to understand \citep{davies1998}, hindering the utility of this approach for prior elicitation and sensitivity analysis.

In this paper we demonstrate how causal logic can be used to address  these challenges. Approaching the problem from a causal inference perspective, we first propose parameterizing the likelihood in terms of three clinically meaningful counterfactual quantities: the baseline risk,  efficacy, and risk of adverse side effects  (BREASE) of the intervention. We then propose a flexible, yet intuitive and tractable jointly independent beta prior distribution on these parameters, which we show to be a generalization of the Dirichlet prior on the joint distribution of potential outcomes. Our approach has a number of desirable characteristics: (i)~it naturally induces prior dependence between the two binomial proportions of the treatment and control arms of the study; (ii)~as the baseline risk, efficacy and risk of adverse side effects are quantities familiar to clinicians, the hyperparameters of the prior are directly interpretable, thus facilitating the elicitation of prior knowledge and sensitivity analysis; and (iii)~we derive analytical formulae for the marginal likelihood, Bayes factor, and other posterior quantities,  as well as  an exact posterior sampling algorithm and an accurate and fast data-augmented Gibbs sampler in cases where traditional MCMC fails. 

\paragraph{Related literature.}  The literature on Bayesian causal inference is extensive---see \cite{li2023} for a recent review. Related to our setup are studies in the analysis of RCTs using a \emph{traditional} Dirichlet prior on response types, such as \cite{chickering1996} and \cite{imbens1997}, or studies using a uniform prior on the response type counts, such as \citet{ding2019}. The Dirichlet prior on response types is a special case of our proposal, and our analysis not only extends it, but also clarifies when and how its use can be desirable as a way to induce causally sound priors on the the two binomial proportions.  Our study also relates to a growing body of literature investigating sensitivity and prior specification in Bayesian causal inference and analysis of experiments. In a seminal paper, \citet{spiegelhalter1994bayesian} argued in favor of the Bayesian analysis of randomized trials with a focus on prior specification for normally distributed data. \citet{robins-wasserman} and \citet{linero2023nonparametric,linero2023prior} discuss the pitfalls of prior independence between the parameters governing the outcome and selection models that can yield inconsistent causal inference in high dimensional observational studies. In a similar vein, our analysis shows that---even in a low-dimensional experimental setting---causally-inspired priors encoding dependence between potential outcomes can lead to more sensible inferences than the traditional conjugate prior asserting their independence.

More generally, when framed in the language of potential outcomes, causal inference can be seen as a missing data problem. Thus, our analysis is most closely related to the literature on contingency tables with missing or incomplete observations on certain cell counts.  In fact, our proposed prior can be shown to induce a \emph{generalized} Dirichlet distribution on the joint distribution of potential outcomes. This distribution has been studied in the 1970s and 1980s \citep{antelman1972, kaufman1973, dickey1983,dickey1987}, though mostly in the context of survey sampling. Similar priors have also appeared in the analysis of diagnostic testing, such as in \citet{branscum2005estimation}. Perhaps due to the intractability of the integrals,  the difficulty in interpretation of the original generalized Dirichlet parameterization, and the missing connection to formal causal inference, this prior has received little to no attention in the analysis of binary experiments. Our analysis shows that the generalized Dirichlet distribution emerges naturally from the causal formulation of the problem, that the parameters of the distribution can be cast in intuitive clinical terms, and that statistical inference is manageable, with exact posterior sampling, efficient data-augmentation algorithms, as well as analytical formulae for Bayes~factors---all of which we derive in this paper.

\paragraph{Outline of the paper.} 
Section \ref{sec:review} introduces the statistical setup for the analysis of binary experiments and reviews existing methods for Bayesian inference in this setting.  Section~\ref{sec:methods} introduces our proposal.  It also derives key results for implementation, such as analytical formulae for the marginal likelihood, algorithms for posterior sampling, and an extension of the model accommodating covariates. Section~\ref{sec:results} demonstrates the utility of our method in three empirical examples. Section~\ref{sec:discussion} concludes the paper, and suggests possible extensions for future research. Code to replicate our analysis is available at \url{https://github.com/njirons/causally-sound}.

\section{Preliminaries}
\label{sec:review}

In this section we set notation, the statistical setup, and briefly review the two main approaches currently used for the Bayesian analysis of binary experiments---the independent beta and logit transformation approaches. We also briefly introduce the response type parameterization of the joint distribution of potential outcomes, which is an important stepping stone for understanding our proposal. 

\subsection{Potential outcomes}

Our analysis is situated within the potential outcomes framework of causal inference \citep{neyman1923,rubin1974estimating}. Let $N$ denote the total number of participants in the study,  $Z_i$ a binary treatment indicator and $Y_i$ a binary outcome indicator for subject $i\in\{1,\ldots,N\}$. We denote by $Y_i(z)$ the potential outcome of subject $i$ under the experimental condition $Z_i = z$, where $z = 0$ indicates the control and $z=1$ the treatment condition.  Under the standard consistency assumption, we have that the observed outcome of subject $i$ equals the potential outcome associated to the experimental condition that subject $i$ has actually received, i.e., $Y_i = Y_i(Z_i)$. 
Throughout the paper, we adopt the convention that $Y_i=1$ denotes an adverse outcome, such as death or the contraction of a disease. 
We take a super-population perspective, and assume that subjects are independent and identically distributed (i.i.d.) draws from a common population. 
We assume complete randomization, which implies ignorability of the treatment assignment, $\{Y_i(1), Y_i(0)\} \indep Z_i$.  

\subsection{Marginal parameterization}
\label{sec:marginal}

When subjects are independently drawn from a common super-population and the treatment is assigned at random, it follows that the observed \emph{counts} of adverse outcomes in each treatment arm, 
\[
y_0 = \sum_{i=1}^N Y_i(1-Z_i), \qquad y_1 = \sum_{i=1}^N Y_iZ_i,
\]
follow independent binomial distributions (see Supplement A Section \ref{app:covariates} for derivation):
\begin{align*}
y_0 &\sim \text{Binomial}(N_0, \theta_0) \quad\indep\quad
y_1 \sim \text{Binomial}(N_1, \theta_1),
\end{align*}
where here, 
$\theta_1=\mathbb{P}(Y_i(1)= 1)$, $N_1=\sum_{i}Z_i$  denote the probability of an adverse outcome and the sample size of the treatment group, and $\theta_0=\mathbb{P}(Y_i(0)= 1)$,  $N_0  = N-N_1$ are the analogous quantities for the control group. 
We refer to the probabilities $\theta_0$ and $\theta_1$ as the \textit{baseline risk} and \textit{risk of treatment}, respectively.

This defines the likelihood under the marginal parameterization of a binary experiment, so called because the parameters $(\theta_0,\theta_1)$ are defined in terms of the marginal distribution of the potential outcomes $Y_i(0)$ and $Y_i(1)$:
\begin{equation}
L(\mathcal{D}|\theta_0,\theta_1) = \binom{N_0}{y_0}\theta_0^{y_0}(1-\theta_0)^{N_0-y_0}\times \binom{N_1}{y_1}\theta_1^{y_1}(1-\theta_1)^{N_1-y_1},
\label{eq:marg-lik}
\end{equation}
where hereafter we denote the observed data by $\mathcal{D}=(y_0,y_1,N_0,N_1)$. 
To determine the effect of treatment, if any, 
Bayesian inference is carried out using the posterior distribution of the parameters $(\theta_0,\theta_1)$, which requires specification of a prior distribution for $(\theta_0,\theta_1)$. There are two main parameterizations with accompanying priors currently in use, discussed extensively in \citet{agresti2005-intervals} and \citet{dablander2022}. These are the independent beta (IB) and logit transformation (LT) approaches.

\subsubsection{Independent beta (IB) approach}
\label{sec:ib}

The independent beta (IB) approach  \citep{jeffreys1935} assigns the prior
\begin{equation}
\theta_0 \sim\text{Beta}(a_0,b_0) \quad\indep\quad \theta_1 \sim\text{Beta}(a_1,b_1),
\label{eq:ib}    
\end{equation}
for some hyperparameters $a_0,b_0,a_1,b_1>0$.  We refer to (\ref{eq:ib}) as the $\text{IB}(a;b)$ prior, where $a=(a_0,a_1),b=(b_0,b_1)$. A common \emph{default} specification is $a_0=b_0=a_1=b_1=1$, which assigns a uniform distribution to $(\theta_0,\theta_1)$. This choice of flat priors is usually thought to encode ignorance of $(\theta_0,\theta_1)$ \textit{a priori}, though it makes strong implicit assumptions as we discuss next. 

The main advantage of the IB approach is its simplicity. As the beta prior is conjugate to the binomial likelihood, 
estimation and posterior simulation can be carried out exactly without resorting to approximate sampling algorithms, such as MCMC. Furthermore, marginal likelihoods and Bayes factors, which are widely used for Bayesian hypothesis testing and can be difficult to calculate in general (usually requiring numerical approximation or estimation via posterior simulation), can be calculated analytically \citep{kass1995}.

A drawback of the IB approach is the restrictive assumption of independence between $\theta_0$ and $\theta_1$. In most experimental settings, we would expect our knowledge about the risks in the control and treatment groups to be dependent. 
For example, if we know that the population prevalence of an infectious disease is approximately 1\%, we would expect the prevalence of the disease among those receiving a vaccine to be concentrated around 1\% or below,  reflecting the common prior belief that it is unlikely that the vaccine would cause the disease. The IB prior fails to accommodate this natural dependence between risks in each arm of the trial. Furthermore, since independence in the prior and the likelihood implies independence \textit{a posteriori}, this failure also extends to the posterior.  

\subsubsection{Logit Transformation (LT) approach}
\label{sec:lt}

The logit transformation (LT) approach \citep{kass1992,agresti2005,dablander2022} reparameterizes the model in terms of the logit-transformed risks, by defining the parameters $(\beta,\psi)$ satisfying
\begin{align*}
\log\left(\frac{\theta_0}{1-\theta_0}\right) = \beta-\frac{\psi}{2}, \qquad
\log\left(\frac{\theta_1}{1-\theta_1}\right) = \beta+\frac{\psi}{2}.
\end{align*}
Note this parameterization is equivalent to a logistic regression of the outcome on the treatment with the encoding $Z \in \{-1/2, 1/2\}$ \citep{gronau2021}. It then assigns an independent normal prior to $(\beta,\psi)$:
\begin{equation}
\beta \sim\text{Normal}(\mu_{\beta},\sigma_\beta^2) \quad\indep\quad \psi \sim\text{Normal}(\mu_{\psi},\sigma_\psi^2),
\label{eq:lt}    
\end{equation}
where $\mu = (\mu_{\beta}, \mu_{\psi})$ and $\sigma=(\sigma_\beta,\sigma_\psi)>0$ are hyperparameters. A common default choice is $\mu=(0,0)$ and $\sigma=(1,1)$. We refer to (\ref{eq:lt}) as the $\text{LT}(\mu; \sigma)$ prior. This prior encodes correlation between $\theta_0$ and $\theta_1$ through their shared dependence on $\beta$ and $\psi$. Figure~\ref{fig:pgm}  depicts probabilistic graphical models comparing the IB and LT parameterizations, as well as the other approaches we will later discuss. 

\begin{figure}[!t]
  \begin{subfigure}[t]{0.25\linewidth}
  \centering
          \begin{tikzpicture}[node distance = .7cm]
            \node[latent] (p01) {$\theta_0$};
            \node[obs, below = of p01] (n0) {$y_0$};
            \node[latent, right = 1cm of p01] (p11) {$\theta_1$};
            \node[obs, below = of p11] (n1) {$y_1$};
            \path (p01) edge  (n0);
            \path (p11) edge  (n1);
          \end{tikzpicture}   
          \caption{Independent Beta} \label{fig:ib}
  \end{subfigure}%
  \begin{subfigure}[t]{0.25\linewidth}
  \centering
          \begin{tikzpicture}[node distance = .7cm]
            \node[doublecircle] (p01) {$\theta_0$};
            \node[obs, below = of p01] (n0) {$y_0$};
            \node[doublecircle, right = 1cm of p01] (p11) {$\theta_1$};
            \node[obs, below = of p11] (n1) {$y_1$};
            \node[latent, above  = 1cm of p01] (a) {$\beta$};
            \node[latent, above  = 1cm  of p11] (b) {$\psi$};
            \path (p01) edge  (n0);
            \path (p11) edge  (n1);
            \path (a) edge  (p01);
            \path (a) edge  (p11);
            \path (b) edge  (p01);
            \path (b) edge  (p11);
          \end{tikzpicture}   
          \caption{Logit Transform} \label{fig:lt}
  \end{subfigure}%
  \begin{subfigure}[t]{0.25\linewidth}
  \centering
          \begin{tikzpicture}[node distance = .7cm]
            \node[doublecircle] (p01) {$\theta_0$};
            \node[obs, below = of p01] (n0) {$y_0$};
            \node[doublecircle, right = 1cm of p01] (p11) {$\theta_1$};
            \node[obs, below = of p11] (n1) {$y_1$};
            \node[latent, above = 1.25cm of $(p01)!0.5!(p11)$] (a) {$\mathbf{p}$};
            \path (p01) edge  (n0);
            \path (p11) edge  (n1);
            \path (a) edge  (p01);
            \path (a) edge  (p11);
          \end{tikzpicture}   
          \caption{Response Type} \label{fig:rt}
  \end{subfigure}%
  \begin{subfigure}[t]{0.25\linewidth}
  \centering
          \begin{tikzpicture}[node distance = .7cm]
            \node[latent] (p01) {$\theta_0$};
            \node[obs, below = of p01] (n0) {$y_0$};
            \node[doublecircle, right = 1cm of p01] (p11) {$\theta_1$};
            \node[obs, below = of p11] (n1) {$y_1$};
            \node[latent, above = 1cm of p11] (a) {$\eta_s$};
            \node[latent, above left = 1cm and 1cm of p11] (b) {$\eta_e$};
            \path (p01) edge  (n0);
            \path (p01) edge  (p11);
            \path (p11) edge  (n1);
            \path (a) edge  (p11);
            \path (b) edge  (p11);
          \end{tikzpicture}   
          \caption{BREASE} \label{fig:pns}
  \end{subfigure}
\caption{Probabilistic graphical models for different parameterizations and prior setups. Gray nodes denote observed variables, white nodes denote latent parameters, and double borders indicate that a node is a deterministic function of its parents. (a) Independent beta priors are placed directly on $\theta_0$ and $\theta_1$; 
(b) Independent Gaussian priors are placed on the log odds quantities $\beta$ and $\psi$; (c) A Dirichlet prior is placed on the response type probabilities $\mathbf{p}$; (d) Our proposal, independent beta priors are placed on $\theta_0$, $\eta_e$, and $\eta_s$.
}   
\label{fig:pgm}
\end{figure}
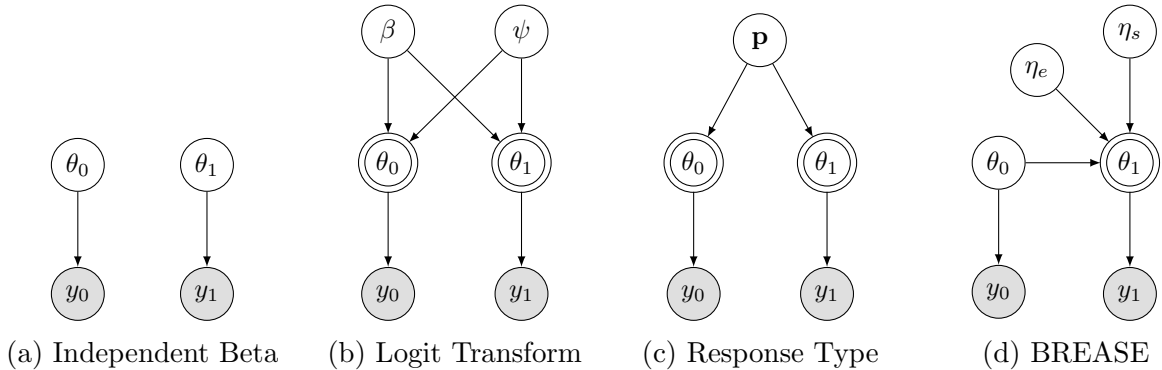

While the LT approach induces prior dependence between $\theta_0$ and $\theta_1$, this comes at the cost of a less intuitive parameterization.
Here $\beta$ is interpreted as the ``grand log odds,'' i.e, the average of the log odds across treatment arms, whereas $\psi$ is the log odds ratio. Odds ratios are notoriously difficult to understand, and thus reasoning about the prior means and variances of log odds---two unbounded hyperparameters---is often challenging in practice. 
The LT approach also has other computational disadvantages relative to the IB prior. Unlike the IB approach, marginal likelihoods and Bayes factors are not available analytically, and posterior sampling must be carried out approximately. 

\subsection{Response type (RT) parameterization}
\label{sec:strata}

The IB and LT approaches focus on the margins of the joint distribution of potential outcomes $(Y_i(0),Y_i(1))$. 
This focus is natural, as the observed data depends only upon the parameters $\theta_0$ and $\theta_1$. However, thinking in terms of their \emph{joint} distribution reveals alternative ways of inducing prior dependence between  these parameters. 
Specifically, the joint distribution of potential outcomes is fully characterized by four probabilities
\begin{equation}
p_{jk} = \mathbb{P}(Y_i(0)=j, Y_i(1)=k), \qquad j,k\in\{0,1\}.
\label{eq:rt}
\end{equation}
The probabilities $\mathbf{p}=\{p_{jk}\}_{j,k\in\{0,1\}}$ describe the frequencies of the four possible response types in the population \citep{copas1973,greenland1986}.
These include: (i)~the ``doomed'' $\{Y_i(0)=1,Y_i(1)=1\}$, for whom death occurs regardless of treatment; (ii)~the ``immune'' $\{Y_i(0)=0,Y_i(1)=0\}$, for whom death does not occur regardless of treatment; (iii)~the ``preventive'' $\{Y_i(0)=1,Y_i(1)=0\}$, for whom treatment \emph{prevents} death; and,  (iv)~the ``causal'' $\{Y_i(0)=0,Y_i(1)=1\}$, for whom treatment \emph{causes} death.  These probabilities are also sometimes referred to as ``probabilities of causation'' \citep{tian2000, pearl:2009}.  Here $\theta_0$ and $\theta_1$, which satisfy $\theta_0 = p_{10}+p_{11}$ and $\theta_1 = p_{01}+p_{11}$, define the margins of Table~\ref{tab:ps}.

\begin{table}[t]
\centering
\renewcommand{\arraystretch}{1.2}
\begin{tabular}{l|rr|r}
\toprule
            & ${Y_i(0)=0}$                      &  ${Y_i(0)=1}$                  &   Row Sum   \\ \cmidrule(lr){1-4} 
${Y_i(1)=0}$  & $p_{00}=(1-\eta_s)(1-\theta_0)$ & $p_{10}=\eta_e\theta_0$      & $1-\theta_1$ \\
${Y_i(1)=1}$  & $p_{01}=\eta_s(1-\theta_0)$     & $p_{11}=(1-\eta_e)\theta_0$  & $\theta_1$   \\ \cmidrule(lr){1-4} 
Column Sum  & $1-\theta_0$                    & $\theta_0$                   &              \\ 
\bottomrule
\end{tabular}
\caption{
$2\times 2$ contingency table of potential outcomes for a binary experiment. Only the margins of the table are identified from the observed data.
}
    \label{tab:ps}
\end{table}

Whereas in the marginal parameterization, independence of the likelihood and prior imply that estimation of $\theta_0$ is only informed by data in the control group (and similarly for $\theta_1$), the response type (RT) parameterization intertwines the data from each arm of the study. The shared dependence of $\theta_0$ and $\theta_1$ on the response type proportions reveals the link between outcomes in the control and treated groups. 

A Bayesian approach to modeling the  
response type probabilities $\mathbf{p}$ requires specification of a prior density supported on the probability simplex, making the Dirichlet distribution a natural candidate
\begin{equation}
\mathbf{p} = (p_{00},p_{10},p_{01},p_{11}) \sim \text{Dirichlet}(a_{00},a_{10},a_{01},a_{11}), \qquad a_{00},a_{10},a_{01},a_{11} > 0.
\label{eq:dirichlet}
\end{equation}
Indeed, priors of this type have been used 
in the analysis of partially identified quantities in randomized trials with non-compliance, such as in \citet{chickering1996};
see also \citet{imbens1997,madigan1999bayesian,hirano2000assessing}.
As we show next, the Dirichlet prior is a special case of our proposal, and our analysis not only extends it, but also clarifies its advantages and limitations as a means to induce the desired joint prior distribution on the two binomial proportions $(\theta_0, \theta_1)$. 

\section{The BREASE framework}
\label{sec:methods}

In this section we introduce the BREASE framework for the analysis of binary experiments. We start by parameterizing the likelihood in terms of the baseline risk,  efficacy, and  risk of adverse side effects of the treatment. We then propose a jointly independent beta prior distributions on these three parameters, which we show to be a generalization of the Dirichlet prior on the response types. Our proposal has a number of advantages.   From a statistical perspective, it induces dependence between the risks in the treatment and control groups, while also   enabling exact posterior sampling, and marginal likelihood calculations.  From a clinical perspective, this parameterization casts the model in terms of natural quantities appearing frequently in the clinician's vocabulary, thereby facilitating  interpretability, elicitation of prior knowledge, and sensitivity analyses. 

\subsection{Baseline risk, efficacy and adverse side effects}

To make things concrete, suppose $Y_i=1$ denotes death. We define the \textit{efficacy} of the treatment, $\eta_e$, as the probability that the treatment \emph{prevents} the death of a patient that would have otherwise died without it:
\begin{equation}
\eta_e = \mathbb{P}(Y_i(1)=0|Y_i(0)=1).
\label{eq:eta-e}
\end{equation}
Similarly, we define the risk of \textit{adverse side effects} of the treatment, $\eta_s$, as the probability that the treatment \emph{causes} the death of a patient that would have otherwise been healthy:
\begin{equation}
\eta_s = \mathbb{P}(Y_i(1)=1|Y_i(0)=0).
\label{eq:eta-s}
\end{equation}
Note that these are severe adverse side effects that  result in an outcome (e.g., death) opposite to the desired outcome of interest (i.e., survival). In the medical literature, this is sometimes called a ``paradoxical reaction'' \citep{smith2012paradoxical}. Such events could be the result not only of severe adverse biological reactions, but also of other forms of iatrogenesis, such as medical errors. 

These quantities can be interpreted as probabilities of sufficient causation \citep{tian2000, cinelli2021generalizing}, i.e., $\eta_e$ is the probability that treatment is sufficient to save or cure a patient,  while $\eta_s$ is the probability that treatment is sufficient to kill or hurt a patient. They correspond directly to the counterfactual interpretation of what clinicians colloquially refer to as ``efficacy'' and ``side effects'' of a drug or vaccine.  Indeed, a commonly used measure in clinical trials called ``efficacy'', defined as $1-\theta_1/\theta_0$, equals precisely $\eta_e$ under the assumption that treatment causes no harm ($\eta_s=0$).  

Applying the law of total probability, we can decompose the risk of treatment in terms of the baseline risk, efficacy, and risk of adverse side effects (BREASE), as 
\begin{equation}
\theta_1 = (1-\eta_e)\theta_0+\eta_s(1-\theta_0).
\label{eq:theta1}
\end{equation}
Table~\ref{tab:ps} shows how the response type probabilities $\mathbf{p}$ can be written as products of $\theta_0$, $\eta_s$, and $\eta_e$. As with the response type approach, this parameterization highlights the natural dependence between $\theta_0$ and $\theta_1$ that is easy to miss without framing the problem in the language of potential outcomes. For example, note that $\theta_0$ and $\theta_1$ are functionally independent only under the strong assumption that $\eta_e=1-\eta_s$, i.e., the probability of treatment saving a patient is equal to the probability that it does not kill one. 

\subsubsection{Likelihood}
\label{sec:lik}

Plugging in (\ref{eq:theta1}), we can rewrite the likelihood (\ref{eq:marg-lik}) in terms of $(\theta_0,\eta_e,\eta_s)$. 
\begin{theorem} 
\label{thm:lik}
Under (\ref{eq:marg-lik})  and (\ref{eq:eta-e})-(\ref{eq:theta1}), the likelihood is
\begin{align}
L (\mathcal{D}| \theta_0,\eta_e,\eta_s)
&= 
\binom{N_0}{y_0}\binom{N_1}{y_1}
\sum_{j=0}^{y_1}\sum_{k=0}^{N_1-y_1}\Bigg\{\binom{y_1}{j}\binom{N_1-y_1}{k} 
\theta_0^{y_0+j+k}(1-\theta_0)^{N-(y_0+j+k)} \nonumber\\
&\quad\times \eta_e^{k}(1-\eta_e)^{j} 
\eta_s^{y_1-j}(1-\eta_s)^{N_1-y_1-k}\Bigg\}, \qquad (\theta_0,\eta_e,\eta_s) \in [0,1]^3.
\label{eq:lik}
\end{align}
\end{theorem}
Theorem~\ref{thm:lik} follows from applying the binomial theorem twice. 
As the likelihood (\ref{eq:lik}) is polynomial in $(\theta_0,\eta_e,\eta_s)$, any prior distribution $\pi(\theta_0,\eta_e,\eta_s)$ for which the moments can be explicitly calculated yields an analytical expression for the marginal likelihood. In particular, if
\[
\pi(\theta_0,\eta_e,\eta_s) \propto \theta_0^{\alpha_0-1}(1-\theta_0)^{\beta_0-1}\times\eta_e^{\alpha_e-1}(1-\eta_e)^{\beta_e-1}\times\eta_s^{\alpha_s-1}(1-\eta_s)^{\beta_s-1}
\]
is a product of independent beta distributions, as we will see in the next section, 
then the marginal likelihood is a weighted sum of beta function values.
Furthermore, the posterior distribution $\pi(\theta_0,\eta_e,\eta_s|\mathcal{D})$ will be a mixture of independent beta distributions, from which we can sample exactly via simulation.

\subsubsection{Partial identification and monotonicity}
\label{sec:partial-mono}

The counterfactual parameters $\eta_e$ and $\eta_s$ are only partially identified by the observed data.  That is, in the limit of infinite data, even though $\theta_0$ and $\theta_1$ are point identified, \eqref{eq:theta1} defines a single equation with two unknowns, $\eta_e$ and $\eta_s$, which cannot be solved uniquely. Without further assumptions, we thus have the bounds
\begin{align}
    \max \left\{0, 1-\frac{\theta_1}{\theta_0}\right\} \leq \eta_e \leq \min\left\{\frac{1-\theta_1}{\theta_0}, 1\right\},  \quad  \max \left\{0, \frac{\theta_1-\theta_0}{1-\theta_0}\right\}\leq \eta_s \leq \min\left\{\frac{\theta_1}{1-\theta_0}, 1\right\}.\nonumber
\end{align}
As the sample size increases, the posterior distribution of $\eta_s$ and $\eta_e$ will not concentrate in a point---rather, it will remain spread over its partially identified region \citep{richardson2011, gustafson:2015}. Notice, however, that this does not affect the behavior of the posterior distribution of $(\theta_0,\theta_1)$.  The BREASE parameterization thus explicitly separates the identified and partially identified parameters---$(\theta_0,\theta_1)$ and $(\eta_e,\eta_s)$, respectively. Even if interest does not lie in the counterfactual probabilities $(\eta_s,\eta_e)$ \emph{per se}, assigning a prior to those quantities can be thought of as a causally principled way to specify a joint prior on the identified target parameters $(\theta_0, \theta_1)$. 

Finally, a common assumption in the potential outcomes literature is called \textit{monotonicity}, which states that the treatment does no harm. In our framework, this corresponds to the constraint $\eta_s=0$. This assumption may be reasonable in many clinical settings. Under monotonicity, the efficacy of the treatment is in fact point identified, and given by $\eta_e = 1-\theta_1/\theta_0$. The quantity $\theta_1/\theta_0$ is known as the risk ratio, and the quantity $1-\theta_1/\theta_0$ is indeed known as ``efficacy'' in the clinical trials literature.  In cases where side-effects are not expected to be exactly zero, but are expected to be small, the BREASE approach allows one to instead place an informative prior on $\eta_s$.

\subsection{Prior specification}
\label{sec:brease-prior}

Bayesian inference with the likelihood (\ref{eq:lik}) requires specifying a prior distribution on three separate and variation independent probabilities, i.e, $(\theta_0, \eta_e, \eta_s) \in [0,1]^3$ \citep{basu1977nuisance}.  We propose setting jointly independent beta prior distributions on these parameters:
\begin{equation}
\theta_0 \sim \text{Beta}^*(\mu_0,n_0) \quad\indep\quad \eta_e \sim \text{Beta}^*(\mu_e,n_e) \quad\indep\quad \eta_s \sim \text{Beta}^*(\mu_s,n_s),
\label{eq:gd-betas}
\end{equation}
where here $\text{Beta}^*(\mu,n)$ denotes a $\text{Beta}(a,b)$ distribution,
with mean $\mu = a/(a+b)$ and prior ``sample size'' $n = a+b$. We refer to (\ref{eq:gd-betas}) as the $\text{BREASE}(\bm{\mu};\bm{n})$ prior, where
$\bm{\mu} = (\mu_0,\mu_e,\mu_s)$, $\bm{n} = (n_0,n_e,n_s)$.

Since (\ref{eq:gd-betas}) defines a jointly independent beta prior on $(\theta_0,\eta_e,\eta_s)$, the discussion in Section \ref{sec:lik} applies. In particular, the posterior of $(\theta_0,\eta_e,\eta_s)$ is a mixture of independent betas, which permits exact sampling via simulation, 
and the marginal likelihood is available analytically as a weighted sum of beta functions, as we show in Sections~\ref{sec:posterior-sampling}~and~\ref{sec:gd-testing}.

\paragraph{Connections to the (generalized) Dirichlet.}

The prior (\ref{eq:gd-betas}) induces a \emph{generalized} Dirichlet distribution \citep{dickey1983,dickey1987,tian2003} on the vector of potential outcomes probabilities $\mathbf{p}$---see Supplement A Section~\ref{sec:gd} for derivation and further discussion.
In particular, the \emph{generalized} Dirichlet reduces to the \emph{traditional} Dirichlet distribution (\ref{eq:dirichlet}) for the following restricted choice of prior sample sizes
\begin{equation}
n_e = \mu_0 n_0, \quad n_s=(1-\mu_0)n_0.
\label{eq:equal-confidence}
\end{equation} 
Moreover, since $\theta_1=p_{01}+p_{11}$, by the aggregation property of the Dirichlet  \citep{ng2011dirichlet}, marginally we have
\begin{equation}
\theta_1\sim\text{Beta}^*\left((1-\mu_e)\mu_0+\mu_s(1-\mu_0),n_0\right),
\label{eq:theta1-dirichlet}
\end{equation}
which resembles the decomposition~(\ref{eq:theta1}). The BREASE approach thus reveals an implicit ``equal confidence'' assumption of the \emph{traditional} Dirichlet: the prior spread for $\theta_0$ determines the spread of the distributions of $\eta_e$, $\eta_s$, and $\theta_1$ \textit{a~priori}. Hence, the \emph{traditional}  Dirichlet is underparameterized, and unsuitable for cases in which, say, we have ample knowledge of the baseline risk but relatively little information about the possible efficacy or side effects of the treatment (or vice-versa), such as in clinical trials with historical control information \citep{schmidli2014robust}. Casting the likelihood in terms of the BREASE parameters makes such choices explicit, by allowing the  hyperparameters governing $\theta_0$, $\eta_e$ and $\eta_s$ to be set independently.

\subsubsection{Induced prior distribution of \texorpdfstring{$(\theta_0, \theta_1)$}{(theta0, theta1)}}
\label{sec:theta1-dist}

As mentioned in Section~\ref{sec:partial-mono}, our goal with the BREASE approach 
is primarily to induce causally sound priors on the identified parameters of interest, the two binomial proportions $(\theta_0$, $\theta_1)$. Thus we now discuss the induced marginal and conditional distribution of the risk of treatment, $\theta_1$, under the BREASE prior~(\ref{eq:gd-betas}). 

From equation (\ref{eq:theta1}) we see that $\theta_1$, conditionally on $\theta_0$,  is distributed as 
a convex combination of independent beta random variables \textit{a priori}. This distribution was studied in \citet{pham-gia1998} and is given in terms of Appell's first hypergeometric function $F_1$---in Supplement A Section~\ref{app:theta1-dist} we derive the explicit formula and provide further discussion. From here, the marginal prior on $\theta_1$ can be obtained as
$\pi(\theta_1) = \int_0^1 \pi(\theta_1|\theta_0)\pi(\theta_0)d\theta_0.$
While the general formula for $\pi(\theta_1|\theta_0)$  may look unwieldy, and 
the integration in  $\pi(\theta_1)$ 
prohibitive, 
there are noteworthy specific cases. 

\paragraph{Equal confidence.} As noted in the previous discussion, under the equal confidence assumption, $n_e = \mu_0 n_0$, $n_s=(1-\mu_0)n_0$, the marginal prior induced on $\theta_1$ is the beta distribution~in~(\ref{eq:theta1-dirichlet}). In particular, to obtain equal marginal priors for the treatment and control groups, i.e., $\theta_z\sim\text{Beta}^*(\mu_0,n_0)$ for $z\in\{0,1\}$, it suffices to set $\mu_s = (\mu_0/(1-\mu_0))\mu_e$, with $0\leq \mu_e \leq \min(1, (1-\mu_0)/\mu_0)$.  Choosing $\mu_0=1/2$, $n_0=2$, and $\mu_e=\mu_s=\mu$ 
results in marginal uniform priors with prior correlation $\text{Cor}(\theta_0,\theta_1)=1-2\mu$.

\paragraph{Monotonicity.} Under the ``no harm'' monotonicity assumption, $\eta_s=0$, we have $\theta_1=(1-\eta_e)\theta_0$, in which case $\theta_1$ is a product of independent beta random variables \textit{a priori}. 
\citet{springer1970} derived the form of this distribution, with the density given as a Meijer $G$-function.
In particular, if 
$n_e = \mu_0n_0$, we can show that $\theta_1 \sim \text{Beta}((1-\mu_e)n_e,\mu_en_e+(1-\mu_0)n_0).$
For another example, if $(\theta_0,\eta_e)\sim\text{Uniform}(0,1)^2$, we have $\pi(\theta_1) = -\log\theta_1.$
Regarding the conditional prior $\pi(\theta_1|\theta_0)$ under the ``no harm'' assumption, it is clearly a scaled beta distribution, since $\theta_1=(1-\eta_e)\theta_0$. 
If $\eta_e\sim\text{Uniform}(0,1)$, we have $\theta_1|\theta_0\sim\text{Uniform}(0,\theta_0)$. Similarly, under the ``no benefit'' assumption $\eta_e=0$, we have that $\theta_1=\theta_0+\eta_s(1-\theta_0)$, which is a scaled and shifted beta random variable conditional on $\theta_0$. If $\eta_s\sim\text{Uniform}(0,1)$, then $\theta_1|\theta_0\sim\text{Uniform}(\theta_0,1)$. 

\paragraph{Moments.} The joint density $\pi(\theta_0,\theta_1)$ induced by the BREASE$(\mu;n)$ prior is generally complicated, but its moments are easily computed in terms of the hyperparameters $(\mu,n)$ as $\theta_1$ is a polynomial in $(\theta_0,\eta_e,\eta_s)$, which are beta distributed \textit{a priori}. For example, 
the prior covariance has a simple form, $\text{Cov}(\theta_0,\theta_1) = \frac{\mu_0(1-\mu_0)}{n_0+1}(1-\mu_e-\mu_s)$.
This implies the following directions of the prior correlation, 
\begin{equation}
\arraycolsep=1.5pt\def\arraystretch{.75}
\text{Cor}(\theta_0,\theta_1) 
\left\{ 
\begin{array}{ll}
<0, & \qquad\mu_e+\mu_s > 1, \\
=0, & \qquad\mu_e+\mu_s = 1, \\
>0, & \qquad\mu_e+\mu_s < 1.
\end{array}\right.
\label{eq:cor}
\end{equation}
In words, $\theta_0$ and $\theta_1$ are positively correlated \textit{a priori} when the expected harm and benefit of treatment are small, and negatively correlated otherwise. 

\paragraph{Default prior.} 
While we encourage the use of  informative priors, it is useful to have reasonable defaults to start the analysis. 
If we would like to put $\theta_0$ and $\theta_1$ on equal footing, the 
$\text{BREASE}(1/2,\mu,\mu;2, 1,1)$ is thus the natural choice, with the following properties: (i) puts flat uniform priors on $\theta_0$ and $\theta_1$ (as with the IB approach); (ii) induces prior correlation between parameters (as with the LT approach); (iii) assumes no effect of treatment, on average (as with the IB and LT approaches); and, (iv)  depends on a single, easily interpretable parameter $\mu$ denoting the expected  benefits (efficacy) or harm (side effects) of the treatment.
When $\mu > 1/2$, $\theta_1$ and $\theta_0$ become anti-correlated, and thus for most cases, $\mu~\leq~1/2$ is a reasonable choice. 
Our preferred specification uses $\mu=0.3$ as the default. As Figure~\ref{fig:density-pns} in Supplement A shows, this (weakly) encodes the expectation of moderate effects and concentrates mass on the diagonal $\theta_0=\theta_1$. This quality is useful in the context of Bayesian hypothesis testing. 
When testing a null hypothesis $H_0$ (e.g., no effect of treatment on average, $H_0:\theta_0=\theta_1$) nested within an alternative $H_1$, it is desirable for the prior under $H_1$ to concentrate mass around the null model \citep{jeffreys1961,gunel1974,casella2009}.

\subsection{Posterior sampling}
\label{sec:posterior-sampling}

\subsubsection{Exact sampling}

The posterior under (\ref{eq:gd-betas}) is given by the following mixture of independent betas
\begin{align}
\pi&(\theta_0,\eta_e,\eta_s|\mathcal{D})
\propto 
\sum_{j=0}^{y_1}\sum_{k=0}^{N_1-y_1}
\Bigg\{ 
\binom{y_1}{j}\binom{N_1-y_1}{k}
\theta_0^{y_0+j+k+\mu_0n_0}(1-\theta_0)^{N-(y_0+j+k)+(1-\mu_0)n_0} \nonumber\\
&\quad \times\eta_e^{k+\mu_en_e}(1-\eta_e)^{j+(1-\mu_e)n_e}
\eta_s^{y_1-j+\mu_sn_s}(1-\eta_s)^{N_1-y_1-k+(1-\mu_s)n_s}\Bigg\}.
\label{eq:post}
\end{align}
As with the prior, this posterior falls into the family of generalized Dirichlet distributions on the vector of potential outcomes probabilities $\mathbf{p}$.   
While some posterior quantities can be obtained analytically (see Supplement A Section~\ref{sec:post-quant}), 
working  with the posterior density can be cumbersome; we now describe how to sample exactly from the posterior via simulation.
See Supplement A Section~\ref{sec:sampling-proof} for a full derivation of Theorem \ref{thm:sampling}.

\begin{algorithm}[t]
\caption{BREASE posterior---exact sampling algorithm}
\label{algo:sampling}
\begin{algorithmic}
\item[]\textbf{Input:} Data $\mathcal{D}=(y_0,y_1,N_0,N_1)$, hyperparameters $(\mu_0,\mu_e,\mu_s,n_0,n_e,n_s)$, and desired number of posterior samples $T$.
\item[]\textbf{Iterate:} For sample $t\in\{1,\ldots,T\}$, 
\begin{enumerate}[(i)]
    \item Sample 
    $P_1\in\{0,\ldots,N_1-y_1\}$ 
    conditional on $\mathcal{D}$ 
    with probability, as per~(\ref{eq:mix-weights}),
    \[
    \pi(P_1|\mathcal{D}) =\sum_{C_1=0}^{y_1} \pi(C_1,P_1|\mathcal{D}).
    \]
    \item Sample 
    $C_1\in\{0,\ldots,y_1\}$ 
    conditional on 
    $(P_1,\mathcal{D})$
    with probability, as per~(\ref{eq:mix-weights}),
    \[
    \pi(C_1|P_1,\mathcal{D}) \propto \pi(C_1,P_1|\mathcal{D}).
    \]
    \item Sample $(\theta_0,\eta_e,\eta_s)$ 
    conditional on
    $(C_1,P_1,\mathcal{D})$ 
    from the 
    distribution (\ref{eq:mix-post}).
\end{enumerate}
\item[]\textbf{Output:} Posterior samples $\{(\theta_0^{(t)},\eta_e^{(t)},\eta_s^{(t)})\}_{t\in\{1,\ldots,T\}}$.
\end{algorithmic} 
\end{algorithm}

\begin{theorem}
Let $(\theta_0,\eta_e,\eta_s)$ be random variables drawn according to Algorithm \ref{algo:sampling}. Then $(\theta_0,\eta_e,\eta_s)$ are distributed according to the BREASE posterior (\ref{eq:post}).
\label{thm:sampling}
\end{theorem}

\begin{proof}[Sketch of proof] 
We define the counterfactual counts
\begin{align*}
C_1 &= \sum_{i=1}^N I(Z_i=1,Y_i(1)=1,Y_i(0)=0), \qquad
P_1 = \sum_{i=1}^N I(Z_i=1,Y_i(1)=0,Y_i(0)=1),
\end{align*}
which are unobserved quantities.
Here, $C_1$ is the number of ``causal'' subjects in the treatment group, i.e., those who died under treatment but would have survived if untreated. Similarly, $P_1$ is the number of ``preventive'' subjects in the treatment group, i.e., those who survived under treatment but would have died if untreated. 
The BREASE posterior can then be expressed as a mixture distribution:
\begin{align}
&\pi(\theta_0,\eta_e,\eta_s|\mathcal{D}) = 
\sum_{C_1=0}^{y_1}\sum_{P_1=0}^{N_1-y_1} \pi(\theta_0,\eta_e,\eta_s|C_1,P_1,\mathcal{D}) 
\times \pi(C_1,P_1|\mathcal{D}).
\label{eq:mix}    
\end{align}
Hence, we can sample from the posterior by first drawing from the distribution of unobserved counts
$(C_1,P_1)$ 
conditional on the observed data $\mathcal{D}$. This distribution has probability mass function
\begin{align}
\pi&(C_1,P_1|\mathcal{D})\propto 
\binom{y_1}{C_1}\binom{N_1-y_1}{P_1}
\text{B}(P_1+\mu_en_e,y_1-C_1+(1-\mu_e)n_e)
\nonumber\\
&\quad\times \text{B}(y_0+y_1-C_1+P_1+\mu_0n_0,N-(y_0+y_1-C_1+P_1)+(1-\mu_0)n_0) \nonumber\\
&\quad\times \text{B}(C_1+\mu_sn_s,N_1-y_1-P_1+(1-\mu_s)n_s).
\label{eq:mix-weights}
\end{align}
We then sample the parameters $(\theta_0,\eta_e,\eta_s)$, which have an independent beta distribution conditional on the augmented data 
$(C_1,P_1,\mathcal{D})$: 
\begin{align}
\pi& (\theta_0,\eta_e,\eta_s |C_1,P_1,\mathcal{D})= \text{Beta}(\eta_e; P_1+\mu_e n_e, y_1-C_1+(1-\mu_e)n_e) \nonumber\\
&\quad\times\text{Beta}(\theta_0;y_0+y_1-C_1+P_1+\mu_0n_0,N-(y_0+y_1-C_1+P_1)+(1-\mu_0)n_0) \nonumber\\
&\quad\times\text{Beta}(\eta_s; C_1+\mu_sn_s,N_1-y_1-P_1+(1-\mu_s)n_s).
\label{eq:mix-post}
\end{align}    
Note that this derivation of the distribution \eqref{eq:mix-weights} provides a counterfactual interpretation of the mixture weights that result from directly normalizing the kernels in \eqref{eq:post}. 
\end{proof}

\subsubsection{Data augmentation (DA) algorithm}

\begin{algorithm}[!ht]
\caption{BREASE posterior---data augmentation algorithm}
\label{algo:gibbs}
\begin{algorithmic}
\item[]\textbf{Input:} Data $\mathcal{D}=(y_0,y_1,N_0,N_1)$, hyperparameters $(\mu_0,\mu_e,\mu_s,n_0,n_e,n_s)$, desired number of posterior samples $T$, number of burn-in iterations $B$, and BREASE parameter initialization $(\theta_0^{(0)},\eta_e^{(0)},\eta_s^{(0)})\in(0,1)^3$.
\item[]\textbf{Iterate:} For sample $t\in\{1,\ldots,T\}$, 
\begin{enumerate}[(i)]
    \item Sample 
    $(C_1^{(t)},P_1^{(t)})$ 
    conditional on $(\theta_0^{(t-1)},\eta_e^{(t-1)},\eta_s^{(t-1)},\mathcal{D})$ from the independent binomial distributions
    \begin{align}
    \textstyle
        C_1^{(t)} \sim \text{Binomial}\left(y_1,\frac{(1-\theta_0^{(t-1)})\eta_s^{(t-1)}}{\theta_1^{(t-1)}}\right), 
        P_1^{(t)} \sim \text{Binomial}\left(N_1-y_1,\frac{\theta_0^{(t-1)}\eta_e^{(t-1)}}{1-\theta_1^{(t-1)}}\right), \nonumber
    \end{align}
    where $\theta_1^{(t-1)} = \theta_0^{(t-1)}(1-\eta_e^{(t-1)})+(1-\theta_0^{(t-1)})\eta_s^{(t-1)}$.
    \item Sample $(\theta_0^{(t)},\eta_e^{(t)},\eta_s^{(t)})$ conditional on 
    $(C_1^{(t)},P_1^{(t)},\mathcal{D})$
    from the independent beta distributions \eqref{eq:mix-post}.
\end{enumerate}
\item[]\textbf{Output:} Posterior samples after burn-in $\{(\theta_0^{(t)},\eta_e^{(t)},\eta_s^{(t)})\}_{t\in\{B+1,\ldots,T\}}$.
\end{algorithmic} 
\end{algorithm}

We now derive a Gibbs sampler targeting the BREASE posterior \eqref{eq:post} based on the data augmentation scheme introduced for Algorithm~\ref{algo:sampling}.
Algorithm \ref{algo:gibbs} defines the Gibbs sampler. It consists of two steps: (i) first, we sample the counterfactual counts 
$C_1$ and $P_1$ 
conditional on the BREASE parameters; and, (ii) we sample $\theta_0, \eta_e, \eta_s$ conditional on the augmented data.
In numerical experiments, we find that the algorithm converges to the BREASE posterior quickly, often mixing within a few hundred iterations, and the sampling is also quite fast. The conditional distribution
of the unobserved counts 
$(C_1,P_1)|(\theta_0,\eta_e,\eta_s,\mathcal{D})$ 
is derived in Supplement A Section~\ref{sec:sampling-proof}.

\subsubsection{Pathological sampling}
\label{sec:pathological}

To demonstrate the utility of our posterior sampling algorithms, 
we now turn to an example for which RJAGS \citep{rjags} and RStan \citep{rstan}, two popular MCMC software packages, 
fail to sample from the BREASE posterior.  
We use the data $y_0~=~20$,  $N_0~=~1000$, $y_1~=~40$,  $N_1~=~1000$, 
and the hyperparameters 
$\mu_0~=~0.5$,  $n_0~=~2$, $\mu_e~=~0.5$,  $n_e~=~2, \mu_s~=~0.01$, $n_s~=~1$.
The prior distributions for $\theta_0$ and $\eta_e$ are 
vague independent 
$\text{Uniform}(0,1)$ distributions.  
On the other hand, the prior on the risk of side effects $\eta_s$ is concentrated near 0 with mean $\mu_s=0.01$. This prior encodes a quasi-monotonicity assumption on the treatment that is clearly in conflict with the data.

\begin{figure}[t]
\centering
\includegraphics[scale=.38]{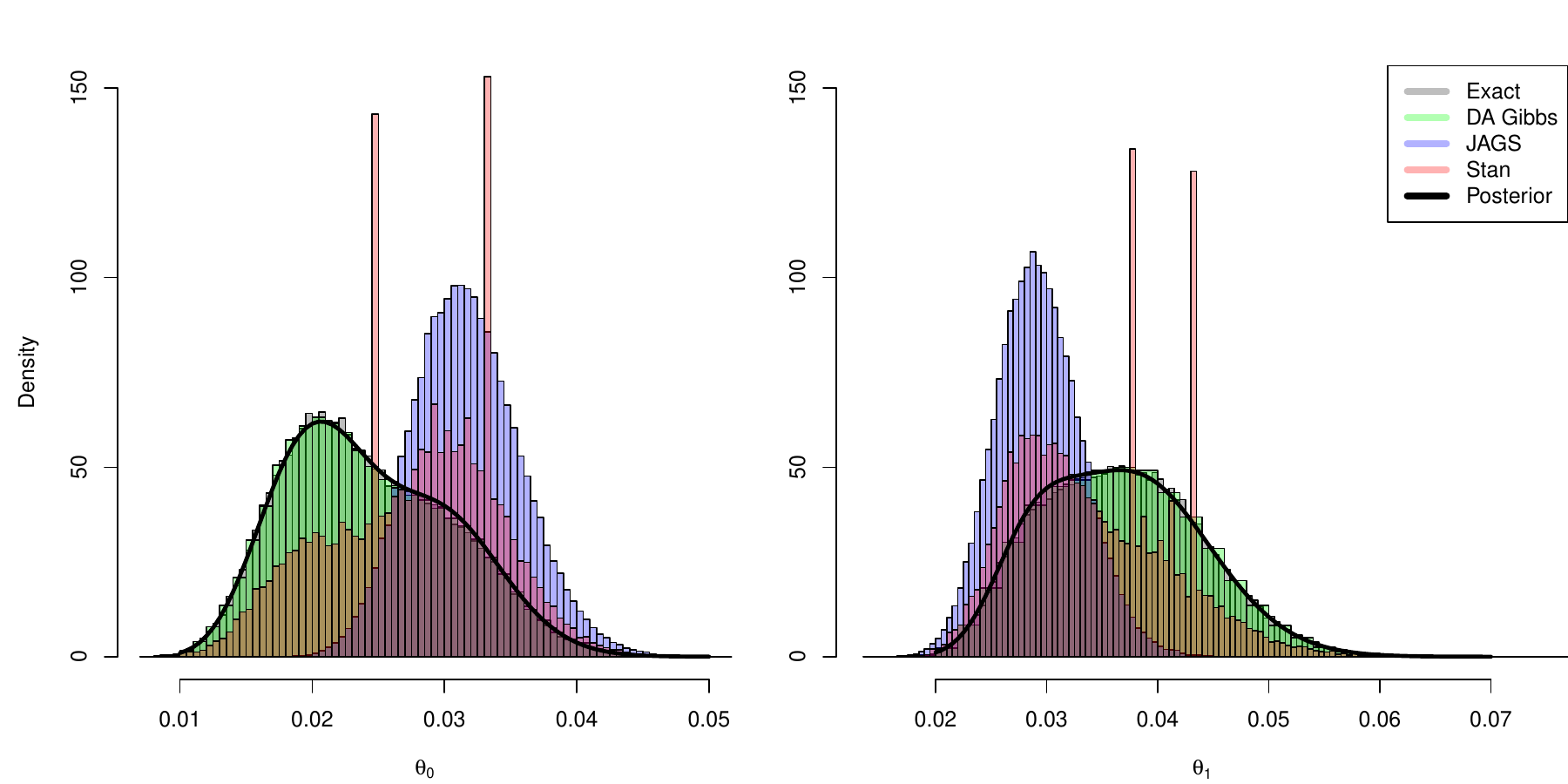}
\caption{Pathological MCMC posterior sampling exhibited in posterior histograms of the baseline risk $\theta_0$ (left) and treatment risk $\theta_1$ (right). The marginal posterior of $\theta_1$ (black curve) was approximated using numerical integration.}   
\label{fig:sampling}
\end{figure}

Prior-data conflict, which arises when the prior is concentrated on parameter values that are unlikely given the data, is a common culprit when diagnosing pathological MCMC sampling \citep{evans2006}. 
It is also a salient issue in the Bayesian analysis of clinical trials, particularly when historical information or clinical expertise are brought to bear on the design and analysis of the study \citep{schmidli2014robust}.
This example is no exception. 
Figure \ref{fig:sampling} shows histograms of 100,000 posterior samples of $\theta_0$ and $\theta_1$ drawn using Algorithm \ref{algo:sampling} (grey), 
Algorithm \ref{algo:gibbs} (green), 
JAGS (blue), and Stan (red). The marginal posterior density is plotted in black for reference. 
The posterior of $\theta_0$, which is a mixture of beta distributions, is exhibited in the left panel of Figure \ref{fig:sampling}. While Algorithms \ref{algo:sampling} and \ref{algo:gibbs}
produce posterior samples that fully capture the distribution, JAGS and Stan fail to adequately explore the left half of the distribution.
Although Stan 
manages to deviate from the right half
as compared to JAGS, 
its chains get stuck at $\theta_0\approx 0.024$ and $\theta_0\approx 0.033$ 
when the sampler rejects numerous proposal draws. 
The story is much the same for~$\theta_1$.

This example demonstrates that it is useful to have bespoke algorithms that perform well, even in adversarial settings. In particular, the algorithms we provide here may prove useful for future extensions of the model, as we will later discuss. Nevertheless, we note that JAGS and Stan do work well for this model in most cases---indeed, this is a pathological example designed to be challenging. Furthermore, in the case of prior-data conflict (or more generally when a sampler is struggling), a reassessment of the prior may be warranted, perhaps in favor of a more robust approach \citep{schmidli2014robust}. In Section \ref{sec:pathological-supp} of Supplement A, we further investigate the numerical issues causing the sampling difficulties in JAGS and Stan and discuss solutions.

\paragraph{Monotonicity.} Posterior sampling under monotonicity constraints can be obtained with similar procedures. See Theorems~\ref{thm:sampling-noharm}-\ref{thm:sampling-nobenefit} of Supplement~A, Section~\ref{sec:sampling-appendix}.

\subsection{Marginal likelihoods and Bayes factors}
\label{sec:gd-testing}

From a Bayesian perspective, hypothesis testing is essentially a model comparison exercise \citep{jeffreys1961,dickey1970,kass1995}. Consider two competing hypotheses, $H_0$ and $H_1$.  For each  hypothesis $H_k$, $k \in \{0,1\}$, the Bayesian approach requires postulating a fully specified model $M_k$, with likelihood $L_k(\mathcal{D}|\theta)$ and prior $\pi_k(\theta)$, respecting the constraints of the hypothesis the model is intended to represent. 
Evidence in favor of $H_1$ relative to $H_0$ is then quantified using the Bayes factor $\text{BF}_{10}$, given by the ratio of the marginal likelihoods of the observed data under each model, $\text{BF}_{10} = L_1(\mathcal{D})/L_0(\mathcal{D})$, where $L_k(\mathcal{D}) = \int L_k(\mathcal{D}|\theta)\pi_k(\theta)d\theta$. Given prior model probabilities ${\mathbb{P}(M_0)}$, ${\mathbb{P}(M_1)}$, the posterior odds of $M_1$ and $M_0$ are then
${\mathbb{P}(M_1|\mathcal{D})}/{\mathbb{P}(M_0|\mathcal{D})} = \text{BF}_{10}\times{\mathbb{P}(M_1)}/{\mathbb{P}(M_0)}.$
In this section we show how to formulate such models instantiating a number of relevant statistical hypotheses with the BREASE approach, and provide analytical formulae for the marginal likelihoods. 
For all models considered here the 
likelihood is the same, so we focus the discussion on the formulation of the prior. 

Let us first consider testing the null hypothesis $H_0: \theta_1 = \theta_0$ against the alternative hypothesis $H_1:\theta_1\neq \theta_0$. For $H_1$, we propose using the unconstrained model $M_1$, with the BREASE prior in (\ref{eq:gd-betas}) and equation (\ref{eq:theta1}),
\begin{align}
M_1: (\theta_0, \eta_e, \eta_s) \sim \text{BREASE}(\mu; n), \qquad \theta_1 = (1-\eta_e)\theta_0+\eta_s(1-\theta_0).     \label{eq:m1}
\end{align}
As for the null hypothesis $H_0:\theta_1= \theta_0$,  we  instantiate it with the null model, 
\begin{align}
M_0: \theta_0 \sim \text{Beta}^*(\mu_0, n_0), \qquad \theta_1 = \theta_0.     \label{eq:m0}
\end{align}
One benefit of $M_0$ is that its prior is logically consistent with the marginal distribution of $\theta_0$ under $M_1$, both implying $\theta_0 \sim \text{Beta}^*(\mu_0, n_0)$ \emph{a priori}. Note that the prior (\ref{eq:m0}) emerges naturally from $M_1$ in at least two ways: (i) when postulating that the treatment does not work at all, by setting $\eta_s=\eta_s=0$; or, (ii) by noting that, if the treatment has no effect on average (i.e, the efficacy of the treatment precisely offsets its side effects), one can side-step thinking about $\eta_s$ and $\eta_e$ altogether. In both cases, we borrow the prior of $\theta_0$ from $M_1$, and simply set $\theta_1$ equal to $\theta_0$.
We discuss alternative prior formulations for $H_0$ in Supplement A Section~\ref{app:h0}.

Other relevant hypothesis one may wish to test are that the treatment is beneficial $H_-: \theta_1 < \theta_0$ or that the treatment is harmful $H_+: \theta_1 > \theta_0$, on average. 
A straightforward approach to specify models for such hypotheses is to note that $M_1$ already induces positive probabilities to the events postulated in $H_-$ and $H_+$. Thus, we can borrow this knowledge, already elicited when forming $M_1$, to define the priors $\pi_-$ and $\pi_+$,
\begin{align}
\pi_-(\theta_0, \eta_e, \eta_s) :=    \pi_1(\theta_0, \eta_e, \eta_s | \theta_1 < \theta_0), \quad
\pi_+(\theta_0, \eta_e, \eta_s) :=    \pi_1(\theta_0, \eta_e, \eta_s | \theta_1 > \theta_0). 
\label{eq:m-prior}
\end{align}
The priors $\pi_-$ and $\pi_+$ result in the models $M_-$ and $M_+$, for $H_-$ and $H_+$ respectively. Similarly to $M_0$, one benefit of these models is that the induced priors on $(\theta_0, \eta_e, \eta_s)$ are logically consistent with the beliefs expressed in $M_1$, under the constraints $H_-$ and $H_+$.  The same strategy employed here can be used for interval hypotheses of the type $H_0^\delta: |\theta_1 - \theta_0| \leq \delta$, with $\delta > 0$ (or, more generally, for any event with nonzero probability under $M_1$). Alternative models for $H_-$ and $H_+$, leveraging instead monotonicity constraints, such as $\eta_s=0$, are discussed in Supplement A Section~\ref{app:h-h+}.

In all cases above, 
the marginal likelihood can be obtained using analytical formulae and simple Monte Carlo approximation, facilitating the computation of Bayes factors.

\begin{theorem} The marginal likelihood of the data under $M_0$ is given by a  beta-binomial distribution. Under $M_1$, it is given by a weighted sum of beta functions:
\begin{align}
L_1(\mathcal{D}) &= 
\binom{N_0}{y_0}\binom{N_1}{y_1}
\sum_{j=0}^{y_1}\sum_{k=0}^{N_1-y_1}\binom{y_1}{j}\binom{N_1-y_1}{k}
\times \frac{\text{B}(k+\mu_en_e,j+(1-\mu_e)n_e)}{\text{B}(\mu_en_e,(1-\mu_e)n_e)}  \nonumber  \\
&\qquad\times \frac{\text{B}(y_0+j+k+\mu_0n_0,N-(y_0+j+k)+(1-\mu_0)n_0)}{\text{B}(\mu_0n_0,(1-\mu_0)n_0)}\nonumber \\
&\qquad\times \frac{\text{B}(y_1-j+\mu_sn_s,N_1-y_1-k+(1-\mu_s)n_s)}{\text{B}(\mu_sn_s,(1-\mu_s)n_s)}.
\label{eq:ml1}
\end{align}
Under $M_-$ and $M_+$, it can be obtained from $L_1(\mathcal{D})$ as follows,
\begin{align}
L_-(\mathcal{D}) = L_1(\mathcal{D})\times \frac{\pi_1(\theta_1 < \theta_0|\mathcal{D})}{\pi_1(\theta_1 < \theta_0)},       \qquad
L_+(\mathcal{D}) = L_1(\mathcal{D})\times \frac{\pi_1(\theta_1 > \theta_0|\mathcal{D})}{\pi_1(\theta_1 > \theta_0)}.      \label{eq:ml+}
\end{align}
\end{theorem}
\begin{proof}
    The result for $M_0$ is well-known. $ L_1(\mathcal{D})$ in (\ref{eq:ml1}) follows directly from integration of~(\ref{eq:lik}) under the prior~(\ref{eq:gd-betas}). $L_-(\mathcal{D})$ and $L_+(\mathcal{D})$ in (\ref{eq:ml+}) follow from Bayes'~rule. 
\end{proof}

\begin{remark}
The prior and posterior probabilities $\pi_1(\theta_1 < \theta_0)$ and $\pi_1(\theta_1 < \theta_0|\mathcal{D})$ 
can be  approximated using Monte Carlo integration with exact samples, as per Section~\ref{sec:posterior-sampling}.
\end{remark}

\begin{remark}
As per (\ref{eq:ml+}), if one postulates prior model probabilities  $\mathbb{P}(M_-|M_1) = \pi_1(\theta_1 < \theta_0)$ and $\mathbb{P}(M_+|M_1) = \pi_1(\theta_1 > \theta_0)$, the Bayes factor testing $H_0:\theta_1 = \theta_0$ against $H_1:\theta_1 \neq \theta_0$ (using $M_1$) conveniently decomposes into the weighted average of the Bayes factors testing  $H_0$ against $H_-$  (using $M_-$)  and $H_0$  against $H_+$ (using $M_+$)---though, of course, users can postulate prior probabilities for the models $M_-$ and $M_+$ as they wish.     
\end{remark}

\subsection{Extension to covariates}
\label{sec:covariates}

We conclude this section by demonstrating how the BREASE approach can be extended to accommodate discrete covariates. By extending the method in this way, we can address a number of important applications, which include: estimating conditional average treatment effects in randomized experiments; accounting for stratification in randomized experiments, or measured confounding in observational studies; and pooling evidence across multiple trials. We leave extensions to continuous covariates to future~work.

\subsubsection{Likelihood}

Suppose we observe i.i.d. samples $(Y_i,Z_i, X_i),~i \in \{1,\ldots,N\}$, where, as before, $Y_i$ and $Z_i$ denote the binary outcome and treatment indicators for subject $i$ and $X_i$ is a discrete pre-treatment covariate taking values in $\mathcal{X}$. We allow for the possibility of selection into treatment based on $X_i$. Hence, we now assume that randomization of the treatment holds only within strata of $X_i$ (also known as \emph{conditional} ignorability) \mbox{$Y_i(z)\indep Z_i | X_i$}. 

Let $y_{z,x}$ denote the observed death count and  $N_{z,x}$  the corresponding sample size for each stratum $x \in \mathcal{X}$ and study arm $z \in \{0, 1\}$. Further define the total count for stratum $x$ as $N_{x} = N_{0,x} + N_{1,x}$ and the total population size $N = \sum_{x \in \mathcal{X}} N_{x}$. We use boldface to indicate vectors, $\bm{N}=\{N_{z,x}\}_{z\in\{0, 1\}, x\in \mathcal{X}}$ and $\bm{y}=\{y_{z,x}\}_{z\in\{0, 1\}, x\in \mathcal{X}}$.  Finally, let $\mathcal{D}=(\bm{y}, \bm{N})$ denote the full data and $(\bm{\theta},\bm{\eta}, \bm{\delta},\bm{p_X})$ parameters, 
\[
\bm{\theta}=\{\theta_{z,x}\}_{z\in\{0, 1\}, x\in \mathcal{X}}, 
\quad \bm{\eta} = \{\eta_{e,x}, \eta_{s, x}\}_{x\in \mathcal{X}},
\quad \bm{\delta}=\{\delta_x\}_{x\in \mathcal{X}}, 
\quad \bm{p_X}=\{p_x\}_{x\in \mathcal{X}},
\]
where $\bm{\theta}$ and $\bm{\eta}$ collect the risks, efficacy and side effects for each stratum;   $\delta_{x}:= P(Z_i=1|X_i=x)$ denotes the propensity score for each stratum $x$; and $p_{x}:=P(X_i = x)$ denotes the marginal probability of $X_i = x$. 

The full likelihood is then given by (see Supplement A Section \ref{app:covariates} for derivation)
\begin{align*}
L(\mathcal{D}| \bm{\theta}_0, \bm{\eta}, \bm{\delta}, \bm{p_X}) &=
\prod_{x \in \mathcal{X}} \Bigg[
    \binom{N_{0,x}}{y_{0,x}} \binom{N_{1,x}}{y_{1,x}} 
    \sum_{j=0}^{y_{1,x}} \sum_{k=0}^{N_{1,x} - y_{1,x}} 
    \Bigg\{
        \binom{y_{1,x}}{j} \binom{N_{1,x}-y_{1,x}}{k}
        \theta_{0,x}^{y_{0,x} + j + k} \\
&\times (1 - \theta_{0,x})^{N_{x} - (y_{0,x} + j + k)}
        \eta_{e,x}^{k} (1 - \eta_{e,x})^{j} 
        \eta_{s,x}^{y_{1,x} - j} (1 - \eta_{s,x})^{N_{1,x} - y_{1,x} - k}
    \Bigg\}
\Bigg] \\
&\times \prod_{x \in \mathcal{X}} 
\binom{N_{x}}{N_{1,x}} 
\delta_x^{N_{1,x}} (1 - \delta_x)^{N_{0,x}} 
\times \frac{N!}{\prod_{x \in \mathcal{X}} N_{x}!} \prod_{x \in \mathcal{X}} p_x^{N_{x}}.
\end{align*}
The first component above corresponds to the BREASE likelihood  (\ref{eq:lik}) for each stratum $x \in \mathcal{X}$; the second component corresponds to the binomial likelihood for the treatment assignment, again for each stratum $x \in \mathcal{X}$; the final component is the marginal likelihood of $X$, which is a multinomial distribution.

\subsubsection{Priors and posterior sampling}
The likelihood factorizes into three independent components,  corresponding to the BREASE parameters $(\bm{\theta}, \bm{\eta})$, to the propensity score parameters $\bm{\delta}$, and finally to the parameters of the marginal distribution of the observed covariates $\bm{p_X}$. Thus, if the priors for these components are also mutually independent, this independence extends to the posterior, allowing the parameters of each component to be sampled independently. 
We make this assumption going forward in our discussion of prior specification. We propose two priors for $(\bm{\theta}, \bm{\eta})$: (i) an independent BREASE prior for each stratum $x \in \mathcal{X}$; and, (ii) a hierarchical prior that pools information across strata. 

\paragraph{Independent BREASE prior.} The simplest prior for this setup is to assign independent BREASE priors to the within-stratum parameters $(\theta_{0,x}, \eta_{e,x}, \eta_{s,x})$. Given that the strata are also independent in the likelihood, posterior samples can be drawn independently for each stratum using either the exact sampler (Algorithm \ref{algo:sampling}) or the data-augmented Gibbs sampler (Algorithm \ref{algo:gibbs}).

\paragraph{Hierarchical BREASE prior.} One drawback of independent priors is that they prevent information from being shared across strata. For example, learning about the efficacy of a vaccine in males would have no impact on our inferences about its efficacy in females. To overcome this, hierarchical priors can be introduced to partially pool information across different categories of $X_i$. This approach also supports meta-analyses across studies, with $X_i$ representing a study indicator.  A natural hierarchical prior would be
\begin{align*}
\theta_{0,x}\sim \text{Beta}^*(\mu_0,n_0), \quad &\mu_0\sim\text{Beta}^*(\lambda_0,\nu_0), \quad n_0\sim \text{Gamma}(\alpha_0,\beta_0),\\
\eta_{e,x}\sim \text{Beta}^*(\mu_e,n_e), \quad &\mu_e\sim\text{Beta}^*(\lambda_e,\nu_e), \quad n_e\sim \text{Gamma}(\alpha_e,\beta_e),\\
\eta_{s,x}\sim \text{Beta}^*(\mu_s,n_s), \quad &\mu_s\sim\text{Beta}^*(\lambda_s,\nu_s), \quad n_s\sim \text{Gamma}(\alpha_s,\beta_s).
\end{align*}
Hence, we specify a BREASE($\bm{\lambda},\bm{\nu}$) prior on the hierarchical BREASE parameters $(\mu_0,\mu_e,\mu_s)$ and Gamma priors on the random effects precision parameters $(n_0,n_e,n_s)$.  
Posterior sampling can proceed in two stages: (i) conditional on the hierarchical parameters, an independent BREASE update for the BREASE parameters; and (ii) conditional on the BREASE parameters, a Metropolis-Hastings update for the hierarchical parameters. We leave for future work the study of other priors and sampling algorithms.

\paragraph{Population effects.} The two procedures described above give us posterior samples of the within-stratum parameters $(\theta_{0,x}, \eta_{e,x}, \eta_{s,x})$, which  allow us to obtain posterior samples of conditional treatment effects, such as the conditional risk ratio, $\tau_x := \theta_{1,x}/\theta_{0,x}$, as well as any contrasts of such effects (e.g., $\tau_x - \tau_{x'}$). To recover population (marginal) effects, we need to average over the marginal distribution of $X$, e.g., $\theta_0 = \sum_{x \in \mathcal{X}} \theta_{0,x} p_x$ and $\theta_1 = \sum_{x \in \mathcal{X}} \theta_{1,x} p_x$. Since  $\bm{p_X}$ and $(\bm{\theta}, \bm{\eta})$ are independent \emph{a posteriori}, this averaging can be done at any point in the analysis by simply generating independent posterior samples of $\bm{p_X}$ (e.g, using a conjugate Dirichlet prior for $\bm{p_X}$).

\section{Empirical Examples}
\label{sec:results}

We now demonstrate the utility of our approach in three empirical examples. 
We show how the BREASE framework can be used to facilitate Bayesian estimation, hypothesis testing, and sensitivity analysis of the results of binary experiments. Concretely, the examples illustrate how our proposal can: (i) help analysts distinguish robust from fragile findings; (ii)~clarify what one needs to believe in order to claim that a treatment is effective; and (iii)~reconcile disparate results obtained from different methods.  See Supplement A Section \ref{sec:testing-ib-lt} for  details of the calculation of Bayes factors for the IB and LT approaches.

\subsection{The effect of aspirin on fatal myocardial infarction}
\label{sec:aspirin}

Cardiovascular disease is the leading cause of death in the United States, responsible for more than one in four deaths \citep{aspirin-recommendation}. The Physicians' Health Study (PHS), a large-scale, randomized, placebo-controlled trial conducted in the 1980s, was designed in part to investigate whether low-dose aspirin reduces the risk of cardiovascular mortality \citep{aspirin}. This landmark study reported significant reductions in both fatal and nonfatal myocardial infarctions in the treatment group, findings that played a crucial role in the widespread adoption of aspirin for heart attack prevention. Here, we revisit the aspirin component of the PHS, applying the BREASE framework to assess the sensitivity of its results to prior specification.

During the study, $y_0=26$ out of $N_0=$ 11,034 subjects in the placebo group experienced fatal myocardial infarction compared to $y_1=10$ out of $N_1=$ 11,037 prescribed aspirin. Using maximum likelihood estimation, the estimated risk ratio $\theta_1/\theta_0$ is 0.38, with 95\% confidence interval (based on inverting Fisher's exact test) $\text{CI}(95\%)=[0.17, 0.82]$. Consequently, we reject the null hypothesis of zero effect, $H_0:\theta_1=\theta_0$, with $p$-value $0.008$. Results based on asymptotic Wald and Pearson tests are nearly identical. Hence, a frequentist would confidently conclude that low-dose aspirin significantly reduces cardiovascular mortality in this population.

Bayesian estimation using \emph{default} priors under the alternative hypothesis (i.e, with a  prior that gives zero probability to the null hypothesis of zero effect) yields qualitatively similar, though more conservative answers. The $\text{BREASE}(1/2, \mu, \mu;  2, 1, 1)$ prior with $\mu=0.3$ yields a posterior median of the risk ratio of 0.44 with a wider 95\% credible interval of $\text{CrI}(95\%)=[0.2, 0.96]$. Results for the default IB and LT priors are qualitatively similar, though less conservative: the LT$(0,0;1,\sigma_\psi)$ with $\sigma_\psi=1$ results in a posterior median of 0.48 and $\text{CrI}(95\%)=[0.25, 0.87]$; the IB$(a,a;a,a)$ with $a=1$ returns posterior median 0.4 and $\text{CrI}(95\%)=[0.18, 0.79]$. But how sensitive are these results to the prior?

Varying the prior hyperparameter $\mu$ of the default BREASE prior (keeping prior sample sizes fixed at $n_e = n_s = 1$) shows that the results are indeed very sensitive to the prior. Credible intervals include the null of no effect as soon as $\mu \leq 0.2$. That is, unless \emph{a priori} we weakly expect efficacy or side-effects to be about 20\% or more, credible intervals would not exclude the null hypothesis of zero effect. This sensitivity also shows up, though it is less apparent, with the IB and LT parameterizations. For the LT prior, this happens when $\sigma_\psi \leq 0.4$. However, notice how the variance of the log odds ratio is harder to be directly interpreted than $\mu$.
For the IB, this happens only when $a \geq 17$; this prior specifies 17 deaths in the control and treatment groups, which is on par with the number of deaths observed in the data. Also notice that, in this example, inferences under an independent prior are less conservative than those under dependent priors. This is to be expected, because the LT and BREASE priors shrink estimates toward the null of no effect whereas the IB does not.

One may also be interested in performing a Bayesian hypothesis test based on the Bayes factor, which assigns nonzero prior probability to $H_0$. As we will see, prior sensitivity is even more pronounced in this case. Here we focus on the exact null, but we note that researchers can also specify an interval null hypothesis, such as $|\theta_1-\theta_0| < \delta$, as per the discussion in Section~\ref{sec:gd-testing}. Perhaps surprisingly, a test based on the IB approach yields a Bayes factor $\text{BF}_{01} = 20.27$, now suggesting that the data provide strong evidence \emph{in favor} of $H_0$. On the other hand, the Bayes factor under the LT approach is $\text{BF}_{10} = 5.24$, which suggests  moderate evidence in favor of $H_1:\theta_1\neq\theta_0$. Finally, the default BREASE prior results in $\text{BF}_{10}=1.2$ providing essentially little evidence in favor of one hypothesis or the other. Hence, when considering Bayes factors, unlike in the previous case, the IB prior results in more conservative inferences compared to the BREASE and LT priors. This occurs, however, for the same reason: under $H_1$, the IB prior assigns a substantial amount of mass to unreasonably large effect sizes. 

How can we make sense of these disparate results? One benefit of the BREASE approach is that it allows one to  clearly encode prior assumptions in terms of  the expected efficacy and side effects of aspirin, and to easily examine how sensitive the BF is to those assumptions, over the whole range possible values. For example, starting with $\mu_s$, aspirin is an over-the-counter medicine, with ample usage, and it would thus be \emph{unreasonable} to expect that aspirin would \emph{cause} myocardial infarction in a large fraction of otherwise healthy patients. Figure~\ref{fig:aspirin-side} inspects how the Bayes factor is affected as we vary the prior expectation of side effects, ranging from 0.01\% (reasonable) to 50\% (unreasonable), while still keeping relatively vague priors on the baseline risk and efficacy. The dashed red, orange, and blue lines denote (slightly modified) Jeffreys' thresholds for weak $(1\leq \text{BF}_{10} \leq 3)$,  
moderate $(3\leq \text{BF}_{10} \leq 10)$, and strong  $(\text{BF}_{10} \geq 10)$  evidence against $H_0$, respectively \citep{jeffreys1961,kass1995}. Indeed, as the plot shows, the results are sensitive to the choice of $\mu_s$. Setting the expected value of side effects to 1\% results in $\text{BF}_{10}=13.45$, yielding strong evidence in favor of $H_1$, while setting it to 50\% results in $\text{BF}_{01}=2.66$, yielding weak evidence in favor of $H_0$.

\begin{figure}[t]
\begin{subfigure}[t]{0.5\linewidth}
    \centering
    \includegraphics[scale = .55]{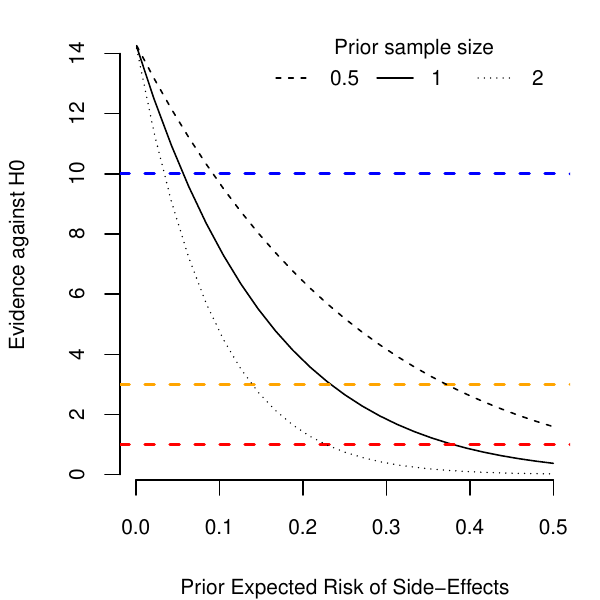}
    \caption{Sensitivity plot varying $\mu_s$ and $n_s$.}
    \label{fig:aspirin-side}
\end{subfigure}%
\begin{subfigure}[t]{0.5\linewidth}
    \centering
    \includegraphics[scale = .55]{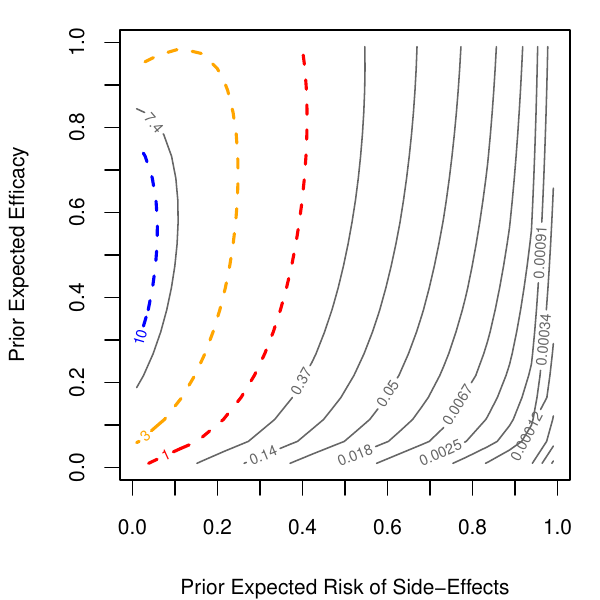}
    \caption{Sensitivity contours varying $\mu_s$ and $\mu_e$.}
    \label{fig:aspirin-contour}
\end{subfigure}
\caption{Sensitivity analysis of $\text{BF}_{10}$ for the aspirin trial.}
\end{figure}

We now conduct a sensitivity analysis with respect to both hyperparameters simultaneously. 
Figure~\ref{fig:aspirin-contour} shows the contour lines of $\text{BF}_{10}$ as a function of  $(\mu_e,\mu_s)\in(0,1)^2$ over their full range of possible values, while keeping $n_e=n_s=1$ fixed.  
Overall, only when (i)~side effects are expected to be small ($< 1\%$), and (ii)~the efficacy is expected to be relatively large (between 30\% and 70\%), does the Bayes factor provide strong evidence against the null of no effect. For all other combinations of prior hyperparameters, the evidence is either moderate, weak, or favors the null. In this light, the results of the trial are ambiguous, and the conclusion that aspirin is effective for primary prevention of fatal heart attack strongly depends on the prior. Note that this need not always be the case, as we show in our reanalysis of the Pfizer-BioNTech COVID-19 vaccine trial. 

\paragraph{Combining data from multiple trials.} Following the PHS, numerous subsequent trials in different study populations have later found mixed evidence for a reduction in cardiovascular events due to aspirin, along with increased risk of major hemorrhage \citep{aspirin-women,aspirin-arrive,aspirin-ascend} and, in older age groups, increased all-cause mortality \citep{aspirin-aspree}. Consequently,  
several organizations recommended against aspirin therapy for primary prevention of cardiovascular disease in elderly patients \citep{aspirin-acc,aspirin-recommendation}. 
In light of these findings, we now demonstrate how to pool evidence across multiple trials using the BREASE approach. Specifically, we focus on the risk of myocardial infarction (both fatal and non fatal), combining data from thirteen trials as analyzed in \citet{zheng2019association},  encompassing a total of 161,680 participants. 

Starting with a \emph{complete pooling} analysis, the default BREASE prior yields a posterior median for the risk ratio of 0.90, with CrI(95\%)=[0.84, 0.97]. The Bayes factor is 2.43, indicating only weak evidence against the null hypothesis.  Despite the large sample size, results are still very sensitive to the prior. For example, the 95\% credible interval includes the null of 1 as soon as $\mu_s<0.1$.  Next we apply a hierarchical BREASE prior, as discussed in Section~\ref{sec:covariates}, to \emph{partially pool} information across studies. We set a BREASE($\lambda$, $\nu$) prior on the hierarchical proportions $(\mu_0,\mu_e,\mu_s)$, with $\lambda = (.5, .5, .5)$, $\nu=(10, 10,10)$, and independent Gamma(10,.1) priors on $(n_0,n_e,n_s)$. As Table~\ref{tab:aspirin-meta-analysis} of Supplement~A shows, there is considerable effect heterogeneity across trials. The posterior median for the average effect is 0.9, with CrI(95\%)=[0.78, 1.13].

\subsection{The Pfizer-BioNTech COVID-19 vaccine trial}
\label{sec:covid}

We now  reexamine the results of the Pfizer-BioNTech mRNA COVID-19 vaccine study \citep{covid-vaccine}.
The experiment was a global multi-phase randomized placebo-controlled trial designed, in part, to evaluate the efficacy of the BNT162b2 vaccine candidate in preventing COVID-19. Vaccine development and evaluation were carried out in rapid response to the emerging SARS-CoV-2 pandemic. The results of the trial were definitive and precipitated the U.S. Food and Drug Administration's emergency use authorization for widespread dissemination of the vaccine \citep{fda-eua}. 

During the study, $y_1 = 9$ out of $N_1=$ 19,965 subjects  contracted COVID-19 subsequent to the second dose of the vaccine, while there were $y_0=169$  cases out of $N_0=$  20,172 subjects receiving placebo injections. In their paper, \citeauthor{covid-vaccine} adopted a Bayesian approach,  focusing particularly on evaluating the vaccine efficacy (VE), defined in the study as  the estimand $\text{VE}:=1-\theta_1/\theta_0$.  The efficacy of the vaccine was estimated at 0.95, with credible interval $\text{CrI}(95\%)= [0.90, 0.97]$. Frequentist estimates are similar, with a point estimate of 0.95, confidence interval $\text{CI}(95\%)= [0.90, 0.97]$, and a $p$-value for testing the null hypothesis of zero effect of the order $6\times10^{-33}$.

\citet{covid-vaccine} estimate VE as the efficacy of the vaccine, but, as per Section~\ref{sec:partial-mono}, this only has the counterfactual interpretation of efficacy (i.e., $\eta_e=1-\theta_1/\theta_0$) under the assumption of monotonicity. Using the BREASE approach we can easily encode the monotonicity assumption by setting $\eta_s=0$ and then proceed with estimation. The default BREASE prior, with the monotonicity constraint, results in posterior median and 95\% credible interval for $\eta_e = 1-\theta_1/\theta_0$ that are essentially the same as the previous results, namely, 0.94  and $\text{CrI}(95\%)= [0.90, 0.97]$. In the absence of the monotonicity assumption, we have that VE is in fact a lower bound on $\eta_e$. Again using the default BREASE prior, results are virtually unchanged, with posterior median and 95\% credible interval for VE of 0.94 and $\text{CrI}(95\%)= [0.90, 0.97]$. Conclusions from the IB(1,1;1,1) prior are practically equivalent: the posterior median of VE is 0.94 with $\text{CrI}(95\%)~=~[0.90,~0.97]$. Under the LT(0,0;1,1) prior, however, we obtain posterior median 0.91 and $\text{CrI}(95\%)=[0.86, 0.95]$, owing to the fact that it not only shrinks $\theta_0$ and $\theta_1$ toward each other, but also toward 0.5---see Figure 3 of \citet{dablander2022}.

\begin{figure}[t]
\begin{subfigure}[t]{0.5\linewidth}
    \centering
    \includegraphics[scale = .55]{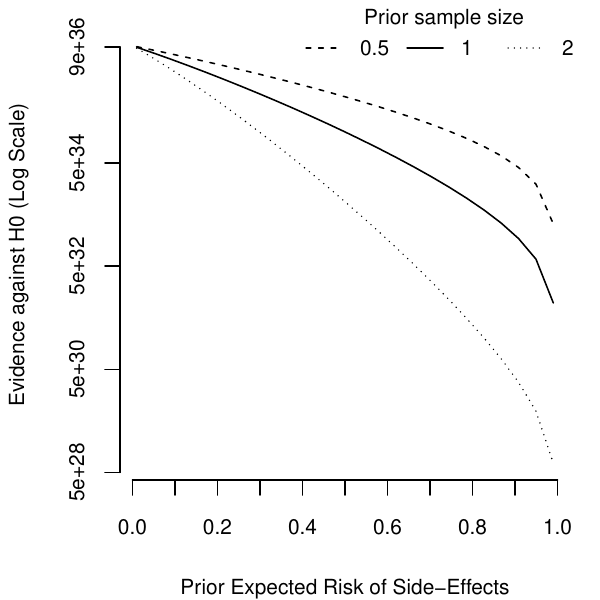}
    \caption{Sensitivity plot varying $\mu_s$ and $n_s$.}
    \label{fig:covid-side}
\end{subfigure}%
\begin{subfigure}[t]{0.5\linewidth}
    \centering
    \includegraphics[scale = .55]{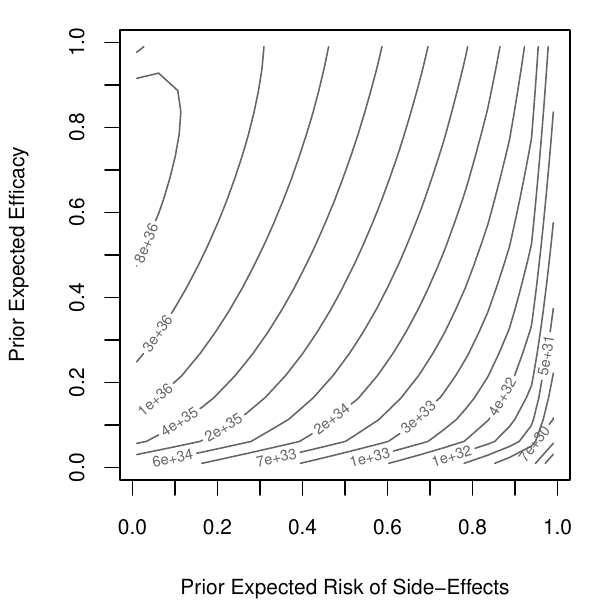}
    \caption{Sensitivity contours varying $\mu_s$ and $\mu_e$.}
    \label{fig:covid-contour}
\end{subfigure}
\caption{Sensitivity analysis of $\text{BF}_{10}$ for the COVID-19 vaccine trial.}
\label{fig:covid-sens}
\end{figure}

Turning to hypothesis testing, differently from the aspirin study, here all approaches point to the same direction, with overwhelming evidence against $H_0$. The Bayes factors against the null hypothesis of zero effect are $9\times 10^{33}$, $5\times 10^{34}$ and $4\times 10^{35}$ for the IB, LT and BREASE default priors, respectively. Further, sensitivity analyses reveal the Bayes factor is in fact robust to variations in the hyperparameters across the whole range of prior expected efficacy and side effects of the vaccine, i.e., $(\mu_e,\mu_s) \in(0,1)^2$. Figure~\ref{fig:covid-sens} replicates the same sensitivity plots of the aspirin study for the COVID-19 trial. Notice that, in all scenarios, the posterior probability of the null hypothesis is essentially zero even if we posit equal prior odds for $H_0$ and $H_1$.

\paragraph{Conditional vaccine efficacy.} 
In addition to overall VE for their sample, \citet{covid-vaccine} report estimates of VE across subgroups stratified by age, sex, race, ethnicity, and country. In many subgroups, sample sizes were too small to establish efficacy of the vaccine at the 30\% threshold prespecified by \citet{covid-vaccine}. For example, in the oldest age group of individuals 75 years or older---who face the greatest risk of death from COVID-19---the 95\% CrI for VE reported by \citet{covid-vaccine} ranges from -13.1\% to 100.0\%, which allows for the possibility that vaccination increases the risk of infection. Similarly, an age-stratified analysis using the independent BREASE prior discussed in Section \ref{sec:covariates} with our choice of default hyperparameters  yields a 95\% credible interval ranging from -10.8\% to 99.4\% for this age group. 

The situation improves if we allow for some pooling of information across age groups using a hierarchical prior, as described in Section \ref{sec:covariates}. Here we use the same hyperparameters as discussed in the aspirin example. Table \ref{tab:covariates} in Supplement A reports estimates of the Pfizer-BioNTech COVID-19 vaccine efficacy stratified by age, race, and country using the independent and hierarchical BREASE priors. With partial pooling, VE in the 75 and older age group now ranges from 45.0\% to 97\%, surpassing the 30\% threshold.

\subsection{Null results in the \textit{New England Journal of Medicine}}
\label{sec:nejm}

\begin{figure}[t]
\begin{subfigure}[t]{0.5\linewidth}
    \centering
    \includegraphics
    [scale = .33]
    {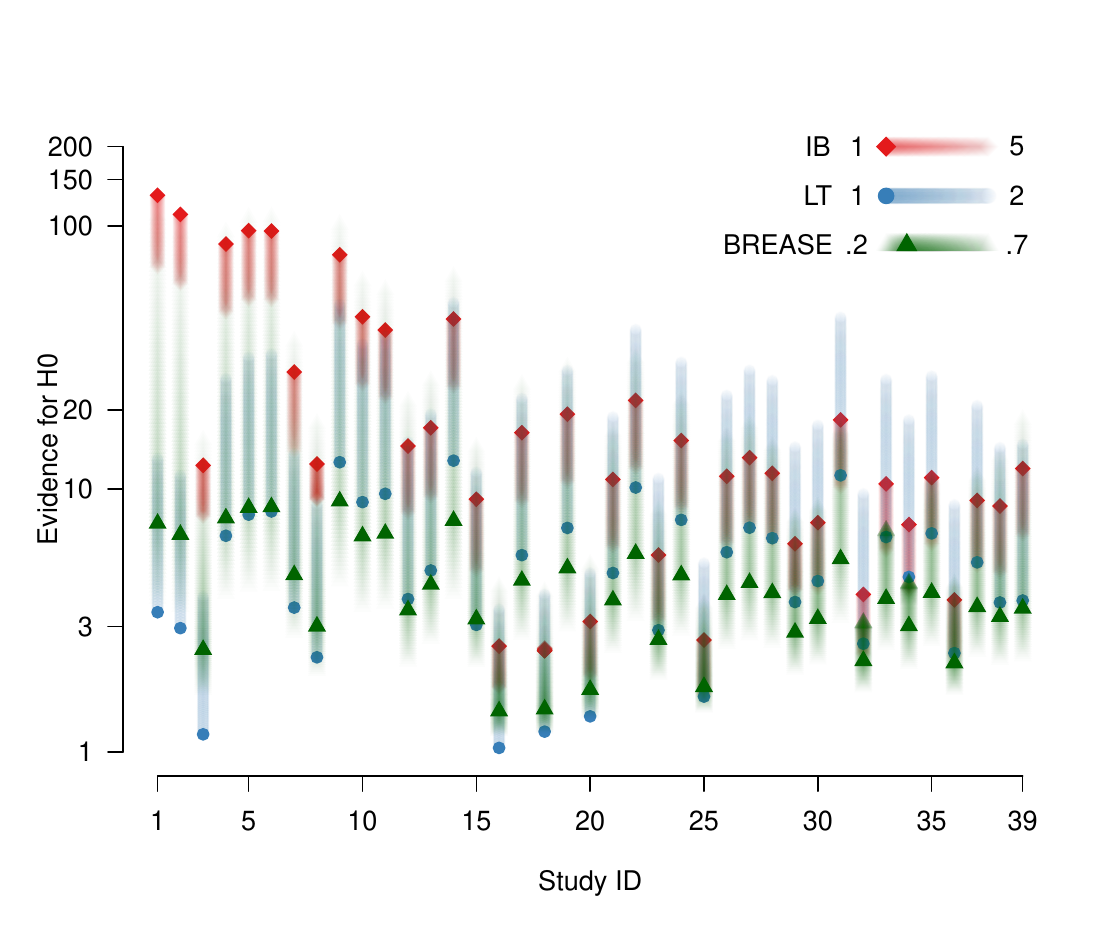}
    \caption{Bayes factors.}
    \label{fig:nejm-bf}
\end{subfigure}%
\begin{subfigure}[t]{0.5\linewidth}
    \centering
    \includegraphics
    [scale = .33]
    {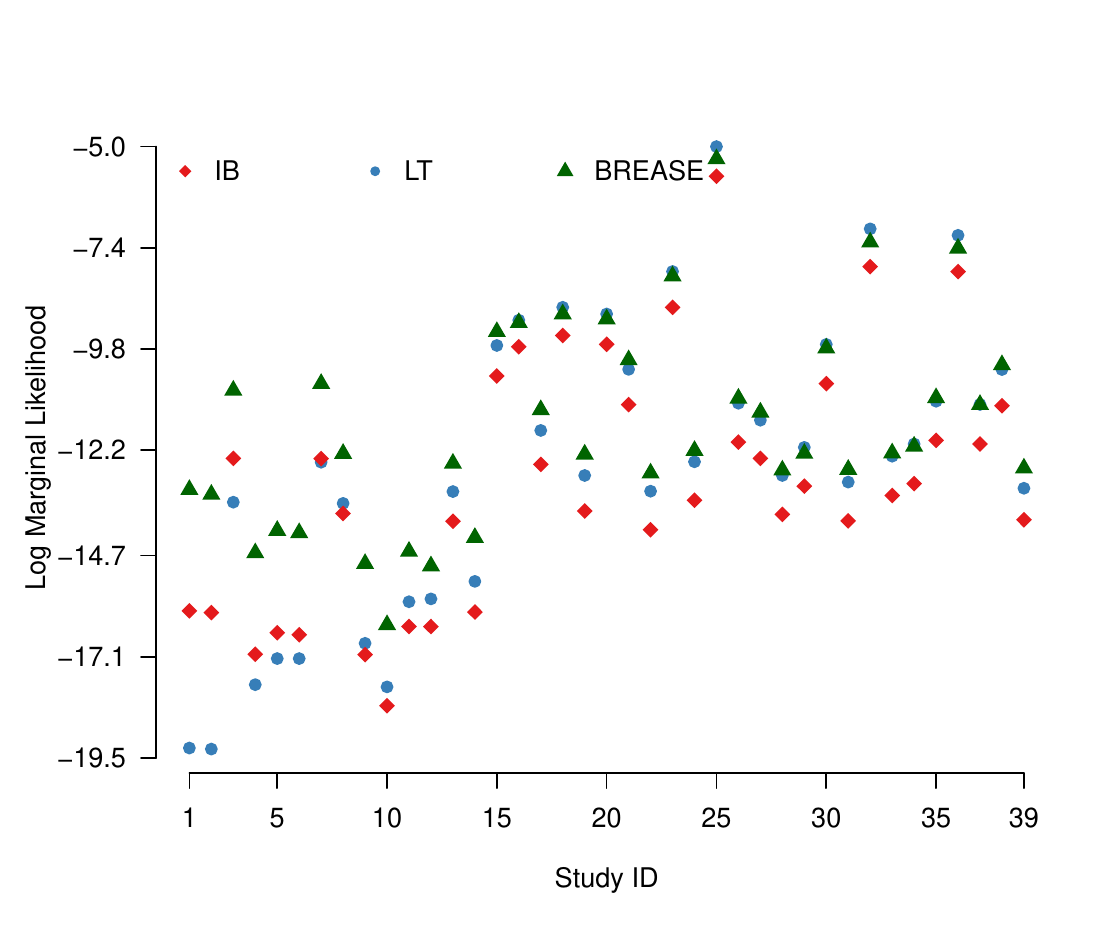}
    \caption{Log marginal likelihoods.}
    \label{fig:nejm-lml}
\end{subfigure}
\caption{Comparisons of log marginal likelihoods and Bayes factors across 39 NEJM studies, for the IB, LT and BREASE priors.}
\label{fig:nejm}
\end{figure}

\citet{dablander2022} conducted a Bayesian reanalysis of 39 binary experiments reporting null results (claiming absence or nonsignificance of an effect of treatment) in the \textit{New England Journal of Medicine} (NEJM). They were particularly concerned with distinguishing between \textit{absence of evidence} and \textit{evidence of absence} of an effect when outcomes in the treatment and control groups are similar. Finding that Bayes factors calculated using the IB approach often strongly favored the null hypothesis (leaning heavily toward \textit{evidence of absence}) whereas LT Bayes factors were generally equivocal, \citeauthor{dablander2022} concluded that the LT approach should be preferred for Bayesian tests for an equality of proportions. 
In our final empirical example, we expand their reanalysis to include the BREASE approach, and we show how it can easily 
address the concerns of \citeauthor{dablander2022}
while also providing a better fit to the data in most cases.

Figure~\ref{fig:nejm-bf} contrasts the Bayes factors in favor of the null hypothesis using: (i) the $\text{IB}(a,a;a,a)$ prior varying $a\in[1,5]$ (red diamonds); (ii) the $\text{LT}(0,0;1,\sigma_{\psi})$ prior varying $\sigma_{\psi} \in [1,2]$ (blue circles); and, the $\text{BREASE}(1/2,\mu,\mu;2, 1,1)$ prior varying $\mu \in [.2, .7]$ (green triangles). The solid color stands for the proposed default values of each method, namely $a=1$ for the IB, $\sigma_{\psi}=1$ for the LT and $\mu=.3$ for the BREASE. Note that the  Bayes factors of the BREASE and LT default priors (solid triangle and circles) are similar across studies. Moreover,
\citet{dablander2022} noted that, in many examples, the Bayes factors of the IB and LT approaches could not 
be easily reconciled,
even when reasonably varying their hyperparameters. The BREASE approach shows that this behavior is a mere artifact of those parameterizations. Indeed, for {all} studies, the BREASE prior easily interpolates between the two regimes, thus solving the apparent contradiction between the results of the LT and IB approaches. Finally, Figure~\ref{fig:nejm-lml} compares the predictive performance of the default IB, LT, and BREASE priors via the log marginal likelihood. The BREASE prior exhibits superior performance in \emph{every study} when compared to the IB prior, and in more than 74\% of the studies when compared to the LT prior. Thus, in this setting, our default prior provides both a more sensible parameterization and a better fit to the data.

\section{Conclusion}
\label{sec:discussion}

We have introduced the BREASE framework for the Bayesian analysis of randomized controlled trials with a binary treatment and outcome.  Framing the problem in the language of potential outcomes, we reparameterized the likelihood in terms of clinically meaningful quantities---the baseline risk,  efficacy, and risk of adverse side effects of the treatment---and proposed a simple, yet flexible  jointly independent beta prior distribution on these parameters. We provided algorithms for exact posterior sampling, an accurate and fast data-augmented Gibbs sampler, as well as analytical formulae for marginal likelihoods, Bayes factors, and other quantities. Finally, we showed with three empirical examples how our proposal facilitates estimation, hypothesis testing, and sensitivity analysis of treatment effects in binary experiments.

Many interesting extensions of our framework are possible. One interesting direction is to incorporate continous covariates in the model. For example, one possibility is to model BREASE parameters as functions of covariates on the logit scale, and use a Gibbs sampler that alternates between our data-augmentation algorithm for the BREASE parameters, and a specialized algorithm for logistic models, such as the P\'{o}lya-Gamma augmentation of \citet{polson2013bayesian}. Another important avenue for future work is handling noncompliance in clinical trials. In Section~\ref{app:iv} of Supplement~A, we lay the groundwork for such extension, and  show how the joint distribution of compliance and response types is naturally amenable to the BREASE parameterization and prior.

Beyond binary experiments, we may also consider trials with nonbinary outcomes or more than two arms. For example, with ordinal outcomes, one option is to replace $\eta_e$ and $\eta_s$ with the probability that treatment improves the outcome by one step and worsens the outcome by one step, respectively. 
In trials with more than two arms, we may again define the baseline risk $\theta_0$ in the control or standard of care group. Then, for each treatment arm $z$, we can introduce treatment-specific efficacy and side effect parameters, $\eta_e^z$ and $\eta_s^z$, respectively, and again place independent beta priors on each parameter to yield a tractable mixture posterior. If the treatments share some feature---e.g., they derive from a common family of therapeutics---and we have reason to believe that treatment effects are similar, we could instead place hierarchical priors on $\eta_e^z$ and $\eta_s^z$ to partially pool information across treatment arms, as described in Section~\ref{sec:covariates}.

Finally, while we have demonstrated how to apply our framework to pool evidence across multiple trials, many interesting questions remain open in that area. For example, under certain assumptions, data from multiple sites may allow one to point identify, or at least narrow the bounds on the fraction of people who benefit from or are harmed by the intervention. These counterfactual probabilities play an important role in public health and legal contexts. In a similar vein,  another possibility is to study our framework in the context of crossover trials. Under certain assumptions of temporal homogeneity, the efficacy and side effects of the treatment may again be identifiable, making our parameterization and prior proposal natural candidates for the study of treatment effects in such designs.

\paragraph{Acknowledgments. }
Irons’s research was supported by a Shanahan Endowment Fellowship, a Eunice Kennedy Shriver NICHD training grant, T32 HD101442-01, to the University of Washington Center for Studies in Demography \& Ecology, the Florence Nightingale Bicentenary Fellowship in Computational Statistics and Machine Learning from the University of Oxford Department of Statistics and the Leverhulme Centre for Demographic Science and the Leverhulme Trust (Grant RC-2018-003). Cinelli's research was supported in part by the Royalty Research Fund at the University of Washington, and by the National Science Foundation under Grant No. MMS-2417955.

\bibliographystyle{ba}
\bibliography{ref}

\appendix 

\section{Supplement}

\subsection{Implied prior on \texorpdfstring{$\theta_1$}{theta1}}
\label{app:theta1-dist}

Let the prior  of $(\theta_0, \eta_e, \eta_s)$ consist of  independent beta distributions with PDFs denoted by $\theta_0 \sim \pi_{\theta_0}(\theta_0)$, $\eta_s \sim \pi_{s}(\eta_s)$, and $\eta_e \sim \pi_{e}(\eta_e)$.  By the law of total probability, the conditional distribution of $\theta_1$ given $\theta_0$ can be written as
\begin{align}
    \pi(\theta_1\mid\theta_0) = \int_{0}^{1} \pi(\theta_1\mid\theta_0, \eta_e)\pi_{e}(\eta_e) d \eta_e, \label{eq:int}
\end{align}
where here we make use of the fact that $\eta_e$ and $\theta_0$ are \textit{a priori} independent. Note that, conditional on $\theta_0$ and $\eta_e$, $\theta_1$ is simply a linear transformation of $\eta_s$, namely $\theta_1 = \theta_0(1-\eta_e) + (1-\theta_0)\eta_s$. We can thus  write the density of $\theta_1$ in terms of the density of $\eta_s$ as
$$
\pi(\theta_1\mid\theta_0, \eta_e)= \left(\frac{1}{1-\theta_0}\right)\pi_{s}\left( \frac{\theta_1 - \theta_0(1-\eta_e)}{1-\theta_0}\right),
$$
where we make use of the fact that $\displaystyle \frac{d\eta_s}{d\theta_1} = \frac{1}{1-\theta_0}$. Substituting this back into Eq.~\ref{eq:int}, we have the following integral
\begin{align}
    \pi(\theta_1\mid\theta_0) = \left(\frac{1}{1-\theta_0}\right)\int_{0}^{1} \pi_{s}\left( \frac{\theta_1 - \theta_0(1-\eta_e)}{1-\theta_0}\right)\pi_{e}(\eta_e) d \eta_e .
\end{align}
For the special case where $\eta_e$ is uniformly distributed, $\pi_e(\eta_e)=1$, the integral simplifies,
\begin{align}
    \pi(\theta_1\mid\theta_0) &= \left(\frac{1}{1-\theta_0}\right)\int_{0}^{1} \pi_{s}\left( \frac{\theta_1 - \theta_0(1-\eta_e)}{1-\theta_0}\right) d \eta_e \\
    &=\left(\frac{1}{\theta_0}\right)\int_{\frac{\theta_1-\theta_0}{1-\theta_0}}^{\frac{\theta_1}{1-\theta_0}} \pi_{s}\left( \eta_s \right) d \eta_s \\
    &= \left(\frac{1}{\theta_0}\right)\left(F_{s}\left( \frac{\theta_1}{1-\theta_0}\right) - F_{s}\left( \frac{\theta_1-\theta_0}{1-\theta_0}\right)\right),
\end{align}
where the second equality follows from change of variables,  noting $d\eta_e = (1-\theta_0)/\theta_0 d\eta_s$. Here $F_s(\cdot)$ denotes the CDF of the beta distribution, which is given by the the regularized incomplete beta function. 

For special cases the expression above simplifies. For instance, when $\eta_s$ is also uniformly distributed, we have that $F_s(x)=x$, and we obtain a simple closed form expression for the conditional density. Specifically, for $\theta_0 \leq 1/2$, 
\begin{align}
    \pi(\theta_1\mid \theta_0) = 
\begin{cases}
    \dfrac{\theta_1}{\theta_0(1-\theta_0)} & \text{if } 0 \leq \theta_1 < \theta_0, \\[10pt]
    \dfrac{1}{1-\theta_0} & \text{if } \theta_0 \leq \theta_1 < 1-\theta_0, \\[10pt]
    \dfrac{1-\theta_1}{\theta_0(1-\theta_0)} & \text{if } 1-\theta_0 \leq \theta_1 \leq 1,
\end{cases}
\end{align}
and zero, otherwise. Analogously, for $\theta_0 \geq 1/2$, 
\begin{align}
    \pi(\theta_1\mid \theta_0) = 
\begin{cases}
    \dfrac{\theta_1}{\theta_0(1-\theta_0)} & \text{if } 0 \leq \theta_1 < 1-\theta_0, \\[10pt]
    \dfrac{1}{\theta_0} & \text{if } 1-\theta_0 \leq \theta_1 < \theta_0, \\[10pt]
    \dfrac{1-\theta_1}{\theta_0(1-\theta_0)} & \text{if } \theta_0 \leq \theta_1 \leq 1,
\end{cases}
\end{align}
and zero, otherwise. 
Notice this is a piece-wise linear function of $\theta_1$. Remarkably,  however, integrating each region over $\theta_0$ results in the following marginal distribution of $\pi(\theta_1)$, 
$$
\pi(\theta_1) = 2(-\theta_1\log \theta_1 - (1-\theta_1)\log(1-\theta_1)),
$$
for $\theta_1\in [0,1]$, and zero otherwise, which is twice the entropy of the $\text{Bernoulli}(\theta_1)$ distribution.

More generally, the distribution of linear combinations of beta random variables was studied in \citet{pham-gia1998} and is given in terms of Appell's first hypergeometric function $F_1$, which is an infinite series in two variables:
\begin{equation}
F_1(x,y;a;b_1,b_2;c) = \sum_{m_1=0}^\infty\sum_{m_2=0}^\infty \frac{\Gamma(a+m_1+m_2)\Gamma(b_1+m_1)\Gamma(b_2+m_2)\Gamma(c)}{\Gamma(a)\Gamma(b_1)\Gamma(b_2)\Gamma(c+m_1+m_2)}\frac{x^{m_1}}{m_1!}\frac{y^{m_2}}{m_2!}.
\label{eq:appell-sum}
\end{equation}
Appell's function also has an integral representation given by
\begin{equation}
F_1(x,y;a;b_1,b_2;c) = \text{B}(a,c-a)^{-1}\int_0^1 u^{a-1}(1-u)^{c-a-1}(1-ux)^{-b_1}(1-uy)^{-b_2} du.
\label{eq:appell-int}
\end{equation}
Applying the results of \citet{pham-gia1998} to our setup, the prior on $\theta_1$ conditional on $\theta_0$ induced by the BREASE prior can be obtained as the following piecewise function: (i) for $\theta_0\le 1/2$, we have

\begin{align}
\pi(\theta_1|\theta_0) &= I(0\le \theta_1\le \theta_0) \nonumber \\
&\qquad\times
\frac{\theta_1^{(1-\mu_e)n_e+\mu_sn_s-1}(\theta_0-\theta_1)^{\mu_en_e-1}\text{B}(\mu_sn_s,(1-\mu_e)n_e)}{\theta_0^{n_e-1}(1-\theta_0)^{\mu_sn_s}\text{B}(\mu_sn_s,(1-\mu_s)n_s)\text{B}((1-\mu_e)n_e,\mu_en_e)} \nonumber\\
&\qquad\times
F_1\left(\frac{-\theta_1}{\theta_0-\theta_1},\frac{\theta_1}{1-\theta_0};\mu_sn_s;1-\mu_en_e,1-(1-\mu_s)n_s;(1-\mu_e)n_e+\mu_sn_s\right) \nonumber \\
&+ I(\theta_0\le \theta_1\le 1-\theta_0) \nonumber \\
&\qquad\times \frac{(\theta_1-\theta_0)^{\mu_sn_s-1}(1-\theta_1)^{(1-\mu_s)n_s-1}}{(1-\theta_0)^{n_s-1}\text{B}(\mu_sn_s,(1-\mu_s)n_s)} \nonumber\\
&\qquad\times F_1\left(\frac{-\theta_0}{\theta_1-\theta_0},\frac{\theta_0}{1-\theta_1};\mu_en_e;1-\mu_sn_s,1-(1-\mu_s)n_s;n_e\right) \nonumber\\
&+ I(1-\theta_0\le \theta_1\le 1) \nonumber \\
&\qquad\times \frac{(1-\theta_1)^{\mu_en_e+(1-\mu_s)n_s-1}(\theta_1-\theta_0)^{\mu_sn_s-1}\text{B}(\mu_en_e,(1-\mu_s)n_s)}{\theta_0^{\mu_en_e}(1-\theta_0)^{n_s-1}\text{B}(\mu_sn_s,(1-\mu_s)n_s)\text{B}((1-\mu_e)n_e,\mu_en_e)}\nonumber\\
&\qquad\times F_1\left(\frac{1-\theta_1}{\theta_0},\frac{\theta_1-1}{\theta_1-\theta_0};\mu_en_e;1-(1-\mu_e)n_e,1-\mu_sn_s;\mu_en_e+(1-\mu_s)n_s\right).
\end{align}
Similarly, (ii) for $\theta_0\ge 1/2$, we have
\begin{align}
\pi(\theta_1|\theta_0) &= I(0\le \theta_1\le 1-\theta_0) \nonumber \\
&\qquad\times
\frac{\theta_1^{(1-\mu_e)n_e+\mu_sn_s-1}(1-\theta_0-\theta_1)^{(1-\mu_s)n_s-1}\text{B}((1-\mu_e)n_e,\mu_sn_s)}{(1-\theta_0)^{n_s-1}\theta_0^{(1-\mu_e)n_e}\text{B}((1-\mu_e)n_e,\mu_en_e)\text{B}(\mu_sn_s,(1-\mu_s)n_s)} \nonumber\\
&\qquad\times
F_1\left(\frac{-\theta_1}{1-\theta_0-\theta_1},\frac{\theta_1}{\theta_0};(1-\mu_e)n_e;1-(1-\mu_s)n_s,1-\mu_en_e;(1-\mu_e)n_e+\mu_sn_s\right) \nonumber \\
&+ I(1-\theta_0\le \theta_1\le \theta_0) \nonumber \\
&\qquad\times \frac{(\theta_1-(1-\theta_0))^{(1-\mu_e)n_e-1}(1-\theta_1)^{\mu_en_e-1}}{\theta_0^{n_e-1}\text{B}((1-\mu_e)n_e,\mu_en_e)} \nonumber\\
&\qquad\times F_1\left(\frac{-(1-\theta_0)}{\theta_1-(1-\theta_0)},\frac{1-\theta_0}{1-\theta_1};(1-\mu_s)n_s;1-(1-\mu_e)n_e,1-\mu_en_e;n_s\right) \nonumber\\
&+ I(\theta_0\le \theta_1\le 1) \nonumber \\
&\qquad\times \frac{(1-\theta_1)^{\mu_en_e+(1-\mu_s)n_s-1}(\theta_1-(1-\theta_0))^{(1-\mu_e)n_e-1}\text{B}((1-\mu_s)n_s,\mu_en_e)}{(1-\theta_0)^{(1-\mu_s)n_s}\theta_0^{n_e-1}\text{B}((1-\mu_e)n_e,\mu_en_e)\text{B}(\mu_sn_s,(1-\mu_s)n_s)}\nonumber\\
&\qquad\times F_1\left(\frac{1-\theta_1}{1-\theta_0},\frac{\theta_1-1}{\theta_1-(1-\theta_0)};(1-\mu_s)n_s;1-\mu_sn_s,1-(1-\mu_e)n_e;\mu_en_e+(1-\mu_s)n_s\right).
\end{align}

\paragraph{Monotonicity.} Under the ``no harm'' monotonicity assumption $\eta_s=0$ we have $\theta_1=(1-\eta_e)\theta_0$, in which case $\theta_1$ is a product of independent beta random variables \textit{a priori}. 
\citet{springer1970} derived the form of this distribution, with the density given as a Meijer $G$-function.
In general, this function is expressed as a contour integral in the complex plane. However, when $a_e=\mu_en_e$, $b_e=(1-\mu_e)n_e$, $a_0=\mu_0n_0$, and $b_0=(1-\mu_0)n_0$ are integers, the prior on $\theta_1$ can be expressed in closed form as 
\[
\pi(\theta_1) = \frac{\Gamma(n_0)\Gamma(n_e)}{\Gamma(\mu_0n_0)\Gamma((1-\mu_e)n_e)}\sum_{k=1}^m\sum_{j=0}^{e_k-1}\frac{K_{kj}\theta_1^{d_k-1}(-\log\theta_1)^{e_k-j-1}}{\Gamma(e_k-j)\Gamma(j+1)},
\]
where $\{d_1,\ldots,d_m\}$ denote the $m$ different integers occurring with multiplicity $\{e_1,\ldots,e_m\}$, respectively, among the sets $\{a_0-1,\ldots,a_0+b_0-2\}$ and $\{a_e-1,\ldots,a_e+b_e-2\}$, and 
\[
K_{kj} = \sum_{r=0}^j\sum_{q\in\{1,\ldots,m\},q\neq k}(-1)^{r+1}\binom{j}{r}\frac{\Gamma(r+1)e_q}{(d_q-d_k)^{r+1}}.
\] 
In particular, if 
$a_e+b_e = a_0$ (equivalently $n_e = \mu_0n_0$, an implicit assumption of the Dirichlet prior), 
we have 
\[
\theta_1 \sim \text{Beta}((1-\mu_e)n_e,\mu_en_e+(1-\mu_0)n_0).
\]
For another example, if $(\theta_0,\eta_e)\sim\text{Uniform}(0,1)^2$, we have
\[
\pi(\theta_1) = -\log\theta_1.
\]
Regarding the conditional prior $\pi(\theta_1|\theta_0)$ under the ``no harm'' assumption, it is clearly a scaled beta distribution, since $\theta_1=(1-\eta_e)\theta_0$. 
If $\eta_e\sim\text{Uniform}(0,1)$, we then have  that $\theta_1|\theta_0\sim\text{Uniform}(0,\theta_0)$. Similarly, under the ``no benefit'' assumption $\eta_e=0$, we have that $\theta_1=\theta_0+\eta_s(1-\theta_0)$, which is a scaled and shifted beta random variable conditional on $\theta_0$. If $\eta_s\sim\text{Uniform}(0,1)$, then $\theta_1|\theta_0\sim\text{Uniform}(\theta_0,1)$.

As for the moments,  applying the law of total covariance to the terms involving $\theta_1$ by conditioning on $\theta_0$ and making use of equation (\ref{eq:theta1}), we obtain
\begin{align*}
\text{Cov}(\theta_0,\theta_1) &= \frac{\mu_0(1-\mu_0)}{n_0+1}(1-\mu_e-\mu_s), \\
\text{Var}(\theta_0) &= \frac{\mu_0(1-\mu_0)}{n_0+1}, \\
\text{Var}(\theta_1) &=  \frac{\mu_0(1-\mu_0)}{n_0+1}(1-\mu_e-\mu_s)^2 \\
&\qquad+ \frac{\mu_e(1-\mu_e)}{n_e+1}\left\{\frac{\mu_0(1-\mu_0)}{n_0+1}+\mu_0^2\right\} \\
&\qquad+ \frac{\mu_s(1-\mu_s)}{n_s+1}\left\{\frac{\mu_0(1-\mu_0)}{n_0+1}+(1-\mu_0)^2\right\}.
\end{align*}
This can be used to obtain the prior correlation,
\begin{align*}
\text{Cor}(\theta_0,\theta_1) &= \frac{\text{Cov}(\theta_0,\theta_1)}{\sqrt{\text{Var}(\theta_0)\text{Var}(\theta_1)}}.
\end{align*}

\begin{figure}[p]
    \includegraphics[width=\textwidth]{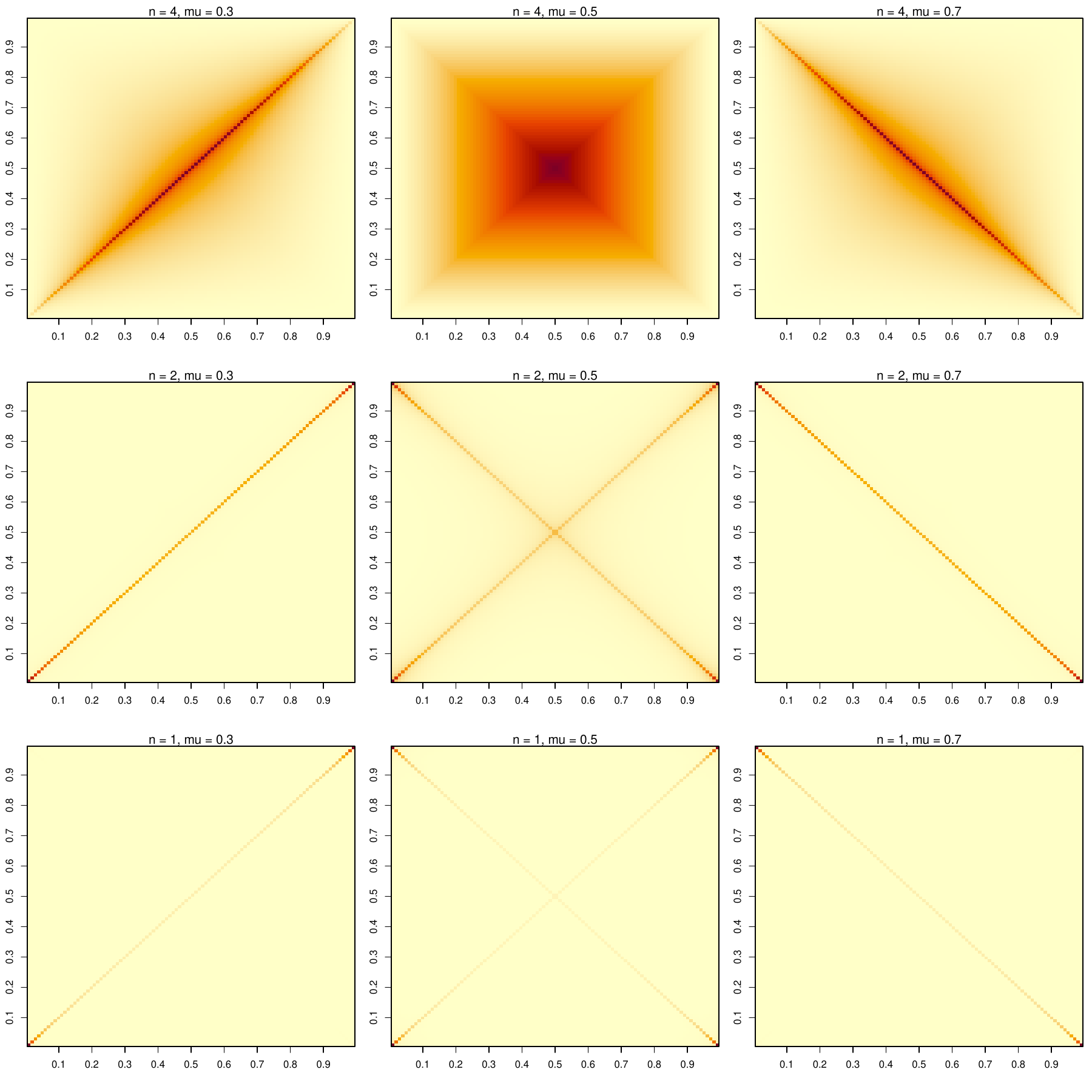}
    \caption{Heatmaps of the joint density of $(\theta_0,\theta_1)$ under the $\text{BREASE}(1/2,\mu,\mu;n,n/2,n/2)$ prior varying $n$ and $\mu$. Our proposed default prior takes $n=2$ and $\mu=.3$. As the plot shows, this: (i) leads to uniform marginals on $\theta_0$ and $\theta_1$; (ii) assumes zero treatment effect on average; (iii) concentrates mass on the diagonal $\theta_0 = \theta_1$; (iv) favors small (or large) proportions, instead of proportions around the center, which is expected when one quantifies rare outcomes such as death (proportions would be small) or, its complement, survival (in which case proportions would be large).}
    \label{fig:density-pns}
\end{figure}

\subsection{The generalized Dirichlet distribution on \texorpdfstring{$\mathbf{p}$}{p}}
\label{sec:gd}

Given a vector of probabilities $\mathbf{p}= (p_1, \dots, p_k)$, such that $\sum_{i=1}^{k}p_i=1$, the generalized Dirichlet distribution \citep{tian2003} is defined as,
\begin{align}
\pi(\mathbf{p}) \propto \prod_{i=1}^{k}p_{i}^{a_i -1} \prod_{j=1}^{m}\left(\sum_{i=1}^{k}\gamma_{ij}p_i\right)^{b_j-1} \label{eq:GD}
\end{align}
where $\Gamma = (\gamma_{ij})$ is a $k\times m$ known scale matrix.  We refer to (\ref{eq:GD}) as $\text{GD}(a, b, \Gamma)$. Now consider the vector of potential outcomes $\mathbf{p}=(p_{00}, p_{01}, p_{10}, p_{11})$.
By  change of variables arguments, if $(\theta_0,\eta_e,\eta_s)\sim \text{BREASE}(\mu;n)$ as in (\ref{eq:gd-betas}), it is easy to show that  $\mathbf{p}$ has density
\begin{align}
\pi(\mathbf{p}) \propto 
p_{00}^{(1-\mu_s)n_s-1}  p_{01}^{\mu_sn_s-1} p_{10}^{\mu_en_e-1} p_{11}^{(1-\mu_e)n_e-1} 
(p_{00}+p_{01})^{(1-\mu_0)n_0-n_s} (p_{10}+p_{11})^{\mu_0n_0-n_e}, 
\label{eq:grd-prior}
\end{align}
which is a $\text{GD}(a, b, \Gamma)$ distribution with parameters
\begin{align}
a &= (\mu_s n_s, (1-\mu_s)n_s, \mu_e n_e, (1-\mu_e)n_e), \nonumber \\
b &= ((1-\mu_0)n_0 - n_s + 1, \mu_0 n_0 - n_e + 1, 1, 1), \nonumber \\
\Gamma &= 
\begin{bmatrix}
1 & 0 & 1 & 0 \\
1 & 0 & 0 & 1 \\
0 & 1 & 1 & 0 \\
0 & 1 & 0 & 1
\end{bmatrix}.    \nonumber
\end{align}
The prior (\ref{eq:grd-prior}) is also a grouped Dirichlet distribution, as defined  in \citet{tian2003} and \citet{ng2008grouped} (which is a special case of the generalized Dirichlet). Similarly, the posterior in (\ref{eq:post}) induces the following posterior distribution on the vector~$\mathbf{p}$,
\begin{align*}
\pi(\mathbf{p}|\mathcal{D}) &\propto 
p_{00}^{(1-\mu_s)n_s-1}  p_{01}^{\mu_sn_s-1} p_{10}^{\mu_en_e-1} p_{11}^{(1-\mu_e)n_e-1} \\
&\times (p_{00}+p_{01})^{N_0-y_0+(1-\mu_0)n_0-n_s} (p_{10}+p_{11})^{y_0+\mu_0n_0-n_e}\\
&\times (p_{00}+p_{10})^{N_1-y_1} (p_{01}+p_{11})^{y_1},
\end{align*}
which is again a generalized Dirichlet distribution, $\text{GD}(a, b', \Gamma)$, with parameters $a$ and $\Gamma$ as in the prior, and updated parameter $b'$ given by
\begin{align}
b' &= (N_0 + y_0 + (1-\mu_0)n_0 - n_s + 1, y_0 + \mu_0 n_0 - n_e + 1, N_1-y_1 + 1, y_1 + 1). \nonumber 
\end{align}

The generalized Dirichlet distribution of \citet{dickey1983}, as well as special cases, such as the grouped Dirichlet and Dirichlet-beta, have been proposed for the Bayesian analysis of categorical data and contingency tables with missing observations \citep{karson1970,antelman1972,kaufman1973,gunel1984,dickey1987,tian2003,ng2008grouped}. 
The history of statistical analysis of contingency tables is extensive; \citet{killion1976} and \citet{agresti2005} provide comprehensive reviews. Along the lines of relevant studies already mentioned,
\citet{tian2003} and \citet{ng2008grouped} identify special cases of Dickey's generalized Dirichlet 
which admit alternative stochastic representations and simplified computation of posterior quantities. 
(Less relevant to our proposed methodology, other priors used to model contingency table proportions have been proposed in \citet{leonard1972,leonard1975,albert1982,basu1982,albert1983-prior,albert1983-dirichlet,albert1985-nonresponse,park1994}.)
These studies largely focused on the derivation of closed-form expressions (when available) and accurate approximations for posterior moments and predictive probabilities used in estimation and inference. 
They did not address the parameterization and interpretation of the generalized Dirichlet in terms of the baseline risk, efficacy, and side effects; algorithms for exact posterior simulation; testing for an effect of treatment and sensitivity analysis using analytical formulae; or the specific application to and prior elicitation for binary experiments.

\paragraph{The Dirichlet as a product of independent betas.}

To better understand the connection of the BREASE prior with the traditional Dirichlet distribution, it is instructive 
to first derive the distribution of $(\theta_0,\eta_e,\eta_s)$ induced by a Dirichlet prior on the response type probabilities~$\mathbf{p}$.  The 
BREASE parameters
can be expressed
as
\begin{align*}
\theta_0 = p_{10}+p_{11}, \qquad
\eta_e = \frac{p_{10}}{p_{10}+p_{11}}, \qquad
\eta_s = \frac{p_{01}}{p_{00}+p_{01}}.
\end{align*}
Elementary properties 
of the Dirichlet distribution then imply that these quantities are mutually independent beta random variables
\citep{ng2011dirichlet}
\begin{equation}
\theta_0 \sim \text{Beta}(a_{10}+a_{11},a_{00}+a_{01}) \quad\indep\quad \eta_e \sim \text{Beta}(a_{10},a_{11}) \quad\indep\quad \eta_s \sim \text{Beta}(a_{01},a_{00}).
\label{eq:dirichlet-betas}
\end{equation}
Similarly, since $\theta_1 = p_{01}+p_{11}$, we also have that $\theta_1\sim\text{Beta}(a_{01}+a_{11},a_{00}+a_{10})$ marginally.

While the Dirichlet density seems like a natural choice for the probability vector 
$\mathbf{p}$, the implied distribution on $(\theta_0,\eta_e,\eta_s)$ reveals some implicit assumptions. In particular, this prior has the peculiar (and potentially undesirable) feature that once we have decided on the parameters underlying the marginal distribution of the efficacy and side effects of treatment $(\eta_e,\eta_s)$---which requires specifying $(a_{00},a_{10},a_{01},a_{11})$---we have fully determined the joint prior on $(\theta_0,\eta_e,\eta_s)$. 
In this sense, the Dirichlet distribution is underparametrized. 

This underparameterization becomes clearer with an
alternative representation of the beta distribution,
in terms of the prior mean and prior ``sample size.'' 
For $\mu = a/(a+b)$ and $n = a+b$, we write $\text{Beta}^*(\mu,n)$ to denote a $\text{Beta}(a,b)$ distribution,
with mean $\mu$ and sample size $n$. 
The Dirichlet joint prior on $(\theta_0,\eta_e,\eta_s)$ has then the following alternative stochastic representation, 
\begin{equation}
\theta_0 \sim \text{Beta}^*(\mu_0,n_0) \quad\indep\quad \eta_e \sim \text{Beta}^*(\mu_e,\mu_0 n_0) \quad\indep\quad \eta_s \sim \text{Beta}^*(\mu_s,(1-\mu_0)n_0),
\label{eq:dirichlet-betas*}
\end{equation}
which is equivalent to the BREASE prior imposing a restriction on the choice of prior sample sizes $n_e$ and $n_s$. Marginally, we also have 
\begin{equation}
\theta_1\sim\text{Beta}^*\left((1-\mu_e)\mu_0+\mu_s(1-\mu_0),n_0\right),
\end{equation}
which resembles the decomposition~(\ref{eq:theta1}).

\subsection{Posterior sampling}
\label{sec:sampling-appendix}

\subsubsection{Proof of Theorem \ref{thm:sampling}}
\label{sec:sampling-proof}

We now describe in greater detail how to sample exactly from the BREASE posterior distribution via simulation. 

\begin{proof}[Proof of Theorem \ref{thm:sampling}]

For treatment or control group $j\in\{1,0\}$, respectively, define the 
response type (doomed, immune, preventive, and causal) counts in that group as
\begin{align*}
D_j &= \sum_{i=1}^N I(Z_i=j,Y_i(0)=1, Y_i(1) = 1),\\
I_j &= \sum_{i=1}^N I(Z_i=j,Y_i(0)=0, Y_i(1) = 0),\\
P_j &= \sum_{i=1}^N I(Z_i=j,Y_i(0)=1, Y_i(1) = 0),\\
C_j &= \sum_{i=1}^N I(Z_i=j,Y_i(0)=0, Y_i(1) = 1).
\end{align*}
The posterior can then be expressed as a mixture distribution:
\begin{align} 
&\pi(\theta_0,\eta_e,\eta_s|\mathcal{D}) = \sum_{C_1=0}^{y_1}\sum_{P_1=0}^{N_1-y_1} \pi(\theta_0,\eta_e,\eta_s|C_1,P_1,\mathcal{D}) 
\times \pi(C_1,P_1|\mathcal{D}).  
\end{align}
We will derive each term in the sum.
A straightforward calculation shows that
\begin{align*}
(D_j,I_j,P_j,C_j)&|(\theta_0,\eta_e,\eta_s,N_j) \sim \\
&\text{Multinomial}_{N_j}(\theta_0(1-\eta_e),(1-\theta_0)(1-\eta_s),\theta_0\eta_e,(1-\theta_0)\eta_s), \qquad j\in\{0,1\},
\end{align*}
and the two distributions are independent: $(D_0,I_0,P_0,C_0)\indep (D_1,I_1,P_1,C_1)$.
Since 
\[
y_1=D_1+C_1 \qquad\text{and}\qquad N_1-y_1= I_1+P_1,
\]
it follows that
\begin{align*}
C_1|(y_1,\theta_0,\eta_e,\eta_s) &\sim \text{Binomial}\left(y_1,\frac{(1-\theta_0)\eta_s}{\theta_1}\right), \\
P_1|(y_1,N_1,\theta_0,\eta_e,\eta_s) &\sim \text{Binomial}\left(N_1-y_1,\frac{\theta_0\eta_e}{1-\theta_1}\right),
\end{align*}
independently. Consequently, we have
\begin{align*}
\pi& (\theta_0,\eta_e,\eta_s|C_1,P_1,\mathcal{D}) \\
&\propto \pi(C_1,P_1,\mathcal{D}|\theta_0,\eta_e,\eta_s)\times \pi(\theta_0,\eta_e,\eta_s) \\
&= \pi(C_1,P_1|\mathcal{D},\theta_0,\eta_e,\eta_s)\times\pi(\mathcal{D}|\theta_0,\eta_e,\eta_s)\times \pi(\theta_0,\eta_e,\eta_s) \\
&= \pi(C_1|y_1,\theta_0,\eta_e,\eta_s) \times\pi(P_1|y_1,N_1,\theta_0,\eta_e,\eta_s)\\
&\qquad\times \pi(\mathcal{D}|\theta_0,\eta_e,\eta_s)\times\pi(\theta_0,\eta_e,\eta_s) \\
&= \text{Binomial}\left(C_1;y_1,\frac{(1-\theta_0)\eta_s}{\theta_1}\right)
\times \text{Binomial}\left(P_1;N_1-y_1,\frac{\theta_0\eta_e}{1-\theta_1}\right) \\ 
&\quad\times \text{Binomial}(y_0;N_0,\theta_0) 
\times \text{Binomial}(y_1;N_1,\theta_1) \\
&\quad\times\text{Beta}(\theta_0;\mu_0 n_0,(1-\mu_0)n_0) 
\times\text{Beta}(\eta_e;\mu_e n_e,(1-\mu_e)n_e) 
\times\text{Beta}(\eta_s;\mu_s n_s,(1-\mu_s)n_s) \\
&\propto 
\theta_0^{y_0+y_1-C_1+P_1+\mu_0n_0-1} (1-\theta_0)^{N-(y_0+y_1-C_1+P_1)+(1-\mu_0)n_0-1} \\
&\qquad\times \eta_e^{P_1+\mu_e n_e-1}(1-\eta_e)^{y_1-C_1+(1-\mu_e)n_e-1} \\
&\qquad\times \eta_s^{C_1+\mu_sn_s-1}(1-\eta_s)^{N_1-y_1-P_1+(1-\mu_s)n_s-1}.
\end{align*}
It follows that
\begin{align}
\pi& (\theta_0,\eta_e,\eta_s |C_1,P_1,\mathcal{D}) \nonumber\\
&=\text{Beta}(\theta_0;y_0+y_1-C_1+P_1+\mu_0n_0,N-(y_0+y_1-C_1+P_1)+(1-\mu_0)n_0) \nonumber\\
&\qquad\times\text{Beta}(\eta_e; P_1+\mu_e n_e, y_1-C_1+(1-\mu_e)n_e) \nonumber\\
&\qquad\times\text{Beta}(\eta_s; C_1+\mu_sn_s,N_1-y_1-P_1+(1-\mu_s)n_s).
\end{align}
Similarly, for the mixture weights we have
\begin{align}
\pi&(C_1,P_1|\mathcal{D}) =  \int \pi(C_1,P_1,\theta_0,\eta_e,\eta_s|\mathcal{D})d\theta_0 d\eta_e d\eta_s \nonumber\\
&= \int \pi(C_1,P_1|\theta_0,\eta_e,\eta_s,\mathcal{D})\pi(\theta_0,\eta_e,\eta_s|\mathcal{D})d\theta_0 d\eta_e d\eta_s \nonumber\\
&\propto \binom{y_1}{C_1}\binom{N_1-y_1}{P_1}
\text{B}(P_1+\mu_en_e,y_1-C_1+(1-\mu_e)n_e)
\nonumber\\
&\qquad\times \text{B}(y_0+y_1-C_1+P_1+\mu_0n_0,N-(y_0+y_1-C_1+P_1)+(1-\mu_0)n_0) \nonumber\\
&\qquad\times \text{B}(C_1+\mu_sn_s,N_1-y_1-P_1+(1-\mu_s)n_s).
\end{align}
Hence, we can sample from the mixture distribution (\ref{eq:mix}) as follows:
\begin{enumerate}[(i)]
    \item Sample 
    $P_1\in\{0,\ldots,N_1-y_1\}$ 
    conditional on $\mathcal{D}$ 
    with probability
    \[
    \pi(P_1|\mathcal{D}) =\sum_{C_1=0}^{y_1} \pi(C_1,P_1|\mathcal{D}).    
    \]
    according to (\ref{eq:mix-weights})
    \item Sample 
    $C_1\in\{0,\ldots,y_1\}$
    conditional on 
    $(P_1,\mathcal{D})$
    with probability
    \[
    \pi(C_1|P_1,\mathcal{D}) \propto \pi(C_1,P_1|\mathcal{D}).
    \]
    according to (\ref{eq:mix-weights}),
    \item Sample $(\theta_0,\eta_e,\eta_s)$ 
    conditional on 
    $(C_1,P_1,\mathcal{D})$ 
    from the independent beta (\ref{eq:mix-post}).
\end{enumerate}
    
\end{proof}

\subsubsection{Sampling under monotonicity: no harm}
\label{sec:sampling-noharm}

Here we derive the BREASE posterior sampling algorithm under the ``no harm'' $(\eta_s=0)$ monotonicity model $M_{-}^{'}$ (\ref{eq:m-'}).

\begin{theorem}
Let $(\theta_0,\eta_e)$ be random variables drawn according to Algorithm \ref{algo:sampling-noharm}. Then $(\theta_0,\eta_e)$ are distributed according to the posterior of model $M_{-}^{'}$ (\ref{eq:m-'}).
\label{thm:sampling-noharm}
\end{theorem}

\begin{proof}
In this case, we make use of the posterior mixture representation
\begin{equation}
\pi(\theta_0,\eta_e|\mathcal{D}) = \sum_{P_1=0}^{N_1-y_1} \pi(\theta_0,\eta_e|P_1,\mathcal{D})\times\pi(P_1|\mathcal{D}).
\label{eq:mix-noharm}
\end{equation}
As discussed in Section \ref{sec:sampling-proof}, we have
\[
P_1|(y_1,N_1,\theta_0,\eta_e) \sim \text{Binomial}\left(N_1-y_1,\frac{\theta_0\eta_e}{1-\theta_1}\right).
\]
Note that $\theta_1=(1-\eta_e)\theta_0$ by hypothesis. Consequently, we have
\begin{align*}
\pi& (\theta_0,\eta_e|P_1,\mathcal{D}) \\
&\propto \pi(P_1,\mathcal{D}|\theta_0,\eta_e)\times \pi(\theta_0,\eta_e) \\
&= \pi(P_1|\mathcal{D},\theta_0,\eta_e)\times\pi(\mathcal{D}|\theta_0,\eta_e)\times \pi(\theta_0,\eta_e) \\
&= \pi(P_1|y_1,N_1,\theta_0,\eta_e)\times \pi(\mathcal{D}|\theta_0,\eta_e)\times\pi(\theta_0,\eta_e) \\
&= \text{Binomial}\left(P_1;N_1-y_1,\frac{\theta_0\eta_e}{1-\theta_1}\right) \times \text{Binomial}(y_0;N_0,\theta_0) 
\times \text{Binomial}(y_1;N_1,\theta_1) \\
&\qquad\times\text{Beta}(\theta_0;\mu_0 n_0,(1-\mu_0)n_0) 
\times\text{Beta}(\eta_e;\mu_e n_e,(1-\mu_e)n_e) \\
&\propto 
\theta_0^{y_0+y_1+P_1+\mu_0n_0-1} (1-\theta_0)^{N-(y_0+y_1+P_1)+(1-\mu_0)n_0-1} \\
&\qquad\times \eta_e^{P_1+\mu_e n_e-1}(1-\eta_e)^{y_1+(1-\mu_e)n_e-1}.
\end{align*}
It follows that
\begin{align}
\pi (\theta_0,\eta_e|& P_1,\mathcal{D}) \nonumber\\
&=\text{Beta}(\theta_0;y_0+y_1+P_1+\mu_0n_0,N-(y_0+y_1+P_1)+(1-\mu_0)n_0) \nonumber\\
&\qquad\times\text{Beta}(\eta_e; P_1+\mu_e n_e, y_1+(1-\mu_e)n_e).
\label{eq:mix-post-noharm}
\end{align}
Similarly, for the mixture weights we have
\begin{align}
\pi(P_1|& \mathcal{D}) =  \int \pi(P_1,\theta_0,\eta_e|\mathcal{D})d\theta_0 d\eta_e \nonumber\\
&= \int \pi(P_1|\theta_0,\eta_e,\mathcal{D})\pi(\theta_0,\eta_e|\mathcal{D})d\theta_0 d\eta_e \nonumber\\
&\propto \binom{N_1-y_1}{P_1}
\text{B}(P_1+\mu_en_e,y_1+(1-\mu_e)n_e)
\nonumber\\
&\qquad\times \text{B}(y_0+y_1+P_1+\mu_0n_0,N-(y_0+y_1+P_1)+(1-\mu_0)n_0).
\label{eq:mix-weights-noharm}
\end{align}
Algorithm \ref{algo:sampling-noharm} defines the procedure to sample from the distribution \ref{eq:mix-noharm} based on these calculations. Algorithm \ref{algo:gibbs-noharm} defines the corresponding data-augmented Gibbs sampler.
\end{proof}

\begin{algorithm}
\caption{``No harm'' $(\eta_s=0)$ exact posterior sampling algorithm}
\label{algo:sampling-noharm}
\begin{algorithmic}
\item[]\textbf{Input:} Data $\mathcal{D}=(y_0,y_1,N_0,N_1)$, hyperparameters $(\mu_0,\mu_e,n_0,n_e)$, and desired number of posterior samples $T$.
\item[]\textbf{Iterate:} For sample $t\in\{1,\ldots,T\}$,
\begin{enumerate}[(i)]
    \item Sample $P_1\in\{0,\ldots,N_1-y_1\}$ 
    with probability $\pi(P_1|\mathcal{D})$ given by (\ref{eq:mix-weights-noharm}).
    \item Sample $(\theta_0,\eta_e)$ 
    conditional on $(P_1,\mathcal{D})$ 
    from the independent beta distribution (\ref{eq:mix-post-noharm}).
\end{enumerate}
\item[]\textbf{Output:} Posterior samples $\{(\theta_0^{(t)},\eta_e^{(t)})\}_{t\in\{1,\ldots,T\}}$.
\end{algorithmic} 
\end{algorithm}

\begin{algorithm}[t]
\caption{``No harm'' $(\eta_s=0)$ data-augmentation algorithm}
\label{algo:gibbs-noharm}
\begin{algorithmic}
\item[]\textbf{Input:} Data $\mathcal{D}=(y_0,y_1,N_0,N_1)$, hyperparameters $(\mu_0,\mu_e,n_0,n_e)$, desired number of posterior samples $T$, number of burn-in iterations $B$, and parameter initialization $(\theta_0^{(0)},\eta_e^{(0)})\in(0,1)^2$.
\item[]\textbf{Iterate:} For sample $t\in\{1,\ldots,T\}$, 
\begin{enumerate}[(i)]
    \item Sample $P_1^{(t)}$ 
    conditional on $(\theta_0^{(t-1)},\eta_e^{(t-1)},\mathcal{D})$ from the binomial distribution
    \begin{align}
        P_1^{(t)} &\sim \text{Binomial}\left(N_1-y_1,\frac{\theta_0^{(t-1)}\eta_e^{(t-1)}}{1-\theta_1^{(t-1)}}\right), 
    \end{align}
    where $\theta_1^{(t-1)} = \theta_0^{(t-1)}(1-\eta_e^{(t-1)})$.
    \item Sample $(\theta_0^{(t)},\eta_e^{(t)})$ conditional on $(P_1^{(t)},\mathcal{D})$
    from the independent beta (\ref{eq:mix-post-noharm}).
\end{enumerate}
\item[]\textbf{Output:} Posterior samples after burn-in $\{(\theta_0^{(t)},\eta_e^{(t)})\}_{t\in\{B+1,\ldots,T\}}$.
\end{algorithmic} 
\end{algorithm}

\subsubsection{Sampling under monotonicity: no benefit}
\label{sec:sampling-nobenefit}

Here we derive the BREASE posterior sampling algorithm under the ``no benefit'' $(\eta_e=0)$ monotonicity model $M_{+}^{'}$ (\ref{eq:m+'}).

\begin{theorem}
Let $(\theta_0,\eta_s)$ be random variables drawn according to Algorithm \ref{algo:sampling-nobenefit}. Then $(\theta_0,\eta_s)$ are distributed according to the posterior of model $M_{+}^{'}$ (\ref{eq:m+'}).
\label{thm:sampling-nobenefit}
\end{theorem}

\begin{proof}
In this case, we make use of the posterior mixture representation
\begin{equation}
\pi(\theta_0,\eta_s|\mathcal{D}) = \sum_{C_1=0}^{y_1} \pi(\theta_0,\eta_s|C_1,\mathcal{D})\times\pi(C_1|\mathcal{D}).
\label{eq:mix-nobenefit}
\end{equation}
As discussed in Section \ref{sec:sampling-proof}, we have
\[
C_1|(y_1,\theta_0,\eta_s) \sim \text{Binomial}\left(y_1,\frac{(1-\theta_0)\eta_s}{\theta_1}\right).
\]
Note that $\theta_1=\theta_0+(1-\theta_0)\eta_s$ by hypothesis. Consequently, we have
\begin{align*}
\pi(\theta_0,\eta_s|& C_1,\mathcal{D}) \\
&\propto \pi(C_1,\mathcal{D}|\theta_0,\eta_s)\times \pi(\theta_0,\eta_s) \\
&= \pi(C_1|\mathcal{D},\theta_0,\eta_s)\times\pi(\mathcal{D}|\theta_0,\eta_s)\times \pi(\theta_0,\eta_s) \\
&= \pi(C_1|y_1,\theta_0,\eta_s)\times \pi(\mathcal{D}|\theta_0,\eta_s)\times\pi(\theta_0,\eta_s) \\
&= \text{Binomial}\left(C_1;y_1,\frac{(1-\theta_0)\eta_s}{\theta_1}\right)
\times \text{Binomial}(y_0;N_0,\theta_0) 
\times \text{Binomial}(y_1;N_1,\theta_1) \\
&\qquad\times\text{Beta}(\theta_0;\mu_0 n_0,(1-\mu_0)n_0)
\times\text{Beta}(\eta_s;\mu_s n_s,(1-\mu_s)n_s) \\
&\propto 
\theta_0^{y_0+y_1-C_1+\mu_0n_0-1} (1-\theta_0)^{N-(y_0+y_1-C_1)+(1-\mu_0)n_0-1} \\
&\qquad\times \eta_s^{C_1+\mu_sn_s-1}(1-\eta_s)^{N_1-y_1+(1-\mu_s)n_s-1}.
\end{align*}
It follows that
\begin{align}
\pi(\theta_0,\eta_s |& C_1,\mathcal{D}) \nonumber\\
&=\text{Beta}(\theta_0;y_0+y_1-C_1+\mu_0n_0,N-(y_0+y_1-C_1)+(1-\mu_0)n_0) \nonumber\\
&\qquad\times\text{Beta}(\eta_s; C_1+\mu_sn_s,N_1-y_1+(1-\mu_s)n_s).
\label{eq:mix-post-nobenefit}
\end{align}
Similarly, for the mixture weights we have
\begin{align}
\pi(C_1|& \mathcal{D}) =  \int \pi(C_1,\theta_0,\eta_s|\mathcal{D})d\theta_0 d\eta_s \nonumber\\
&= \int \pi(C_1|\theta_0,\eta_s,\mathcal{D})\pi(\theta_0,\eta_s|\mathcal{D})d\theta_0 d\eta_s \nonumber\\
&\propto \binom{y_1}{C_1} \text{B}(y_0+y_1-C_1+\mu_0n_0,N-(y_0+y_1-C_1)+(1-\mu_0)n_0) \nonumber\\
&\qquad\times \text{B}(C_1+\mu_sn_s,N_1-y_1+(1-\mu_s)n_s).
\label{eq:mix-weights-nobenefit}
\end{align}
Algorithm \ref{algo:sampling-nobenefit} defines the procedure to sample from the distribution \ref{eq:mix-nobenefit} based on these calculations. Algorithm \ref{algo:gibbs-nobenefit} defines the corresponding data-augmented Gibbs sampler.
\end{proof}

\begin{algorithm}
\caption{``No benefit'' $(\eta_e=0)$ exact posterior sampling algorithm}
\label{algo:sampling-nobenefit}
\begin{algorithmic}
\item[]\textbf{Input:} Data $\mathcal{D}=(y_0,y_1,N_0,N_1)$, hyperparameters $(\mu_0,\mu_s,n_0,n_s)$, and desired number of posterior samples $T$.
\item[]\textbf{Iterate:} For sample $t\in\{1,\ldots,T\}$,
\begin{enumerate}[(i)]
    \item Sample 
    $C_1\in\{0,\ldots,y_1\}$ 
    conditional on $\mathcal{D}$ 
    with probability 
    $\pi(C_1|\mathcal{D})$ 
    given by (\ref{eq:mix-weights-nobenefit}).
    \item Sample $(\theta_0,\eta_s)$ 
    conditional on 
    $(C_1,\mathcal{D})$ 
    from the independent beta distribution (\ref{eq:mix-post-nobenefit}).
\end{enumerate}
\item[]\textbf{Output:} Posterior samples $\{(\theta_0^{(t)},\eta_s^{(t)})\}_{t\in\{1,\ldots,T\}}$.
\end{algorithmic} 
\end{algorithm}

\begin{algorithm}[t]
\caption{``No benefit'' $(\eta_e=0)$ data-augmentation algorithm}
\label{algo:gibbs-nobenefit}
\begin{algorithmic}
\item[]\textbf{Input:} Data $\mathcal{D}=(y_0,y_1,N_0,N_1)$, hyperparameters $(\mu_0,\mu_s,n_0,n_s)$, desired number of posterior samples $T$, number of burn-in iterations $B$, and parameter initialization $(\theta_0^{(0)},\eta_s^{(0)})\in(0,1)^2$.
\item[]\textbf{Iterate:} For sample $t\in\{1,\ldots,T\}$, 
\begin{enumerate}[(i)]
    \item Sample $C_1^{(t)}$ 
    conditional on $(\theta_0^{(t-1)},\eta_s^{(t-1)},\mathcal{D})$ from the binomial distribution
    \begin{align}
        C_1^{(t)} &\sim \text{Binomial}\left(y_1,\frac{(1-\theta_0^{(t-1)})\eta_s^{(t-1)}}{\theta_1^{(t-1)}}\right), 
    \end{align}
    where $\theta_1^{(t-1)} = \theta_0^{(t-1)} + (1-\theta_0^{(t-1)})\eta_s^{(t-1)}$.
    \item Sample $(\theta_0^{(t)},\eta_s^{(t)})$ conditional on $(C_1^{(t)},\mathcal{D})$
    from the independent beta (\ref{eq:mix-post-nobenefit}).
\end{enumerate}
\item[]\textbf{Output:} Posterior samples after burn-in $\{(\theta_0^{(t)},\eta_s^{(t)})\}_{t\in\{B+1,\ldots,T\}}$.
\end{algorithmic} 
\end{algorithm}

\subsubsection{Sampling with an alternate prior under 
\texorpdfstring{$H_0:\theta_0=\theta_1$}{H0} 
}
\label{sec:sampling-agg-dir-h0}

We now derive a sampling algorithm for the aggregated Dirichlet prior under $H_0$ introduced in Section \ref{sec:agg-dir-h0}:
\[
\mathbf{p}^* = (p_{00},p^*_{10},p_{11}) \sim \text{Dirichlet}((1-\mu_s)n_s,\mu_en_e+\mu_sn_s,(1-\mu_e)n_e), \qquad p^*_{10} = p_{10}+p_{01}.
\]
The algorithm is based on the posterior decomposition
\begin{align}
\pi(\mathbf{p}^*|\mathcal{D}) = \sum_{w(0)=0}^{y_0+y_1}\sum_{w(1)=0}^{N_0+N_1-y_0-y_1}&
\pi(\mathbf{p}^*|w(0),w(1),\mathcal{D})
\times\pi(w(0),w(1)|\mathcal{D}),
\label{eq:mix-post-h0}
\end{align}
where 
\[
w(0) = P_0+C_1, \qquad w(1) = C_0+P_1.
\]
We have
\begin{align*}
(P_0,D_0,I_0,C_0)&|(\mathbf{p}^*,N_0) \sim \text{Multinomial}_{N_0}(p^*_{10}/2,p_{11},p_{00},p^*_{10}/2), \\
(C_1,D_1,I_1,P_1)&|(\mathbf{p}^*,N_1) \sim \text{Multinomial}_{N_1}(p^*_{10}/2,p_{11},p_{00},p^*_{10}/2),
\end{align*}
and the two distributions are independent. 
It follows that
\begin{align*}
P_0|(y_0,\mathbf{p}^*) &\sim \text{Binomial}\left(y_0,\frac{p^*_{10}}{p^*_{10}+2p_{11}}\right), \\
C_0|(y_0,N_0,\mathbf{p}^*) &\sim \text{Binomial}\left(N_0-y_0,\frac{p^*_{10}}{p^*_{10}+2p_{00}}\right), \\
C_1|(y_1,\mathbf{p}^*) &\sim \text{Binomial}\left(y_1,\frac{p^*_{10}}{p^*_{10}+2p_{11}}\right), \\
P_1|(y_1,N_1,\mathbf{p}^*) &\sim \text{Binomial}\left(N_1-y_1,\frac{p^*_{10}}{p^*_{10}+2p_{00}}\right),
\end{align*}
independently. Hence, $w(0)$ and $w(1)$ are distributed independently as
\begin{align*}
w(0)
|(y_0,y_1,\mathbf{p}^*) 
&\sim \text{Binomial}\left(y_0+y_1,\frac{p^*_{10}}{p^*_{10}+2p_{11}}\right), \\
w(1)
|(\mathcal{D},\mathbf{p}^*) 
&\sim \text{Binomial}\left(N_0+N_1-y_0-y_1,\frac{p^*_{10}}{p^*_{10}+2p_{00}}\right),
\end{align*}
Consequently, we have
\begin{align*}
\pi (\mathbf{p}^*|w(0),& w(1),\mathcal{D}) \\
&\propto \pi(w(0),w(1),\mathcal{D}|\mathbf{p}^*)\times \pi(\mathbf{p}^*) \\
&= \pi(w(0),w(1)|\mathcal{D},\mathbf{p}^*)\times\pi(\mathcal{D}|\mathbf{p}^*)\times \pi(\mathbf{p}^*) \\
&= \pi(w(0)|y_0,y_1,\mathbf{p}^*) \times\pi(w(1)|\mathcal{D},\mathbf{p}^*) \\
&\qquad\times \pi(\mathcal{D}|\mathbf{p}^*)\times\pi(\mathbf{p}^*) \\
&= \text{Binomial}\left(w(0);y_0+y_1,\frac{p^*_{10}}{p^*_{10}+2p_{11}}\right) \\
&\qquad\times \text{Binomial}\left(w(1);N_0+N_1-y_0-y_1,\frac{p^*_{10}}{p^*_{10}+2p_{00}}\right) \\
&\qquad\times \text{Binomial}(y_0;N_0,p^*_{10}/2+p_{11}) 
\times \text{Binomial}(y_1;N_1,p^*_{10}/2+p_{11}) \\
&\qquad\times (p^*_{10})^{\mu_en_e+\mu_sn_s-1}p_{11}^{(1-\mu_e)n_e-1}p_{00}^{(1-\mu_s)n_s-1} \\
&\propto 
(p^*_{10})^{w(0)+w(1)+\mu_en_e+\mu_sn_s-1}
p_{11}^{y_0+y_1-w(0)+(1-\mu_e)n_e-1}p_{00}^{N_0+N_1-y_0-y_1-w(1)+(1-\mu_s)n_s-1}
\end{align*}
It follows that
\begin{align}
\mathbf{p}^*|(w(0),w(1),\mathcal{D}) &\sim
\text{Dirichlet}(a_{00},a_{10},a_{11}),
\label{eq:dir-h0}
\end{align}
where
\begin{align*}
a_{00} &= N_0+N_1-y_0-y_1-w(1)+(1-\mu_s)n_s, \\
a_{10} &= w(0)+w(1)+\mu_en_e+\mu_sn_s, \\
a_{11} &= y_0+y_1-w(0)+(1-\mu_e)n_e.
\end{align*}
Consequently, for the mixture weights we have
\begin{align}
\pi&(w(0),w(1)|\mathcal{D}) =  \int \pi(w(0),w(1),\mathbf{p}^*|\mathcal{D})d\mathbf{p}^* \nonumber\\
&= \int \pi(w(0),w(1)|\mathbf{p}^*,\mathcal{D})\pi(\mathbf{p}^*|\mathcal{D})d\mathbf{p}^* \nonumber\\
&\propto \binom{y_0+y_1}{w(0)}\binom{N_0+N_1-y_0-y_1}{w(1)} \nonumber\\
&\qquad\times\int (p^*_{10}/2)^{w(0)+w(1)+\mu_en_e+\mu_sn_s-1}p_{11}^{y_0+y_1-w(0)+(1-\mu_e)n_e-1}p_{00}^{N_0+N_1-y_0-y_1-w(1)+(1-\mu_s)n_s-1}d\mathbf{p}^* \nonumber\\
&\propto 2^{-(w(0)+w(1))}\binom{y_0+y_1}{w(0)}\binom{N_0+N_1-y_0-y_1}{w(1)}\text{B}(a_{00},a_{10},a_{11}).
\label{eq:mix-weights-h0}
\end{align}
Algorithm \ref{algo:sampling-h0} defines the procedure to sample from the distribution \ref{eq:mix-post-h0} based on these calculations. Algorithm \ref{algo:gibbs-h0} defines the corresponding data-augmented Gibbs sampler.

\begin{algorithm}[htbp!]
\caption{
Alternate $H_0:\theta_0=\theta_1$ exact posterior sampling algorithm}
\label{algo:sampling-h0}
\begin{algorithmic}
\item[]\textbf{Input:} Data $(y_0,y_1,N_0,N_1)$, hyperparameters $(\mu_e,\mu_s,n_e,n_s)$, and 
posterior samples $T$.
\item[]\textbf{Iterate:} For sample $t\in\{1,\ldots,T\}$,
\begin{enumerate}[(i)]
    \item Sample $w(1)\in\{0,\ldots,N_0+N_1-y_0-y_1\}$ conditional on $(y_0,y_1,N_0,N_1)$ 
    as
    \[
    \pi(w(1)|y_0,y_1,N_0,N_1) =\sum_{w(0)=0}^{y_0+y_1} \pi(w(0),w(1)|y_0,y_1,N_0,N_1).
    \]
    \item Sample $w(0)\in\{0,\ldots,y_0+y_1\}$ conditional on $(w(1),y_0,y_1,N_0,N_1)$
    with probability
    \[
    \pi(w(0)|w(1),y_0,y_1,N_0,N_1) \propto \pi(w(0),w(1)|y_0,y_1,N_0,N_1).
    \]
    \item Sample $\mathbf{p}^*=(p_{00},p^*_{10},p_{11})$ 
    conditional on $(w(0),w(1),y_0,y_1,N_0,N_1)$ 
    from the Dirichlet distribution (\ref{eq:dir-h0}).
    \item Transform $\mathbf{p}^*$ to obtain samples of $(\theta_0,\theta_1,\eta_e,\eta_s)$ via
    \begin{align*}
    \theta_0 &= p^*_{10}/2+p_{11} = \theta_1, \qquad
    \eta_e = \frac{p^*_{10}}{p^*_{10}+2p_{11}}, \qquad
    \eta_s = \frac{p^*_{10}}{p^*_{10}+2p_{00}}.
    \end{align*}
\end{enumerate}
\item[]\textbf{Output:} Posterior samples $\{((\mathbf{p}^*)^{(t)},\theta_0^{(t)},\theta_1^{(t)},\eta_e^{(t)},\eta_s^{(t)})\}_{t\in\{1,\ldots,T\}}$.
\end{algorithmic} 
\end{algorithm}

\begin{algorithm}[t]
\caption{
Alternate $H_0:\theta_0=\theta_1$ data-augmentation algorithm}
\label{algo:gibbs-h0}
\begin{algorithmic}
\item[]\textbf{Input:} Data $\mathcal{D}=(y_0,y_1,N_0,N_1)$, hyperparameters $(\mu_e,\mu_s,n_e,n_s)$, number of posterior samples $T$, number of burn-in iterations $B$, and simplex parameter initialization $(\mathbf{p}^*)^{(0)}$.
\item[]\textbf{Iterate:} For sample $t\in\{1,\ldots,T\}$, 
\begin{enumerate}[(i)]
    \item Sample 
    $(w(0)^{(t)},w(1)^{(t)})$ 
    conditional on $((\mathbf{p}^*)^{(t-1)},\mathcal{D})$ from the independent binomial
    \begin{align*}
        w(0)^{(t)}
        |(y_0,y_1,(\mathbf{p}^*)^{(t-1)}) 
        &\sim \text{Binomial}\left(y_0+y_1,\frac{(p^*_{10})^{(t-1)}}{(p^*_{10})^{(t-1)}+2p_{11}^{(t-1)}}\right), \\
        w(1)^{(t)}
        |(\mathcal{D},(\mathbf{p}^*)^{(t-1)}) 
        &\sim \text{Binomial}\left(N_0+N_1-y_0-y_1,\frac{(p^*_{10})^{(t-1)}}{(p^*_{10})^{(t-1)}+2p_{00}^{(t-1)}}\right).
    \end{align*}
    \item Sample $(\mathbf{p}^*)^{(t)}$ conditional on 
    $(w(0)^{(t)},w(1)^{(t)},\mathcal{D})$
    from the Dirichlet \eqref{eq:dir-h0}.
    \item Transform $(\mathbf{p}^*)^{(t)}$ to obtain samples $(\theta_0^{(t)},\theta_1^{(t)},\eta_e^{(t)},\eta_s^{(t)})$ via
    \begin{align*}
    \theta_0 &= p^*_{10}/2+p_{11} = \theta_1, \qquad
    \eta_e = \frac{p^*_{10}}{p^*_{10}+2p_{11}}, \qquad
    \eta_s = \frac{p^*_{10}}{p^*_{10}+2p_{00}}.
    \end{align*}    
\end{enumerate}
\item[]\textbf{Output:} Posterior samples after burn-in $\{((\mathbf{p}^*)^{(t)},\theta_0^{(t)},\theta_1^{(t)},\eta_e^{(t)},\eta_s^{(t)})\}_{t\in\{B+1,\ldots,T\}}$.
\end{algorithmic} 
\end{algorithm}

\subsection{Posterior quantities of interest}
\label{sec:post-quant}

In addition to marginal likelihoods, we can derive analytical expressions for certain relevant functionals of the BREASE posterior distribution $\pi(\theta_0,\eta_e,\eta_s|\mathcal{D})$. While posterior quantities can generally be easily estimated using simple Monte Carlo approximation with samples obtained from Algorithm \ref{algo:sampling}, analytical formulae may be of value, e.g., for conducting prior sensitivity analysis of treatment effect estimands 
without needing to sample the posterior for every choice of the hyperparameters $(\mu,n)$.

The risk difference $\theta_1-\theta_0$ 
and risk ratio $\theta_1/\theta_0$ 
are of particular interest in practice, with expectations of their posterior distributions often reported. We first note that, since the posterior $\pi(\theta_0,\eta_e,\eta_s|\mathcal{D})$ is a mixture of independent beta distributions, conditional and marginal expectations and percentiles can be easily computed by first calculating expectations or percentiles of the beta summands and averaging these quantities across the mixture weights. For example, using the mixture representation (\ref{eq:mix}) of the posterior, we have
\begin{align*}
\mathbb{E}[\theta_0|\mathcal{D}] &= 
\int \theta_0 \cdot \pi(\theta_0,\eta_e,\eta_s|\mathcal{D})d\theta_0d\eta_e d\eta_s \\
&= \sum_{y_1(0)=0}^{y_1}\sum_{x_1(1)=0}^{N_1-y_1}\pi(y_1(0),x_1(1)|\mathcal{D}) \int \theta_0 \cdot\pi(\theta_0,\eta_e,\eta_s|y_1(0),x_1(1),\mathcal{D}) d\theta_0d\eta_e d\eta_s.
\end{align*}
Applying equations (\ref{eq:mix-weights}) and (\ref{eq:mix-post}) then yields an expression for $\mathbb{E}[\theta_0|\mathcal{D}]$ in terms of the data $\mathcal{D}$ and hyperparameters $(\mu,n)$, which we omit for brevity. In a similar fashion, by exploiting the mixture-of-betas representation of the posterior, we can easily calculate posterior expectations of polynomials $
\sum_{(\alpha_0,\alpha_e,\alpha_s)}a_{(\alpha_0,\alpha_e,\alpha_s)}\theta_0^{\alpha_0}\eta_e^{\alpha_e}\eta_s^{\alpha_s}$, including those with negative exponents, assuming $\mathcal{D}$ and $(\mu,n)$ are such that the integrals converge. 

In particular, assuming treatment is not harmful $(\eta_s=0)$, the efficacy can be written in terms of the risk ratio as $\eta_e = 1-\theta_1/\theta_0$. The formulae derived in Appendix \ref{sec:sampling-noharm} can then be applied to calculate $\mathbb{E}[\theta_1/\theta_0|\mathcal{D}]=1-\mathbb{E}[\eta_e|\mathcal{D}]$ using the posterior $\pi(\theta_0,\eta_e|\mathcal{D})$ under the monotonicity assumption. More generally, 
we have
\begin{align*}
\mathbb{E}[\theta_1/\theta_0|\mathcal{D}] &= 
\mathbb{E}\left[\frac{\theta_0(1-\eta_e-\eta_s)+\eta_s}{\theta_0}\bigg|\mathcal{D}\right] \\
&= 1-\mathbb{E}[\eta_e|\mathcal{D}]-\mathbb{E}[\eta_s|\mathcal{D}] + 
\mathbb{E}[\theta_0^{-1}\eta_s|\mathcal{D}].
\end{align*}
Similarly, the expected posterior risk difference can be obtained as
\begin{align*}
\mathbb{E}[\theta_1-\theta_0|\mathcal{D}] &= \mathbb{E}[\eta_s|\mathcal{D}] - \mathbb{E}[\theta_0\eta_e|\mathcal{D}] - \mathbb{E}[\theta_0\eta_s|\mathcal{D}].
\end{align*}
In Section \ref{sec:results} we demonstrate how to conduct sensitivity analysis with the BREASE prior for Bayes factors using the marginal likelihoods derived in Section \ref{sec:gd-testing}. The discussion therein applies just as well to treatment effects and other posterior quantities.  

\subsection{Alternative models and priors}

\subsubsection{Other priors for \texorpdfstring{$H_0$}{}}
\label{app:h0}

Recalling that $\theta_0=p_{10}+p_{11}$ and $\theta_1=p_{01}+p_{11}$, we see that $\theta_0=\theta_1$ if and only if $p_{10} = p_{01}$. In this light, we discuss some alternate priors 
that conform to these constraints.
While  instantiating  $H_0$ using the beta-binomial model $M_0$ (\ref{eq:m0}) 
should be preferable in most applications, the prior we discuss here may apply in cases where one has stronger prior information concerning the efficacy and side effects of treatment $(\eta_e,\eta_s)$ rather than the baseline risk $\theta_0$ itself.

\subsubsection{Aggregated Dirichlet}
\label{sec:agg-dir-h0}

With a $\text{Dirichlet}^*(\mu_0,\mu_e,\mu_s;n_0)$ prior on $\mathbf{p}$, we have by the aggregation property of the Dirichlet distribution 
\citep{ng2011dirichlet}
\[
(p_{00},p_{10}+p_{01},p_{11}) \sim \text{Dirichlet}((1-\mu_s)n_s,\mu_en_e+\mu_sn_s,(1-\mu_e)n_e),
\]
where $n_e = \mu_0n_0$ and $n_s = (1-\mu_0)n_0$.
Assuming $H_0$ holds, and defining $p^*_{10}= p_{10}+p_{01}=2p_{10}$, 
we obtain the Dirichlet prior density on the aggregated cell probabilities
\begin{align*}
\pi(p_{00},p^*_{10}) &= \text{B}((1-\mu_s)n_s,\mu_en_e+\mu_sn_s,(1-\mu_e)n_e)^{-1}p_{00}^{(1-\mu_s)n_s-1}(p^*_{10})^{\mu_en_e+\mu_sn_s-1}p_{11}^{(1-\mu_e)n_e-1},
\end{align*}
where $p_{11}=1-p_{00}-p^*_{10}$ and $\text{B}(a_{00},a_{10},a_{11})$ is the multivariate beta function:
\[
\text{B}(a_{00},a_{10},a_{11}) = \frac{\Gamma(a_{00})\Gamma(a_{10})\Gamma(a_{11})}{\Gamma(a_{00}+a_{10}+a_{11})}.
\]
This prior allows for exact posterior sampling and marginal likelihood calculation in cases where we may have stronger prior information concerning the efficacy and side effects of treatment $(\eta_e,\eta_s)$ than the baseline risk $\theta_0$. Indeed, note that the prior is fully specified by the hyperparameters $(\mu_e,\mu_s,n_e,n_s)$. 
Recalling that the Dirichlet$^*$ prior is obtained from the generalized Dirichlet by setting $n_e=\mu_0n_0$ and $n_s=(1-\mu_0)n_0$, we see that this prior assumes that we have as much prior knowledge on $\theta_0$ as we do on $(\eta_e,\eta_s)$.

With this parametrization, the likelihood under $H_0$ is given by
\begin{align*}
L(\mathcal{D}|p) &= \binom{N_0}{y_0}\binom{N_1}{y_1}(p_{10}^*/2+p_{11})^{y_0+y_1}(p_{00}+p_{10}^*/2)^{N_0+N_1-y_0-y_1}.
\end{align*}
The posterior is then
\begin{align*}
\pi(p_{00},p^*_{10}&|\mathcal{D}) \propto \binom{N_0}{y_0}\binom{N_1}{y_1} \text{B}((1-\mu_s)n_s,\mu_en_e+\mu_sn_s,(1-\mu_e)n_e)^{-1} \\
&\qquad\times (p^*_{10}/2+p_{11})^{y_0+y_1} (p_{00}+p^*_{10}/2)^{N_0+N_1-y_0-y_1} \\
&\qquad\times  (p^*_{10})^{\mu_en_e+\mu_sn_s-1}p_{11}^{(1-\mu_e)n_e-1}p_{00}^{(1-\mu_s)n_s-1}.
\end{align*}
From here
we can apply the binomial theorem twice to quickly see that the posterior is a mixture of Dirichlet densities on the probability vector $\mathbf{p}^*=(p_{00},p^*_{10},p_{11})$. This yields the marginal likelihood formula
\begin{align*}
&L(\mathcal{D}) = \binom{N_0}{y_0}\binom{N_1}{y_1}
\text{B}((1-\mu_s)n_s,\mu_en_e+\mu_sn_s,(1-\mu_e)n_e)^{-1}\\ 
&\quad\times\sum_{j=0}^{y_0+y_1}\sum_{k=0}^{N_0+N_1-y_0-y_1}2^{-(j+k)}\binom{y_0+y_1}{j}\binom{N_0+N_1-y_0-y_1}{k} \text{B}(a_{00}(j,k),a_{10}(j,k),a_{11}(j,k)),
\end{align*}
where we define
\begin{align*}
a_{00}(j,k) &= N_0+N_1-y_0-y_1+(1-\mu_s)n_s-k, \\ 
a_{10}(j,k) &= j+k+\mu_en_e+\mu_sn_s, \\ 
a_{11}(j,k) &= y_0+y_1+(1-\mu_e)n_e-j.
\end{align*}
In Section \ref{sec:sampling-agg-dir-h0}, we derive an algorithm for exact posterior sampling using the aggregated Dirichlet prior on $(p_{00},p^*_{10},p_{11})$.

\subsubsection{Other priors for \texorpdfstring{$H_-$}{} and \texorpdfstring{$H_+$}{}}
\label{app:h-h+}

Another approach for specifying models for $H_-$ and $H_+$, which is both natural and computationally convenient, is to impose a monotonicity assumption on $M_1$, and set $\eta_s=0$ or $\eta_e=0$ respectively. This results in the following models,
\begin{align}
M_{-}^{'}:~&(\theta_0,\eta_e) \sim \text{Beta}^*(\mu_0,n_0)\times\text{Beta}^*(\mu_e,n_e), \qquad \theta_1 = (1-\eta_e)\theta_0 \label{eq:m-'}\\
M_{+}^{'}:~&(\theta_0,\eta_s) \sim \text{Beta}^*(\mu_0,n_0)\times\text{Beta}^*(\mu_s,n_s), \qquad \theta_1 = \theta_0 + \eta_s(1-\theta_0),\label{eq:m+'}
\end{align}
with marginal likelihoods given by
\begin{align*}
L_{-}^{'}(\mathcal{D}) 
&= 
\binom{N_0}{y_0}\binom{N_1}{y_1}\sum_{k=0}^{N_1-y_1}\binom{N_1-y_1}{k} \\
&\qquad\times \frac{\text{B}(y_0+y_1+k+\mu_0n_0,N-(y_0+y_1+k)+(1-\mu_0)n_0)}{\text{B}(\mu_0n_0,(1-\mu_0)n_0)} \\
&\qquad\times \frac{\text{B}(k+\mu_en_e,y_1+(1-\mu_e)n_e)}{\text{B}(\mu_en_e,(1-\mu_e)n_e)},
\end{align*}
and
\begin{align*}
L_{+}^{'}(\mathcal{D})
&= 
\binom{N_0}{y_0}\binom{N_1}{y_1}
\sum_{j=0}^{y_1}\binom{y_1}{j} \\
&\qquad\times \frac{\text{B}(y_0+j+\mu_0n_0,N-(y_0+j)+(1-\mu_0)n_0)}{\text{B}(\mu_0n_0,(1-\mu_0)n_0)} \\
&\qquad\times \frac{\text{B}(y_1-j+\mu_sn_s,N_1-y_1+(1-\mu_s)n_s)}{\text{B}(\mu_sn_s,(1-\mu_s)n_s)}.
\end{align*}
Algorithms to sample exactly from the posterior under $M_{-}^{'}$ and $M_{+}^{'}$ are provided in Appendix \ref{sec:sampling-appendix}. Note that here we interpret the constraint $\eta_s=0$ (or $\eta_e=0$) simply as a causally principled way to derive a prior compatible with the desired constraint $H_-:\theta_1 < \theta_0$ (or $H_+: \theta_1 > \theta_0$), and not as testing the former constraint in lieu of the latter.\footnote{In general, the data cannot  differentiate the stronger constraint, such as  $\eta_s=0$ (no one is hurt by the treatment), from the weaker constraint $\theta_1 < \theta_0$ (the treatment is beneficial on average), since the likelihood depends only on $\theta_1$ and $\theta_0$. Thus, in this case,  differences in using $M_{-}$ or $M_{-}^{'}$ amount  to  differences only in the induced priors satisfying the same testable constraint $\theta_1 <\theta_0$, such as one placing more (or less) mass on smaller (or larger) effects than the other.
} One interesting characteristic of models $M_{-}^{'}$ and $M_{+}^{'}$ is that they do not put $\theta_0$ and $\theta_1$ on equal footing, even when choosing beta priors compatible with the $\text{BREASE}(1/2, \mu, \mu; 2, 1, 1)$ distribution, which places flat marginals on $\theta_0$ and $\theta_1$. This is usually desirable, e.g., when the control condition indeed denotes a well understood baseline, such as a standard of care. 

Symmetry of  $\theta_0$ and $\theta_1$, however, can also be easily restored by switching the roles of the ``treatment'' and ``control'' conditions,  as we now discuss. 

\begin{figure}[t]
    \centering
    \includegraphics[width=\textwidth]{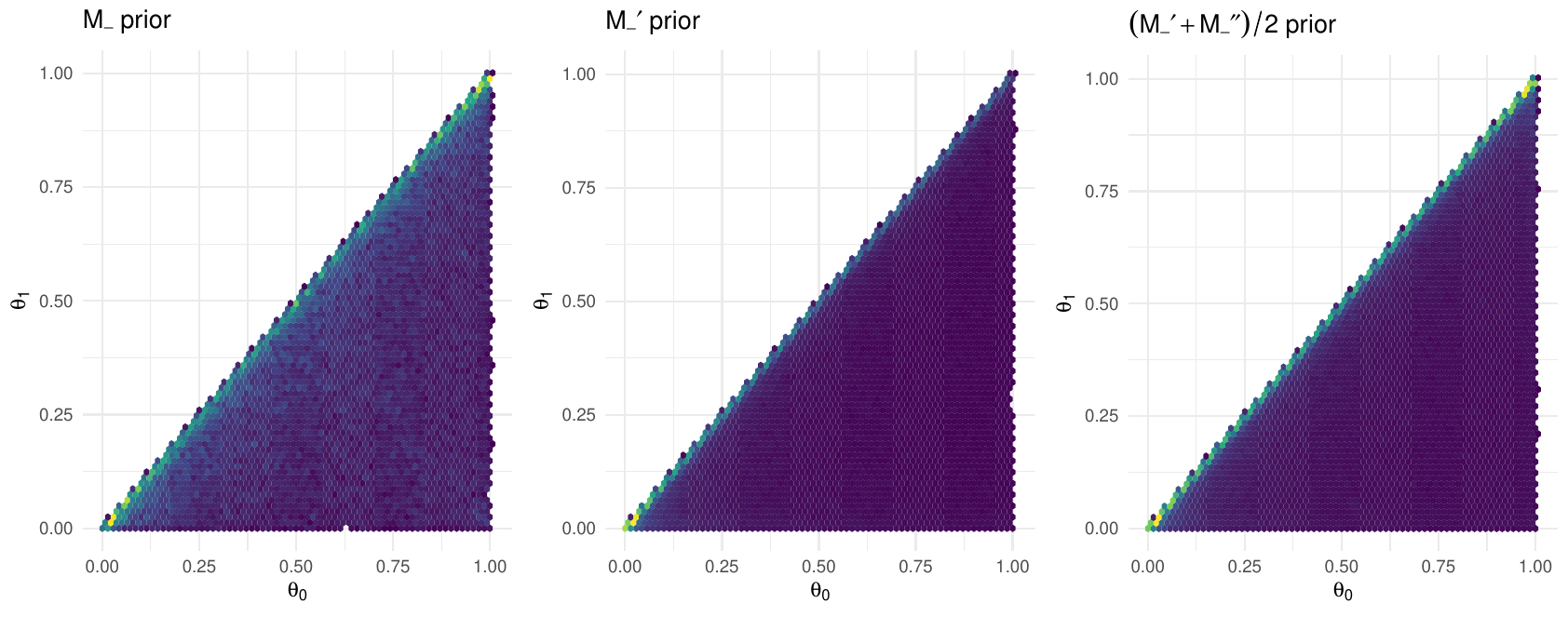}
    \caption{Left: heatmap of joint prior on $(\theta_0,\theta_1)$ implied by the $M_-$ prior (\ref{eq:m-prior}) with $\mu_0=1/2,\mu_e=\mu_s=0.3,n_0=2,n_e=n_s=1$. Center: prior on $(\theta_0,\theta_1)$ under $M_{-}^{'}$ with the same values of $(\mu_0,\mu_e,n_0,n_e)$. Right: prior on $(\theta_0,\theta_1)$ under the mixture model $(M_{-}^{'}+M_{-}^{''})/2$ with $\mu_1=1/2,\mu_s'=0.3,n_1=2,n_s'=1$ and the same values of $(\mu_0,\mu_e,n_0,n_e)$.}
    \label{fig:asym}
\end{figure}

Returning to the model $M_{-}^{'}$ $(\ref{eq:m-'})$, 
some natural values for the prior hyperparameters are 
\[
\mu_0=\mu_e=1/2,\qquad n_0=n_e=2,
\]
which define a flat $\text{Uniform}(0,1)^2$ prior on $(\theta_0,\eta_e)$. The resulting conditional prior on $\theta_1$ is
\[
\theta_1|\theta_0 \sim \text{Uniform}(0,\theta_0),
\]
which presents an intuitive representation of the hypothesis $H_-:\theta_1<\theta_0$. 
Note, however, that this specification of the model handles $\theta_0$ as the baseline quantity. 
We can also go in the other direction, specifying priors on $\theta_1$ and the ``side effects of placebo'' $\eta_s'$ and defining
\[
\theta_0 = \theta_1 + (1-\theta_1)\eta_s',
\]
which also instantiates $H_-:\theta_1<\theta_0$. 
We denote by 
$M_{-}^{''}$ the model
\begin{align*}
(\theta_1,\eta_s') &\sim \text{Beta}^*(\mu_1,n_1)\times\text{Beta}^*(\mu_s',n_s'), \\
\theta_0 &= \theta_1 + (1-\theta_1)\eta_s'.
\end{align*}
This asymmetry in our handling of $\theta_0$ and $\theta_1$ is reflected in the joint priors of $(\theta_0,\theta_1)$ under 
$M_{-}^{'}$ and $M_{-}^{''}$. 
As the central panel of Figure \ref{fig:asym} exhibits, the $M_{-}^{'}$ joint prior tends to favor small proportions (whereas $M_{-}^{''}$, not plotted, favors large proportions). On the other hand, sampling $(\theta_0,\eta_e,\eta_s)$ from the BREASE prior truncated to the set $\{(\theta_0,\eta_e,\eta_s):\theta_1<\theta_0\}$ (i.e., the $M_-$ prior (\ref{eq:m-prior})) yields a symmetric joint density on $(\theta_0,\theta_1)$ (left panel of Figure \ref{fig:asym}).
To assuage this asymmetry, we can put $\theta_0$ and $\theta_1$ on equal footing when testing the one-sided hypothesis $H_-$ (and, similarly, $H_+$) by using a prior that averages those under $M_{-}^{'}$ and $M_{-}^{''}$, as in the right panel of Figure \ref{fig:asym}.
In practice, we can decompose $H_-$ into the submodels $M_{-}^{'}$ and $M_{-}^{''}$ and report the marginal likelihood of $H_-$ as the average of the submodel marginal likelihoods. As the marginal likelihood under $M_{-}^{''}$ is also available analytically, this procedure comes with negligible added computational cost.

\subsubsection{An empirical Bayes prior}
\label{sec:gd-eb}

\begin{figure}[t]
    \centering
    \includegraphics[width=\textwidth]{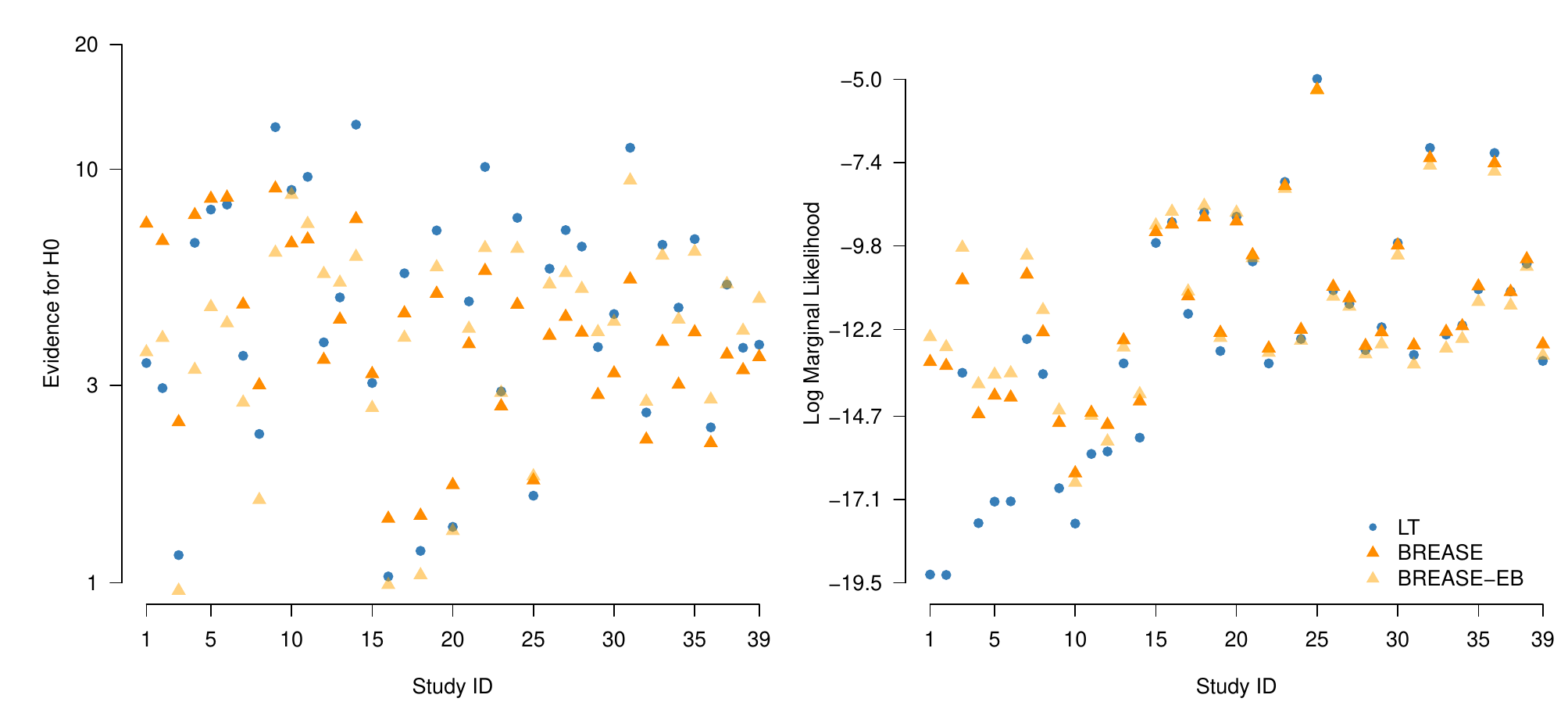}
    \caption{Comparison of Bayes factors ($\text{BF}_{01}$) and log marginal likelihoods under model $M_1$ (\ref{eq:m1}) of the default LT, default BREASE, and empirical Bayes BREASE priors across the 39 \textit{NEJM} studies.}
    \label{fig:brease-eb}
\end{figure}

As $\eta_e$ and $\eta_s$ are counterfactual probabilities, they are not generally point-identified from data. However, since $\theta_0$ and $\theta_1$ are identifiable, we can derive robust bounds on their range of possible values based on the observed data \citep{tian2000}. Equation~(\ref{eq:theta1}) implies the following algebraic constraints on $\eta_e$ and $\eta_s$:
\begin{align}
\max\left\{0,\frac{\theta_0-\theta_1}{\theta_0}\right\} &\le \eta_e \le \min\left\{1,\frac{1-\theta_1}{\theta_0}\right\}, 
\label{etae-region}
\\
\max\left\{0,\frac{\theta_1-\theta_0}{1-\theta_0}\right\} &\le \eta_s \le \min\left\{1,\frac{\theta_1}{1-\theta_0}\right\}.
\label{etas-region}
\end{align}
The inequalities (\ref{etae-region}) and (\ref{etas-region}) define the (marginal) partially identified regions of $\eta_e$ and $\eta_s$, respectively. Denote these intervals by $I_e(\theta_0,\theta_1)=[\ell_e(\theta_0,\theta_1),u_e(\theta_0,\theta_1)]$ and $I_s(\theta_0,\theta_1)=[\ell_s(\theta_0,\theta_1),u_s(\theta_0,\theta_1)]$. 
In the limit of infinite data, the posterior mass of $\eta_e$ and $\eta_s$ will concentrate within $I_e(\theta_0^*,\theta_1^*)$ and $I_s(\theta_0^*,\theta_1^*)$, respectively, assuming $\theta_0^*,\theta_1^*$ are the true values.

When conducting a Bayesian hypothesis test, a main concern is the sensitivity of Bayes factors to the prior. As demonstrated in Section \ref{sec:results}, a prior that places unreasonable assumptions on the treatment effects can lead to questionable conclusions. In this light, it may be desired to take a data-driven approach to prior specification that concentrates prior mass near the partially identified intervals of $\eta_e$ and $\eta_s$. 
For example, we can set the prior means $\mu_e$ and $\mu_s$ to equal their midpoints:
\begin{align*}
\hat\mu_e &=  \frac{1}{2}\left(\ell_e(\hat\theta_0,\hat\theta_1)+u_e(\hat\theta_0,\hat\theta_1)\right), \\   
\hat\mu_s &=  \frac{1}{2}\left(\ell_s(\hat\theta_0,\hat\theta_1)+u_s(\hat\theta_0,\hat\theta_1)\right), 
\end{align*}
where we use point estimates of the population proportions:
\begin{align*}
\hat\theta_0 &= \frac{y_0+1}{N_0+2}, \qquad \hat\theta_1 = \frac{y_1+1}{N_1+2}.
\end{align*}
As $\hat\theta_0$ shrinks the sample proportion toward $1/2$, it avoids division by zero in (\ref{etae-region}) and (\ref{etas-region}). 
Hence, we might consider priors of the form $\text{BREASE}(1/2,\hat\mu_e,\hat\mu_s;2,n,n)$ with $n\ge 0$. 
As this prior is estimated from the observed data, it can be thought of as an empirical Bayes approach \citep{robbins1992}. As such, we denote it by $\text{BREASE-EB}(n)$.

Note that when $n=1$ and $\hat\theta_0=\hat\theta_1=1/2$ (e.g., in the absence of data or when the sample proportions are $1/2$), we obtain a vague Jeffreys marginal prior $\text{Beta}(1/2,1/2)$ on $\eta_e$ and $\eta_s$. The choice of prior sample size $n=1$ yields something resembling a unit information prior \citep{kass1995unit}, wherein the prior mean is estimated from data and its spread is chosen so that the information content of the prior matches that of a single observation.

Figure \ref{fig:brease-eb} compares Bayes factors ($\text{BF}_{01}$) and log marginal likelihoods under model $M_1$ (\ref{eq:m1}) of the default $\text{LT}(0,0;1,1)$, $\text{BREASE}(1/2,0.3,0.3;2,1,1)$, and $\text{BREASE-EB}(1)$ priors across the 39 \textit{NEJM} studies reporting null results. The BREASE and BREASE-EB priors tend to provide the most equivocal Bayes factors on average, with mean $\text{BF}_{01}$ equal to 4.41, 4.42, and 5.38 for the BREASE-EB, BREASE, and LT priors, respectively. However, BREASE-EB Bayes factors tend to be closer to those of the LT approach than the default BREASE prior, with mean absolute percentage differences from the LT $\text{BF}_{01}$ of 19\% for the former and 32\% for the latter. 

Comparing log marginal likelihoods, which quantify the predictive performance of a model, we see that the BREASE-EB and default BREASE priors perform similarly, and generally better than the default LT prior, although the default BREASE performs slightly better overall. Indeed, the default BREASE log marginal likelihood exceeds the LT in 74\% of the studies compared to 59\% for the BREASE-EB prior. Furthermore, the default BREASE outperforms BREASE-EB in 62\% of the studies.

\subsection{Bayes factors with the IB and LT approaches}
\label{sec:testing-ib-lt}

Following \citet{dablander2022}, we calculate the Bayes factor $\text{BF}_{10}$ for the IB approach using the Savage-Dickey density ratio method applied to the difference of proportions $\eta=\theta_0-\theta_1$ \citep{wagenmakers2010}. A formula for the prior density of $\eta$ at the null $H_0:\eta=0$ can be found in Appendix A of \citep{dablander2022}. The Bayes factor using the $\text{IB}((a,a),(a,a))$ prior under $H_1$ as described in Section \ref{sec:ib} is then
\[
\text{BF}_{10} = \frac{\text{B}(2a-1,2a-1)\text{B}(a+y_0,a+N_0-y_0)\text{B}(a+y_1,a+N_1-y_1)}{\text{B}(2a+y_0+y_1-1,2a+N_0-y_0+N_1-y_1-1)\text{B}(a,a)^2}.
\]
Posterior estimates and credible intervals under $H_1$ are calculated using exact samples from the independent beta posterior.

Bayes factors $\text{BF}_{10}$ for the LT approach are calculating using the \textbf{abtest} package in R \citep{gronau2021}. The package uses a Laplace approximation to calculate $\text{BF}_{10}$, which is shown to have good performance. The LT prior under $H_1$ is as described in Section \ref{sec:lt}. Under $H_0:\psi=0$, the prior is $\beta\sim\text{Normal}(\mu_\beta,\sigma_\beta)$ with default values $(\mu_\beta,\sigma_\beta)=(0,1)$.  Posterior estimates and credible intervals under $H_1$ are calculated using posterior samples output by \textbf{abtest}. As \textbf{abtest} only reports marginal likelihoods up to a multiplicative constant, we used RJAGS \citep{rjags} to generate MCMC samples from the LT posterior and THAMES \citep{metodiev2023} to estimate the LT marginal likelihood for Figure \ref{fig:nejm-lml} using the samples.

\subsection{Likelihood derivation}
\label{app:covariates}

Suppose we observe i.i.d. samples $(X_i,Y_i,Z_i),~ i\in\{1,\ldots,N\}$, where, as before, $Y_i$ and $Z_i$ denote the binary outcome and treatment indicators for subject $i$ and $X_i$ is a discrete covariate with $k$ categories, $\{x_1, \dots, x_k\}$, such as sex or ethnicity (or both). 
We allow for the possibility of selection into treatment based on $X_i$. Hence, we now have only \emph{conditional} ignorability, $(Y_i(z)\indep Z_i) | X_i$, as opposed to a completely randomized treatment for which $(X_i,Y_i(z))\indep Z_i$. 
The likelihood for a single observation factorizes as 
\begin{align*}
\PP(X_i=x,Y_i=y,Z_i=z) &= \PP(Y_i=y|Z_i=z,X_i=x)\PP(Z_i=z|X_i=x)\PP(X_i=x) \\
&= \text{Bernoulli}(y;\theta_{z,x})\times\text{Bernoulli}(z;\delta_{x})\times p_{x},
\end{align*}
where   
we define 
the population proportion of $x$ as
\[
p_x = \PP(X_i=x),
\]
the propensity score
\[
\delta_x = \PP(Z_i=1|X_i=x),
\]
and the risk under treatment $z$ in stratum $x$ as
\begin{align*}
\theta_{z,x}
&= \PP(Y_i=1|X_i=x,Z_i=z)\\
&= \PP(Y_i(z)=1|X_i=x,Z_i=z) \tag{consistency}\\
&= \PP(Y_i(z)=1|X_i=x) \tag{$Y_i(z)\indep Z_i|X_i$}
\end{align*} 

Let $y_{z,x}$ denote the observed death count and  $N_{z,x}$  the corresponding sample size for each stratum $x \in \mathcal{X}$ and study arm $z \in \{0, 1\}$. Further define the total count for stratum $x$ as $N_{x} = N_{0,x} + N_{1,x}$ and the total population size $N = \sum_{x \in \mathcal{X}} N_{x}$. We use boldface to indicate vectors, $\bm{N}=\{N_{z,x}\}_{z\in\{0, 1\}, x\in \mathcal{X}}$, $\bm{N_X}=\{N_{x}\}_{x\in \mathcal{X}}$, and $\bm{y}=\{y_{z,x}\}_{z\in\{0, 1\}, x\in \mathcal{X}}$.  Finally, let $\mathcal{D}=(\bm{y}, \bm{N})$ denote the full data and 
$(\bm{\theta}, \bm{\delta},\bm{p_X})$ 
parameters, 
\[
\bm{\theta}=\{\theta_{z,x}\}_{z\in\{0, 1\}, x\in \mathcal{X}}, 
\quad \bm{\delta}=\{\delta_x\}_{x\in \mathcal{X}}, 
\quad \bm{p_X}=\{p_x\}_{x\in \mathcal{X}}.
\]
The full likelihood is then given by
\begin{align*}
L(\mathcal{D}|\bm{\theta},\bm{\delta},\bm{p_X}) &= \prod_{z,x}\text{Binomial}(y_{z,x};N_{z,x},\theta_{z,x}) \times 
\prod_x\text{Binomial}(N_{1,x};N_x,\delta_x)\\
&\qquad\times
\text{Multinomial}(\bm{N_X};\bm{p_X}).
\end{align*}
Note that, when we do not have covariates, and if we condition on the observed treatment vector, the likelihood reduces to
\begin{equation}
L(\mathcal{D}|\theta_0,\theta_1) = \binom{N_0}{y_0}\theta_0^{y_0}(1-\theta_0)^{N_0-y_0}\times \binom{N_1}{y_1}\theta_1^{y_1}(1-\theta_1)^{N_1-y_1},
\label{eq:marg-lik-2}
\end{equation}
which leads to the marginal parameterization without covariates as given in the main text. Finally, to obtain the likelihood with the BREASE parameterization, we define the efficacy and side-effects for every $x\in \mathcal{X}$, i.e., $\bm{\eta} = \{\eta_{e,x}, \eta_{s, x}\}_{x\in \mathcal{X}}$, and simply note that the following BREASE decomposition holds, $\theta_{1,x} = (1-\eta_{e,x})\theta_{0,x}+\eta_{s,x}(1-\theta_{0,x}),$ for every $x\in \mathcal{X}$.

\subsection{Extensions to non-compliance}
\label{app:iv}

Here we outline a natural extension of the BREASE framework to account for non-compliance. Non-compliance occurs when individuals do not adhere to their assigned treatment, a common issue in randomized experiments, creating a discrepancy between the assigned treatment and the actual treatment received. Notice this setup is equivalent to an instrumental variable design \citep{angrist1996identification}, where the treatment assignment serves as the instrument for the actual treatment received. Below we establish the foundations for this extension, by setting up the BREASE parameterization, BREASE prior and noticing that the posterior should also be a mixture of independent betas. A full development of this model deserves its own treatment, and thus we leave the study of induced priors, as well as the development of bespoke algorithms for future work. 

\paragraph{BREASE parameterization.} Let $Z_i$ denote the binary treatment assignment, $Y_i$ the binary outcome, and $D_i$ the binary indicator for \emph{actual treatment received}. We define potential outcomes $Y_i(d,z)$ as the outcome for individual $i$ if they were to receive treatment $d$ under assignment $z$. Similarly, we define $D_i(z)$ as the potential treatment uptake for individual $i$ if assigned to treatment $z$. A key assumption in this setting is the exclusion restriction, which posits that $Y_i(d,z) = Y_i(d)$. This implies that the treatment assignment $Z$ affects the outcome $Y$ only through its influence on the actual treatment received $D$. We maintain this assumption throughout our analysis.

Now let $C_i \in \{\text{AT}, \text{NT}, \text{CP}, \text{DF}\}$ denote the \emph{compliance type} of the individual, defined as:
\begin{align*}
\text{``always takers'' (AT):} & \quad \{D_i(0) = 1, D_i(1) = 1\}, \\
\text{``never takers'' (NT):} & \quad \{D_i(0) = 0, D_i(1) = 0\}, \\
\text{``compliers'' (CP):} & \quad \{D_i(0) = 0, D_i(1) = 1\},~\text{and,}  \\
\text{``defiers'' (DF):} & \quad \{D_i(0) = 1, D_i(1) = 0\}.
\end{align*}

Under the exclusion restriction assumption, the full joint distribution of compliance types and  response types is a probability table with $16$ entries. This distribution can be factorized as:
\begin{align*}
&\PP(Y_i(1)=y(1), Y_i(0)=y(0), D_i(1)=d(1), D_i(0)=d(0)) \\
&= \PP(Y_i(1)=y(1) \mid Y_i(0)=y(0), D_i(1)=d(1), D_i(0)=d(0) \\
&\quad \times \PP(Y_i(0)=y(1) \mid D_i(1)=d(1), D_i(0)=d(0)) \\
&\quad \times \PP(D_i(1) =d(1)\mid D_i(0) = d(0)) \PP(D_i(0)=d(0)).
\end{align*}
Now note this factorization is naturally amenable to the BREASE parameterization. The first two terms $\PP(Y_i(1) = y(1) \mid Y_i(0) = y(0), D_i(1) = d(1), D_i(0)=d(0))$ and $\PP(Y_i(0)=y(0) \mid D_i(1)=d(1), D_i(0)=d(0))$ can be parameterized in terms of the \emph{baseline risk, efficacy and side-effects} of the \emph{actual} treatment $D_i$ on the outcome $Y_i$, for each compliance type $C_i$:
\begin{align*}
\theta_{y0}^c &:= \PP(Y_i(0) = 1 \mid C_i = c), \\
\eta_{ye}^c   &:= \PP(Y_i(1) = 0 \mid Y_i(0) = 1, C_i = c), \\
\eta_{ys}^c   &:= \PP(Y_i(1) = 1 \mid Y_i(0) = 0, C_i = c),\\
\theta_{y1}^c &:= \PP(Y_i(1) = 1 \mid C_i = c) = (1 - \eta_{ye}^c) \theta_{y0}^c + \eta_{ys}^c (1 - \theta_{y0}^c).
\end{align*}

The last two terms of the factorization $\PP(D_i(1) \mid D_i(0))P(D_i(0))$ can be parameterized in terms of the \emph{baseline uptake, efficacy of assignment, and side-effects of assignment} on treatment uptake: 
\begin{align*}
\theta_{d0} &:= \PP(D_i(0) = 1), \\
\eta_{de}   &:= \PP(D_i(1) = 1 \mid D_i(0) = 0), \\
\eta_{ds}   &:= \PP(D_i(1) = 0 \mid D_i(0) = 1), \\
\theta_{d1} &:= \PP(D_i(1) = 1) = (1 - \eta_{ds}) \theta_{d0} + \eta_{de} (1 - \theta_{d0}).
\end{align*}
Here, since we would like the treatment assignment to result in uptake of the treatment, we reverse the definitions of efficacy and side-effects (e.g. $Z_i$ is efficacious if it turns $D_i(0)=0$ into $D_i(1)=1$).
Notice that a common assumption in the instrumental variable setting, also called ``monotonicity,'' can be easily expressed in this setting as $\eta_{ds} = 0$, that is, assignment to treatment does not make someone who would have gotten treatment to not take it. 

\paragraph{BREASE Prior.} Given the above parameterization, a natural prior for this model is to place a BREASE prior on the parameters of the treatment uptake, and either an independent or hierarchical BREASE prior on the parameters of the outcome model, for each compliance type. Notice that, similarly to the case of the BREASE prior with full compliance, the BREASE prior with non-compliance  should generalize  previously used priors such as that of \cite{chickering1996}, which places a traditional Dirichlet prior on $\PP(Y_i(1), Y_i(0), D_i(1), D_i(0))$ directly. 

\paragraph{Observed data likelihood.}

For one observation, conditional on $Z_i=z$,  we have
\begin{align*}
&\PP(Y_i = y, D_i = d \mid Z_i = z) = \PP(Y_i(d,z) = y, D_i(z) = d \mid Z_i = z) \tag{consistency}\\
&= \PP(Y_i(d) = y, D_i(z) = d) \tag{independence and exclusion} \\
&= \PP(Y_i(d) = y, D_i(z) = d, D_i(1-z)=1) + \PP(Y_i(d) = y, D_i(z) = d, D_i(1-z)=0) \tag{law of total prob.}\\
&= \PP(Y_i(d)=y \mid D_i(z)=d, D_i(1-z)=1)\PP(D_i(1-z)=1, D_i(z)=d) \\
&\qquad + \PP(Y_i(d)=y \mid D_i(z)=d, D_i(1-z)=0)\PP(D_i(1-z)=0, D_i(z)=d)\tag{chain rule}
\end{align*}
Thus, the observed data likelihood is given by the mixture of the probabilities of observing the outcome $Y_i=y$ under the two compliance types that are compatible with the data. For example, the likelihood for a unit with $Y_i=1$, $Z_i=1$ and $D_i=1$ is,
$$
\theta^{AT}_{y1}\times (1-\eta_{ds})\theta_{d0} + \theta^{CP}_{y1} \times \eta_{de}(1-\theta_{d0})
$$
namely, the weighted average of the probability of observing $Y_i=1$ for an always taker or a complier, the two compliance types compatible with $Z_i=1$ and $D_i=1$. Expanding $\theta^c_{y1}$ with the BREASE parameters for the outcome, we obtain,
$$
[(1 - \eta_{ye}^{AT}) \theta_{y0}^{{AT}} + \eta_{ys}^{{AT}} (1 - \theta_{y0}^{AT})]\times (1-\eta_{ds})\theta_{d0} + [(1 - \eta_{ye}^{CP}) \theta_{y0}^{{CP}} + \eta_{ys}^{{CP}} (1 - \theta_{y0}^{CP})] \times \eta_{de}(1-\theta_{d0}).
$$
Similar decomposition applies for the 8 possible values that $y, z, d$ can take. 
To obtain the likelihood for all observations, we simply multiply the independent contributions for each unit $i$. For brevity, we omit the resulting expression. 

\paragraph{Posterior.} As can be seen above, the likelihood is polynomial in the BREASE parameters for the outcome and compliance. Thus, under the independent BREASE prior,  the posterior will again decompose as a mixture of independent beta distributions over the full set of parameters. Therefore, similar to the case with full compliance, we expect the posterior to be mathematically tractable, admitting exact posterior sampling, an efficient data-augmented Gibbs algorithm, and analytical formulae for marginal likelihoods and Bayes factors. 

\subsection{Pathological sampling: diagnosing the issue}
\label{sec:pathological-supp}

We return to the pathological sampling example of Section \ref{sec:pathological}.
While JAGS sounds no alarm and its diagnostics provide no indication that the MCMC has failed to converge, Stan produces warning messages that the sampler is struggling. It reports a substantial proportion of divergent transitions, indicating that the chain has reached a portion of state space with extreme posterior curvature, which poses a challenge to gradient evaluation. Based on further investigation, the problem may be due to numerical issues in dealing with the Beta$^*(0.01,1)$ prior on $\eta_s$, which is highly concentrated and diverges at the boundary $\eta_s = 0$. (Indeed, in our experiments, Stan struggles to sample from the Beta$^*(0.01,1)$ distribution even in the absence of data and the other BREASE parameters.) There are a number of potential ways to ameliorate this issue. One is to express the posterior in its mixture form \eqref{eq:post}. In this representation we obtain a beta distribution more amenable to sampling by updating the prior on $\eta_s$ with the counterfactual counts $C_1, P_1$. When $N_1$ and $y_1$ are large, however, the double sum in \eqref{eq:post} can become prohibitively expensive to evaluate, which slows MCMC sampling significantly. Alternatively, we may circumvent dealing with $\eta_e$ and $\eta_s$ entirely by marginalizing one of them out of the prior and carrying out a change of variables to evaluate the conditional prior $\pi(\theta_1|\theta_0)$ using numerical integration as described in Section \ref{app:theta1-dist} of Supplement A. In either case, however, we find that the resulting sampler remains less efficient and more challenging to implement (even with existing software) than the exact sampler of Algorithm 1 and the data-augmented Gibbs sampler of Algorithm 2. 

\newpage

\subsection{Empirical examples: additional figures and tables}
\label{sec:pfizer-supp}

\begin{table}[h!]
\centering
\begin{tabular}{llll}
\toprule
\textbf{Study}  & \textbf{Sample Size} & \textbf{Independent} & \textbf{Hierarchical} \\
\midrule
HOT             & 18790                & 0.66 (0.5--0.9)               & 0.69 (0.55--0.86)             \\
TPT (Exc warfarin) & 2540             & 0.85 (0.62--1.06)             & 0.72 (0.59--0.9)              \\
PPP             & 4495                 & 0.82 (0.43--1.21)             & 0.74 (0.54--1.12)             \\
WHS             & 39876                & 1 (0.87--1.2)                 & 1.01 (0.81--1.23)             \\
BDS             & 5139                 & 1.01 (0.84--1.29)             & 0.94 (0.73--1.21)             \\
PHS             & 22071                & 0.59 (0.48--0.72)             & 0.64 (0.52--0.76)             \\
AAA             & 3350                 & 0.98 (0.74--1.2)              & 0.83 (0.65--1.09)             \\
POPADAD         & 1276                 & 1.03 (0.84--1.42)             & 0.82 (0.65--1.11)             \\
JPAD            & 2539                 & 0.98 (0.47--1.63)             & 0.81 (0.57--1.51)             \\
JPPP            & 14464                & 0.85 (0.5--1.19)              & 0.76 (0.57--1.15)             \\
ASCEND          & 15480                & 0.97 (0.82--1.07)             & 0.9 (0.77--1.06)              \\
ARRIVE          & 12546                & 0.93 (0.7--1.11)              & 0.81 (0.65--1.06)             \\
ASPREE          & 19114                & 0.98 (0.79--1.11)             & 0.9 (0.74--1.12)              \\
\bottomrule
\end{tabular}
    \caption{Posterior median and 95\% credible intervals for the risk ratio of low-dose aspirin on the reduction of  myocardial infarction using the independent and hierarchical BREASE priors.}
    \label{tab:aspirin-meta-analysis}
\end{table}

\begin{table}[htbp!]
    \centering
    \begin{tabular}{cccc}
   \toprule\textbf{Subgroup} &\textbf{Sample size} & \textbf{Independent}  & \textbf{Hierarchical}\\
    \toprule Age \\
     \hline 16 to 55 yr & 19,852  & 95.0 (89.4--98.1)  & 93.13 (87.34--96.81)  \\
     56 to 64 yr & 7,315 & 90.5 (72.2--97.9)  &  89.82 (77.11--96.16)  \\
     65 to 74 yr & 6,169 & 87.2 (48.6--98.2)  & 88.83 (66.86--96.35)  \\
     75+ yr & 1,559 & 81.8 (-10.8--99.4)  & 88.16 (45.27--96.57)  \\
     \hline  Race \\     
     \hline White & 29,174 & 94.7 (89.7--97.6)  & 92.83 (87.46--96.48)  \\
     Black & 2,988 & 88.6 (7.0--99.6)  & 87.49 (51.90--96.43)  \\
     Other & 2,760 & 79.5 (1.9--97.3)  & 85.32 (46.40--95.31)  \\
     \hline  Country \\     
     \hline Argentina & 5,066 & 95.2 (83.0--99.3)  & 91.50 (80.90--97.13)  \\
     Brazil & 2,250 & 75.5 (-2.4--96.9) & 85.26 (41.10--95.59)  \\
     United States & 26,865 & 94.3 (88.6--97.6) & 92.62 (86.69--96.37)  \\
     \bottomrule
    \end{tabular}
    \caption{Posterior median and 95\% credible intervals for the Pfizer-BioNTech COVID-19 vaccine efficacy, expressed as a percentage, stratified across subgroups using the independent and hierarchical BREASE priors.}
    \label{tab:covariates}
\end{table}

\end{document}